\documentclass{article}
\usepackage{amsmath,amsfonts,amsthm,amssymb,geometry,bbm}
\usepackage{graphics,color}
\usepackage{epsfig}
\usepackage{caption}
\usepackage{algorithmic}
\usepackage[ruled]{algorithm2e}
\usepackage{blkarray}
\usepackage[T1]{fontenc}
\usepackage{enumitem}
\usepackage[numbers]{natbib}

\newtheorem*{theorem*}{Theorem}

\newtheorem{definition}{Definition}

\newtheorem{lemma}{Lemma}
\newtheorem{proposition}{Proposition}
\newtheorem{theorem}{Theorem}
\newtheorem{remark}{Remark}

\newcommand{\Z}{{\mathbb Z}}
\newcommand{\R}{{\mathbb R}}

\newcommand{\D}{\mathcal{D}}

\newcommand{\varG}{{\bar G}}
\newcommand{\varV}{{\bar V}}
\newcommand{\varE}{{\bar E}}
\newcommand{\varS}{{\bar S}}
\newcommand{\vark}{{\bar k}}

\newcommand{\infmax}{{\sc InfMax}\xspace}
\newcommand{\SC}{{\sc SetCover}\xspace}
\newcommand{\VC}{{\sc VertexCover}\xspace}

\newcommand{\yes}{{\sf YES}\xspace}
\newcommand{\no}{{\sf NO}\xspace}

\DeclareMathOperator*{\E}{\mathbb{E}}

\DeclareMathOperator*{\argmax}{argmax}
\DeclareMathOperator*{\argmin}{argmin}
\DeclareMathOperator*{\sign}{sign}

\title{Beyond Worst-Case (In)approximability of Nonsubmodular Influence Maximization\thanks{The authors gratefully acknowledge the support of the National Science Foundation under Career Award 1452915 and AifT Award 1535912.}}
\author{
Grant Schoenebeck\\
University of Michigan\\
schoeneb@umich.edu
\and
Biaoshuai Tao\\
University of Michigan\\
bstao@umich.edu
}

\begin{document}
\date{}
\maketitle

\begin{abstract}
We consider the problem of maximizing the spread of influence in a social network by choosing a fixed number of initial seeds, formally referred to as the \emph{influence maximization problem}.
  It admits a $(1-1/e)$-factor approximation algorithm if the influence function is \emph{submodular}.
  Otherwise, in the worst case, the problem is NP-hard to approximate to within a factor of $N^{1-\varepsilon}$.
  This paper studies whether this worst-case hardness result can be circumvented by making assumptions about either the underlying network topology or the cascade model.
  All our assumptions are motivated by many real life social network cascades.

  First, we present strong inapproximability results for a very restricted class of networks called the \emph{(stochastic) hierarchical blockmodel}, a special case of the well-studied \emph{(stochastic) blockmodel} in which relationships between blocks admit a tree structure.
  We also provide a dynamic-programming-based polynomial time algorithm which optimally computes a directed variant of the influence maximization problem on hierarchical blockmodel networks.
  Our algorithm indicates that the inapproximability result is due to the bidirectionality of influence between agent-blocks.

  Second, we present strong inapproximability results for a class of influence functions that are ``almost'' submodular, called \emph{2-quasi-submodular}.
  Our inapproximability results hold even for any 2-quasi-submodular $f$ fixed in advance.
  This result also indicates that the ``threshold'' between submodularity and nonsubmodularity is sharp, regarding the approximability of influence maximization.
\end{abstract}

\section{Introduction}
\label{sect::intro}

A \emph{cascade} is a fundamental social network process in which a number of nodes, or agents, start with some property that they then may spread to neighbors.
The importance of network structure on  cascades has been shown to be relevant in a wide array of environments, including the adoption of products~\cite{Bass69,Brown87,GoldenbergLM01,MahajanMB90}, farming technology~\cite{ConleyU10}, medical practices~\cite{ColemanKM57}, participation in microfinancing~\cite{BanerjeeCDJ13}, and the spread of information over social networks~\cite{LermanG10}.

A natural question, known as \emph{the influence maximization problem} (\infmax), is how to place a limited number $k$ of initial seeds to maximize the spread of the resulting cascade~\cite{DomingosR01,RichardsonD02,KempeKT03,KempeKT05,MosselR10}.
To study influence maximization, we first need to understand how cascades spread.
Many cascade models have been proposed~\cite{Arthur89,Morris00,Watts02}, and two simple examples are the Independent Cascade model~\cite{KempeKT03,KempeKT05,MosselR10} and the Threshold model~\cite{Granovetter78}.
In the Independent Cascade model, each newly infected node infects each currently uninfected neighbor in the subsequent round with some fixed probability $p$.
In the Threshold model each node has a threshold (0, 1, 2, etc.) and becomes infected when the number of infected neighbors meets or surpasses that threshold.

In general, it is NP-hard even to approximate \infmax to within $N^{1-\epsilon}$ of the optimal expected number of infections~\cite{KempeKT03}.
However, assuming that we are using a particular class of cascades, called \emph{submodular} cascades, a straightforward greedy algorithm can efficiently find an answer that is at least a $(1- 1/e)$ fraction of the optimal answer.

In submodular cascade models, such as the Independent Cascade model, a vertex $v$'s marginal probability of becoming infected after a new neighbor $t$ is infected given $S$ is the set of $v$'s already infected neighbors is at least the marginal probability that $v$ is infected after $t$ is newly infected given $T\supseteq S$ is the set of $v$'s already infected neighbors~\cite{KempeKT03}.
Submodular cascade models are fairly well understood theoretically, and properties of these cascades are usually closely related to a network's degree distribution and conductance~\cite{Jackson08}.  Unfortunately, empirical research shows that many cascades are not submodular~\cite{RomeroMK11,BackstromHKL06,LeskovecAH06}.

Cascade models that violate the submodularity property are called \emph{nonsubmodular} cascades (or sometimes complex cascades).
In nonsubmodular contagion models, like the Threshold model, the marginal probability of being infected may increase as more neighbors are infected.
For example, if a vertex has a threshold of 2, then the first infected neighbor has zero marginal impact, but the second infected neighbor causes this vertex to become infected with probability 1.
Unlike submodular contagions, nonsubmodular contagions can require well-connected regions to spread~\cite{Centola10}.

Influence maximization becomes qualitatively different in nonsubmodular settings.
In the submodular case, seeds erode each other's effectiveness, and so should generally not be put too close together.
However, in the nonsubmodular case, it may be advantageous to place the initial adopters close together to create synergy and yield more adoptions.
The intuition that it is better to saturate one market first, and then expand implicitly assumes nonsubmodular influence in the cascades.

\textbf{Key Question: Can this worst-case inapproximability result of $N^{1-\epsilon}$  for nonsubmodular influence maximization be circumvented by making realistic assumptions about either the underlying network topology or the cascade model?}

We know a lot about what social networks look like, and previous hardness reductions make no attempt to capture realistic features of networks.
It is very plausible that by restricting the space of networks we might regain tractability.

In this paper, we consider two natural network topologies: the hierarchical blockmodel and the stochastic hierarchical blockmodel.
Each is a natural restriction on the classic \emph{(stochastic) blockmodel}~\cite{dimaggio1986structural,holland1983stochastic,WhiteBB76}  network structure.
In (stochastic) blockmodels, agents are partitioned into $\ell$ blocks.  The weight (or likelihood in the stochastic setting) of an edge between two vertices is based solely on blocks to which the vertices belong.
The weights (or probabilities) of edges between two blocks can be represented by an  $\ell \times \ell$ matrix.
In the (stochastic) hierarchical blockmodel, the structure of the $\ell \times \ell$ matrix is severely restricted to be ``tree-like''.\footnote{Previous work on community detection in networks~\cite{lyzinski2015community} defines a different, but related stochastic hierarchical blockmodel, where the hierarchy is restricted to two levels.}

Our (stochastic) hierarchical blockmodel describes the hierarchical structure of the communities, in which a community is divided into many sub-communities, and each sub-community is further divided, etc.
Typical examples include the structure of a country, which is divided into many provinces, and each province can be divided into cities.
Our model captures the natural observation that people in the same sub-community in the lower hierarchy tend to have tighter (or more numerous) bonds among each other~\cite{clauset2008hierarchical}.  Such a highly abstracted model necessarily fails to capture all features of social networks.  However, when we use this model as a lower bound, that is a strength as it shows that the problem is hard even in the case that communities structure can be represented by a tree.  Additionally, we feel that this is a very natural model which captures salient features of real-world networks, so our upper bounds in this model are still interesting.

We also consider restrictions on the cascade model.  The same research showing that cascades are often not submodular empirically also shows that the local submodularity often fails in one particular way---the second infected neighbor of an agent is, on average, more influential than the first.
When \citet{LeskovecAH06} studied the probability a person buys a book versus the number of incoming recommendations he receives, they observed a peak in the marginal probability of buying at $2$ incoming recommendations and then a slow drop.
While this work presents observational evidence, it suggests that if a person does not buy a book after
the first recommendation, but receives another, he is more likely to be persuaded by the second
recommendation. But thereafter, they are less likely to respond to additional recommendation

\citet{BackstromHKL06} made the same observation when they calculated the probability a person joins a community (e.g., LiveJournal and DBLP) as a function of the number $t$ of his friends already in the community.  
\citet{RomeroMK11} studied hashtag adoption in the Twitter network, and considered the fraction of users who adopt a hashtag after having $t$ neighbors' adoptions.  They coalesced their study's observations into a model where the marginal influence increases linearly from zero to two adopting neighbors and then linearly decreases thereafter.

These empirical studies motivate our study of the \emph{2-quasi-submodular} cascade model where the marginal effect of the second infected neighbor is greater than the first, but after that the marginal effect decreases.

\subsection{Our Results}
First, we present inapproximability results for \infmax in both the hierarchical blockmodel and the stochastic hierarchical blockmodel.
We show that \infmax is NP-hard to approximate within a factor of $N^{1-\varepsilon}$ for arbitrary $\varepsilon>0$.
Moreover, this result holds in the hierarchical blockmodel even if we assume all agents have unit threshold $\theta_v=1$.
We also extend this hardness result to the stochastic hierarchical blockmodel.

Moreover, for the hierarchical blockmodel, we present a dynamic-programming-based polynomial time algorithm for \infmax when we additionally assume the influence from one block to another is ``one-way''.
This provides insights to the above intractability result: the difficulty comes from the bidirectionality of influence between agent-blocks.

Secondly, we present an inapproximability result for the 2-quasi-submodular cascade model.
In particular, for \emph{any} 2-quasi-submodular influence function $f$, we show that it is NP-hard to approximate \infmax within a factor of $N^\tau$ when each agent has $f$ as its local influence function, where $\tau > 0$ is a constant depending on $f$.
This can be seen as a threshold result for approximability of \infmax, because if $f$ is submodular, then the problem can be approximated to within a $(1-1/e)$-factor, but if $f$ is just barely nonsubmodular the problem can no longer be approximated to within any constant factor.
We also show that, for any $\gamma\in(0,1)$, when only $N^\gamma$ agents have the fixed 2-quasi-submodular $f$ as their local influence functions and the remaining agents' local influence functions are submodular (or even identical to a fixed submodular function), \infmax is still NP-hard to approximate to within a factor of $N^\tau$, where $\tau>0$ is a constant depending on $f$ and $\gamma$.

Finally, we pose the open question of whether enforcing the aforementioned restrictions simultaneously on the network and the cascade renders the problem tractable.

\subsection{Related Work}
The influence maximization problem was posed by Domingos and Richardson~\cite{DomingosR01,RichardsonD02}.
Kempe, Kleinberg, and Tardos showed that a simple greedy algorithm obtains a $(1-1/e)$ factor approximation to the problem in the independent cascade model and linear threshold model~\cite{KempeKT03}, and extended this result to a family of submodular cascades which captures the prior results as a special case~\cite{KempeKT05}.
Mossel and Roch~\cite{MosselR10} further extended this result to capture all submodular cascades.

Perhaps most related to the present work, are several inapproximability results for \infmax.
If no assumption is made for the influence function, \infmax is NP-hard to approximate to within a factor of $N^{1 - \varepsilon}$ for any $\varepsilon > 0$~\cite{KempeKT03}.

Chen~\cite{chen2009approximability} found inapproximability results on a similar optimization problem: instead of maximizing the total number of infected vertices given $k$ initial targets, he considered the problem of finding a minimum-sized set of initial seeds such that all vertices will eventually be infected.
This work studied restrictions of this problem to various threshold models.

An important difference between our hardness result in Section~\ref{sect:hardness_2submod} and all the previous results is that our result holds for \emph{any} 2-quasi-submodular functions.
In particular, in this work, $f$ is fixed in advance before the NP-hardness reduction, while in previous work, specific influence functions were constructed within the reductions.

Several works looked at slightly different aspects of influence maximization.
Borgs,  Brautbar,  Chayes, and Lucier~\cite{borgs2012influence} provably showed fast running times when the influence function is the independent cascade model.
Lucier, Oren, and Singer~\cite{lucier2015influence} showed how to parallelize (in a model based on Map Reduce) the subproblem of determining the influence of a particular seed.
Seeman and Singer~\cite{seeman2013adaptive} studied the special case where only a subset of the nodes in the network are available to be infected.
They showed a constant factor approximation to the problem in their setting.
\citet{HeK2016} and \citet{chen2016robust} looked at a robust version of the problem where the exact parameters of the cascade are unknown.
Several works~\cite{Bharathi2007,GoyalK12} studied the problem as a game between two different infectors.

Following the work of Kempe, Kleinberg, and Tardos~\cite{KempeKT03,KempeKT05}, there were extensive works to solve \infmax based on the heuristic implementations of the greedy algorithm designed to be efficient and scalable~\cite{chen2009efficient,ChenYZ10,lucier2015influence}.

The notion of ``near submodularity'' was also proposed and studied in \cite{Horel16}.
Our definition differs from the one in \cite{Horel16} in that a 2-quasi-submodular function can be, intuitively, very far from being submodular (for example, the 2-threshold cascade model).
However, our reduction in Section~\ref{sect:hardness_2submod} works for all 2-quasi-submodular functions, and 2-quasi-submodular functions can be arbitrarily close to submodular functions.

Our algorithm in Section~\ref{sect::oneway} was further studied and generalized by Angell and Schoenebeck in~\cite{angell2016don}.
They showed that, empirically, this generalized algorithm works very well even for arbitrary graphs.
Specifically, they run dynamic programming on a hierarchical decomposition of general graphs, and, empirically, the algorithm effectively leverages the resultant hierarchical structures to return seed sets substantially superior to those of the greedy algorithm.


Similar to our inapproximability result for the 2-quasi-submodular cascade model in Sect.~\ref{sect:hardness_2submod} but independent to our work\footnote{Both \cite{li2017influence} and the WINE 2017 conference version of this paper were published in December, 2017.}, \citet{li2017influence} studies influence maximization with almost submodular local influence functions, and shows that, for any $\gamma,\varepsilon\in(0,1)$, \infmax is hard to approximate to within factor $1/N^{\frac{\gamma}{c}}$ even if the graph only contains $N^\gamma$ vertices that admit nonsubmodular local influence functions that are $\varepsilon$-almost submodular (and the remaining vertices admit submodular local influence functions), where $c=3+3/\log\frac2{2-\varepsilon}$.
When the number of vertices admitting $\varepsilon$-almost submodular local influence functions is a constant, \citet{li2017influence} provides a constant-factor approximation algorithm for \infmax.

Our result in Sect.~\ref{sect:hardness_2submod} can be seen as a generalization to the inapproximability result in Li et al.: their results construct a 2-quasi-submodular influence function $f$ (although $\varepsilon$ can be arbitrarily small and fixed in advance, making $f$ arbitrarily close to a submodular function); in contrast, our result holds for any $f$ that is fixed in advance and universal for all vertices.  In addition, our inapproximability result holds even for undirected graphs, while the graph constructed in the reduction in Li et al.\ is directed.

At a high level, the techniques of the two approaches are similar.  However, the gadgets used in our more general result require additional ideas.


In Appendix~\ref{sect:hardness_2submod_compare}, we show that our results seamlessly extend to the setting of Li et al.\ where only a sublinear fraction of vertices (e.g., $N^\gamma$) admit nonsubmodular local influence functions.

\section{Preliminaries}
\label{sect::prelim}
In general a \emph{cascade} on a graph is a stochastic mapping from a subset of vertices---the \emph{seed vertices}, to another set of vertices that always contain the seed vertices---the \emph{infected vertices}.
The cascades we study in this paper all belong to the general threshold model \cite{MosselR10}, which captures the local decision-making of vertices.

\begin{definition}\label{defi:GTM}
  The \emph{general threshold model} $I^{G}_{F,\D}$, is defined by a graph $G=(V,E)$ which may or may not be edge-weighted, and for each vertex $v$:
  \begin{enumerate}[noitemsep,nolistsep]
    \item[i.] a monotone local influence function $f_v:\{0, 1\}^{|\Gamma(v)|} \to\R_{\geq0}$ where $\Gamma(v)$ denotes the neighbor vertices of $v$ and $f_v(\emptyset) = 0$,  and
    \item[ii.]  a threshold distribution $\D_v$ whose support is $\R_{\geq0}$.  Let $F$  and $\mathcal{D}$ denote the collection of $f_v$ and $\mathcal{D}_v$ respectively.
  \end{enumerate}
  On input $S \subseteq V$,  $I^{G}_{F,\D}(S)$ outputs a set of vertices as follows:
  \begin{enumerate}[noitemsep,nolistsep]
    \item[1.] Initially only vertices in $S$ are infected, and for each vertex $v$ the threshold $\theta_v\sim\D_v$ is sampled from $\mathcal{D}_v$ independently.\footnote{The rationale of sampling thresholds \emph{after} the seeds' selection is to capture the scenario that the seed-picker does not have the full information on the agents in a social network, and this setting has been used in many other works \cite{KempeKT03,MosselR10}.}
    \item[2.]  In each subsequent round, a vertex $v$ becomes infected if the influence of its infected neighbors exceeds its threshold.
    \item[3.] The set of infected vertices is the output (after a round where no additional vertices are infected).
  \end{enumerate}
\end{definition}

We use $k$ to denote $|S|$, the number of seeds, and use $N$ to denote $|V|$, the total number of vertices in $G$.
Let
$$\sigma^{G}_{F,\D}(S) = \E\left[\left|I^{G}_{F,\D}(S)\right|\right]$$
be the \emph{expected} total number of infected vertices due to the influence of $S$, where the expectation is taken over the samplings of the thresholds of all vertices.
We refer to $\sigma^{G}_{F,\mathcal{D}}(\cdot)$ as the \emph{global influence function}.
Sometimes we write $\sigma(\cdot)$ with the parameters $G,F,\mathcal{D}$ omitted, when there is no confusion.
Because each $f_v$ is monotone, it is straightforward to see that $\sigma$ is monotone.

\begin{definition}\label{defi:infmax}
  The \infmax problem is an optimization problem which takes as inputs $G=(V,E)$, $F$, $\D$, and an integer $k$, and outputs $\argmax_{S \subseteq V: |S| = k} \sigma^{G}_{F,\mathcal{D}}(S)$, a seed set of size $k$ that maximizes the global influence.
\end{definition}

In this paper, we consider several special cases of the general threshold model $I^G_{F,\D}$ by making assumptions on the network topology $G$, or the cascade model\footnote{The phrase ``cascade model'' here, as well as in the abstract and Section~\ref{sect::intro}, refers to the description how each vertex is influenced by its neighbors, which is completely characterized by $F$ and $\D$ in the general threshold model.} $F,\D$.

\subsection{Assumptions on Graph $G$}
We consider two graph models---the \emph{hierarchical blockmodel} and the \emph{stochastic hierarchical blockmodel}, which are the special case of the well studied \emph{blockmodel}~\cite{WhiteBB76} and \emph{stochastic blockmodel}~\cite{holland1983stochastic} respectively.

\paragraph{The Hierarchical Blockmodel}\label{sect::Definition_HBGM}
\begin{definition}\label{defi:HBG}
  A \emph{hierarchical blockmodel} is an undirected \emph{edge-weighted} graph $G=(V,T)$, where $V$ is the set of all vertices of the graph $G$, and  $T=(V_T,E_T,w_T)$ is a node-weighted binary tree $T$ called a \emph{hierarchy tree}.
  In addition, $w_T$ satisfies $w_T(t_1)\leq w_T(t_2)$ for any $t_1,t_2\in V_T$ such that $t_1$ is an ancestor of $t_2$.\footnote{Since, as it will be seen later, each node in the hierarchy tree represents a community and its children represent its sub-communities, naturally, the relation between two persons is stronger if they are in a same sub-community in a lower level.}
  Each leaf node $t\in V_T$ corresponds to a subset of vertices $V(t)\subseteq V$, and the $V(t)$ sets partition the vertices of $V$.
  In general, if $t$ is not a leaf, we denote $V(t)=\cup_{t':\text{ a leaf, and an offspring of }t}V(t')$.

  For $u, v \in V$, the weight of the edge $(u, v)$ in $G$ is just the weight of the least common ancestor of $u$ and $v$ in $T$.
  That is $w(u, v) = \max_{t: u, v \in V(t)} w(t)$.
  If this weight is 0, then we say that the edge does not exist.
\end{definition}

To avoid possible confusion, we use the words \emph{node} and \emph{vertex} to refer to the vertices in $T$ and $G$ respectively.

Figure~\ref{fig:hierarchicalblockgraphexample} provides an example of how a hierarchy tree defines the weights of edges in the corresponding graph.

\begin{figure}
    \centerline{\includegraphics[width=\textwidth]{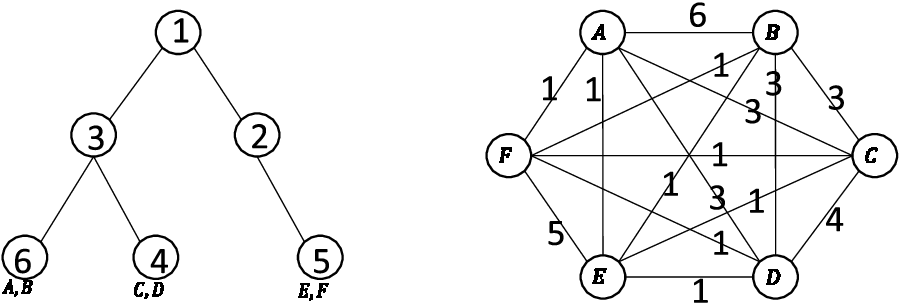}}
    \caption{An example of a hierarchy tree with its corresponding graph.  The number on each node of the hierarchy tree on the left-hand side indicates the weight of the node, which reflects the weight of the corresponding edges on the hierarchical block graph on the right-hand side in the above-mentioned way.}
    \label{fig:hierarchicalblockgraphexample}
\end{figure}

Additionally, we can assume without loss of generality that the hierarchy tree is a \emph{full} binary tree, as a node in $T$ having only one child plays no role at deciding the weights of edges in $G$.
For example, in Figure~\ref{fig:hierarchicalblockgraphexample}, the node having weight $2$ does not affect the weight configuration on the right-hand side.
We can delete this node and promote the node with weight $5$ to be a child of the root node.
We will keep the full binary tree assumption from now on.

\paragraph{The Stochastic Hierarchical Blockmodel}\label{sect::Definition_SHBGM}
The \emph{stochastic hierarchical blockmodel} is similar to the hierarchical blockmodel defined in the last section, in the sense that the structure of the graph is determined by a hierarchy tree.
Instead of assigning weights to different edges measuring the strength of relationships, here we assign a probability with which the edge between each pair of vertices appears.
Technically speaking, a stochastic hierarchical blockmodel is a distribution of unweighted undirected graphs, where each edge is sampled with a certain probability.

\begin{definition}\label{defi:SHBG}
  A \emph{stochastic hierarchical blockmodel} is a distribution $\mathcal{G}=(V,T)$ of unweighted undirected graphs where $V, T$ are the same as they are in Definition~\ref{defi:HBG} with the additional restriction that the node weights in $T$ belong to the interval $[0, 1]$.
  Let $H$ be the weighted graph defined by the hierarchical blockmodel $H=(V,T)$, and let $w(e)$ denote the weight of edge $e$ in $H$.
  Then $G=(V,E)$ is sampled by independently including each edge $e$ with probability $w(e)$.
\end{definition}

When it comes to the choices of $S$, the \infmax problem can be defined in two different ways, regarding whether we allow the seed-picker to see the sampling $G\sim\mathcal{G}$ \emph{before} choosing the seed set $S$.

\begin{definition}\label{defi:infmax_pre}
  \emph{Pre-sampling stochastic hierarchical blockmodel \infmax} is an optimization problem which takes as inputs $\mathcal{G}$, $F$, $\D$, an integer $k$ and outputs $$\argmax_{S \subseteq V: |S| = k} \E_{G \sim \mathcal{G}} \left[\sigma^{G}_{F,\mathcal{D}}(S)\right],$$
  a seed set of size $k$ that maximizes the expected global influence.
\end{definition}

\begin{definition}\label{defi:infmax_post}
  \emph{Post-sampling stochastic hierarchical blockmodel \infmax} is an average case version of \infmax which takes as input $\mathcal{G}$, $F$, $\D$, and an integer $k$, and outputs the solution of the \infmax instance $(G, F, \D, k)$ after sampling $G$ from $\mathcal{G}$.
\end{definition}

\subsection{Assumptions on Cascade Model $F,\D$}
We consider several generalizations of the well-studied \emph{linear threshold model}~\cite{KempeKT03}.
The linear threshold model is a special case of the general threshold model $I^G_{F,\D}$, with each $f_v$ being linear (see Definition~\ref{defi:linearLIF} below), and each $\D_v$ being the uniform distribution on $[0,1]$.

The cascade model in Definition~\ref{defi:linearLIF} generalizes the linear threshold model by removing the assumption on $\D_v$.
The \emph{universal local influence model} defined in Definition~\ref{defi:ulim}, generalizes the linear threshold model by allowing non-linear $f_v$, while it restricts our attention to unweighted graphs.
We also consider a special case where $f_v$ is \emph{2-quasi-submodular} in the last subsection.

\paragraph{Linear and Counting Local Influence Functions}
A natural selection of local influence function $f_v$ is the linear function, by which the influences from $v$'s neighbors are additive.
\begin{definition}\label{defi:linearLIF}
  Given a general threshold model $I^G_{F,\D}$ with a weighted graph $G$, we say that $F$ is \emph{linear} if for each $v\in V$ we have $f_v(S_v)=\sum_{u\in S_v}w(u,v)$ for any $S_v\subseteq\Gamma(v)$.
\end{definition}

For a general threshold model $I^G_{F,\D}$ with linear $F$, if we additionally assume each $\D_v$ is the uniform distribution on $[0,1]$, then this becomes the linear threshold model.

Definition~\ref{defi:linearLIF} defines a cascade model for weighted graphs.
We have the following definition which is the unweighted counterpart to Definition~\ref{defi:linearLIF}.

\begin{definition}\label{defi:countingLIF}
  Given a general threshold model $I^G_{F,\D}$ with an unweighted graph $G$, we say that $F$ is \emph{counting} if for each $v\in V$ we have $f_v(S_v)=|S_v|$ for any $S_v\subseteq\Gamma(v)$.
\end{definition}

\paragraph{Universal Local Influence Functions}
We say $f_v$ is \emph{symmetric} if $f_v(S_v)$ only depends on the \emph{number} of $v$'s infected neighbors $|S_v|$ so that each of $v$'s infected neighbors is of equal importance.
In this case, $f_v$ can be viewed as a function $f_v:\Z_{\geq0}\to\R_{\geq0}$ which takes an integer as input, rather than a set of vertices.
Thus, $f_v$ can be encoded by an increasing sequence of positive real numbers $a_0, a_1, a_2, \ldots$ so that $f_v(i)=a_i$.
Note that $f_v(0)= a_0 = 0$, as we have assumed $f_v(\emptyset)=0$.

For instance, the local influence function $f_v$ defined in Definition~\ref{defi:countingLIF} is symmetric, with $a_i=i$.
In contrast, $f_v$ in Definition~\ref{defi:linearLIF} is not symmetric, as the neighbors connected by heavier edges contribute more to $f_v(S_v)$.

\begin{definition}\label{defi:ulim}
  Given an increasing function $f:\Z_{\geq0}\to[0,1]$, the \emph{universal local influence model} $I_f^G$ is a special case of the general threshold model $I^G_{F,\D}$, such that for each $v\in V$ we have that
  \begin{itemize}
    \item $f_v$ is symmetric, and $f_v=f$ (such that all $f_v$'s are identical), and
    \item $\D_v$ is the uniform distribution on $[0,1]$.
  \end{itemize}
\end{definition}

Notice that we can assume without loss of generality that $G$ is unweighted in Definition~\ref{defi:ulim}, as each $f_v$ is fixed to be some increasing function $f$ which does not depend on the weights of edges.


As a final remark, for any general threshold model $I^G_{F,\D}$ with each $D_v$ being the uniform distribution on $[0,1]$, we can intuitively view $f_v(S_v)$ as the \emph{probability} that $v$ will be infected (where we take $f_v(S_v)>1$ as probability $1$).
In the universal local influence model, $a_i$ can be viewed as the probability that a vertex will be infected, given that it has $i$ infected neighbors.

\paragraph{Submodular and 2-Quasi-Submodular Functions} \label{sect::Definition_2submodular}
Let $g: \{0,1\}^{S} \to \R$ be a function which takes as input a subset of a set $S$.
Formally, $g$ is \emph{submodular} if $g(A\cup\{u\})-g(A)\geq g(B\cup\{u\})-g(B)$ for any $u \in S\setminus B$ and sets $A\subseteq B \subseteq S$.
Intuitively, this means that the marginal effect of each element decreases as the set increases.

The definition above can be applied to each local influence functions $f_v:\{0,1\}^{|\Gamma(v)|}\to\R_{\geq0}$, as well as the global influence function $\sigma_{F,\D}^G:\{0,1\}^{|V|}\to\R_{\geq0}$.
Given $G, F, \D$ we say that a general threshold model $I^G_{F,\D}(\cdot)$ is submodular if $\sigma^G_{F,\D}(\cdot)$ is.
In~\cite{MosselR10}, it has been shown that the local monotonicity and submodularity of all $f_v$'s implies the global monotonicity and submodularity of $I^G_{F,\D}(\cdot)$ for all $G$ when $\D_v$ is the uniform distribution on $[0, 1]$.

We are particularly concerned with the universal local influence model in Definition~\ref{defi:ulim}.
Here $f$ is submodular if the marginal gain of $f$ by having one more infected neighbor is non-increasing as the number of infected neighbors increases.
Formally, for $i_1<i_2$, we have
$$f(i_1+1) - f(i_1) \geq f(i_2+1) - f(i_2).$$
Intuitively, $f$ is submodular if its domain can be smoothly extended to $\R_{\geq0}$ to make $f$ concave.

We will consider \emph{2-quasi-submodular} local influence functions $f$, which is ``almost'' submodular such that the submodularity is only violated for the first two inputs of $f$.
In particular, we fail to have the submodular constraint $f(1)-f(0)\geq f(2)-f(1)$, and instead we have $f(1)-f(0)<f(2)-f(1)$, which is just $f(2)>2f(1)$ as $f(0)=0$.

\begin{definition}\label{defi:2quasisubmodular}
  $f:\Z_{\geq0}\to[0,1]$ is \emph{2-quasi-submodular} if $f(2)>2f(1)$ and $f(i)-f(i-1)$ is non-increasing in $i$ for $i\geq2$.
\end{definition}

In general, for any non-zero submodular function $f$, if we sufficiently decrease $f(1)$, $f$ becomes 2-quasi-submodular.
Thus, from any non-zero submodular function, we can obtain a 2-quasi-submodular function.

We note that the 2-threshold cascade model, where each vertex will be infected if it has at least $2$ infected neighbors, can be viewed as the universal local influence model with a 2-quasi-submodular $f$ (with $f(0)=f(1)=0$ and $f(i)=1$ for $i\geq2$, keeping the assumption that $\theta_v$ is drawn uniformly at random from $[0,1]$).

\section{Hierarchical Blockmodel Influence Maximization}
\label{sect::resultHBGM}
In this section, we provide a strong inapproximability result for \infmax problem for the hierarchical blockmodel cascade even when all vertices have a deterministic threshold $1$.
Specifically, we will show that it is NP-hard to approximate optimal $\sigma(S)$ within a factor of $N^{1-\varepsilon}$ for any $\varepsilon>0$ (recall that $N=|V|$ is the total number of vertices in the graph).
The same inapproximability result holds for the most general case where $\D$ is given as input to \infmax.


\begin{theorem}\label{hardnessHBGM}
  Consider the \infmax problem $(G,F,\D,k)$.
  For any constant $\varepsilon>0$, even if $G$ is a hierarchical blockmodel, $F$ is linear (see Definition~\ref{defi:linearLIF}), and $\D_v$ is the point-mass distribution with $\Pr_{\theta_v\sim\D_v}(\theta_v=1)=1$ for each $v\in V$, it is NP-hard to distinguish between the following two cases:
  \begin{itemize}
      \item \yes: there exists a seed set $S$ with $|S|=k$ such that $\sigma_{F,\D}^G(S)=\Theta(N)$;
      \item \no: for any seed set $S$ with $|S|=k$, we have $\sigma_{F,\D}^G(S)=O(N^\varepsilon)$.
  \end{itemize}
\end{theorem}

We will prove this by a reduction from the \VC problem, a well-known NP-complete problem.
\begin{definition}
  Given an undirected graph $\varG=(\varV,\varE)$ and a positive integer $\vark$, the \VC problem $(\varG,\vark)$ asks if we can choose a subset of vertices $\varS\subseteq\varV$ such that $|\varS|=\vark$ and such that each edge is incident to at least one vertex in $\varS$.
\end{definition}

\paragraph{The Reduction}
  Given a \VC instance $(\varG,\vark)$, let $n=|\varV|$ and $m=|\varE|$.
  We use $A_1,\ldots,A_n$ to denote the $n$ vertices and $e_1,\ldots,e_m$ to denote the $m$ edges.\footnote{We use the letter $A$ to denote the vertices in a \VC instance instead of commonly used $v$, while $v$ is used for the vertices in an \infmax instance. Since \VC can be viewed as a special case of \SC with vertices corresponding to subsets and edges corresponding to elements, the letter $A$, commonly used for subsets, is used here.}
  We make the assumptions $n>\vark$ is an integer power of $2$ and $m>n+\vark$.\footnote{For the assumption that $n$ is an integer power of $2$, we can just add isolated vertices to $\varG$. For the assumption $m>n+\vark$, notice that allowing the graph $\varG$ to be a multi-graph does not change the nature of \VC, we can ensure $m$ to be sufficiently large by just duplicating edges.}
  Let $W=nm$, $M=(n(2W+m)-1)^\frac{1}{\varepsilon}$, and $\delta>0$ be a sufficiently small real number.

  We will construct the graph $G=(V,E,w)$ by constructing a hierarchy tree $T$ which uniquely determines $G$ (see Definition~\ref{defi:HBG} in Section~\ref{sect::Definition_HBGM}).
  The construction of $T$ is shown in Figure~\ref{fig:reductionHBGM}.
  The first $\log_2n$ levels of $T$ is a full balanced binary subtree with $n$ leaves, and the weight of the nodes in all these levels is $\delta$.
  Each of those $n$ leaves is the root of a subtree corresponding to each vertex $A_i$ in the \VC instance.

  \begin{figure}
    \centerline{\includegraphics[width=\textwidth]{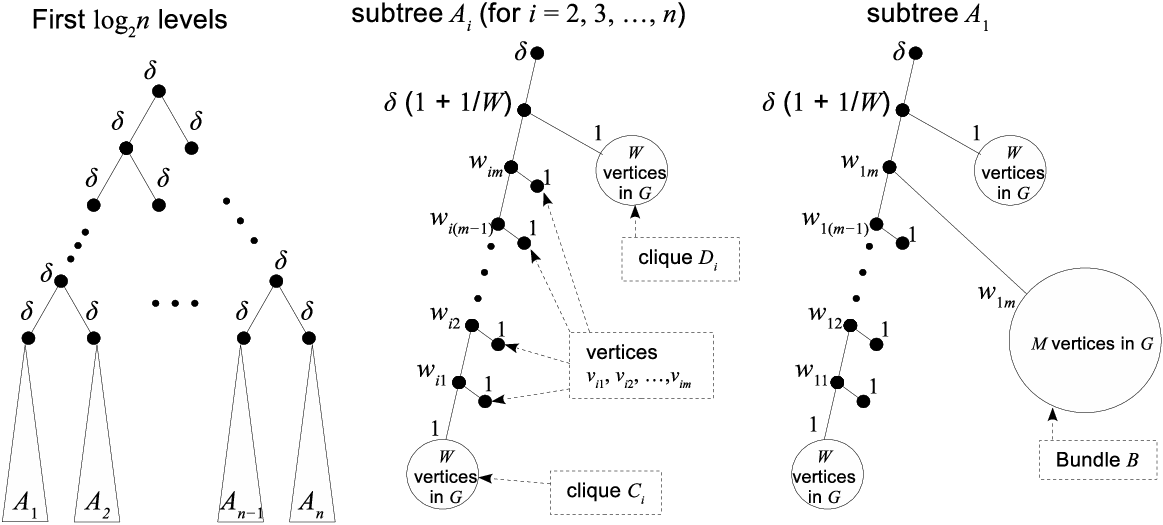}}
    \caption{The construction of the hierarchy tree $T$.}
    \label{fig:reductionHBGM}
  \end{figure}

  The structure of the subtrees corresponding to $A_2,\ldots,A_n$ and $A_1$ are shown on the right-hand side of Figure~\ref{fig:reductionHBGM}.
  The numbers on the tree nodes indicate the weights, and in particular
  \begin{equation}\label{eqn:wij}
  w_{ij}=\left\{\begin{array}{ll}
    \frac{[1-(n+\vark-1)W\delta-(n-1)(j-1)\delta-2\delta]+\delta}{W-1+j} & \mbox{if edge }e_j\mbox{ is incident to }A_i\\
    \frac{1-(n+\vark-1)W\delta-(n-1)(j-1)\delta-2\delta}{W-1+j} & \mbox{otherwise}\\
  \end{array}\right.,
  \end{equation}
  for each $i=1,\ldots,n$ and $j=1,\ldots,m$.

  The leaves of each subtree $A_i$ are the leaves of $T$, which, as we recall from Definition~\ref{defi:HBG} correspond to subsets of vertices in $G=(V,E,w)$.
  Among all the leaves shown on the right-hand side of Figure~\ref{fig:reductionHBGM}, each solid dot corresponds to a subset of $V$ containing only one vertex, and each hollow circle corresponds to a subset of $V$ containing many vertices with the corresponding number of vertices shown.

  For each subtree $A_i$ with $i=2,\ldots,n$, we have constructed $m+2$ leaves corresponding to $2W+m$ vertices in $G$.
  They are, in up-to-down order, a clique $D_i$ of $W$ vertices, vertices $v_{im},v_{i(m-1)},\ldots,v_{i1}$, and a clique $C_i$ of $W$ vertices.
  As each vertex has threshold $1$ and the leaf nodes corresponding to $C_i,D_i$ both have weight $1$, infecting any vertex in $C_i$ or $D_i$ will cause the infection of all $W$ vertices (which justifies the name ``clique'').

  The construction of $A_1$ is similar.
  The only difference is that, instead of connecting to a node corresponding to the vertex $v_{1m}$, the node with weight $w_{1m}$ is now connected to another node with the same weight and corresponding to a bundle $B$ in $G$ with $M$ vertices.
  We shall not call this large bundle $B$ a ``clique'', as the weight of the edge between each pair of these $M$ vertices is $w_{1m}\ll1$, which is much weaker.

  It is easy to calculate the total number of vertices in the construction: $N=M+M^\varepsilon$.

  We present a toy example illustrating the construction of $T$ in Fig.~\ref{fig:reductionHBGM_toy}, where the explicit construction of $T$ corresponding to a small graph with $4$ vertices and $4$ edges is given.

  \begin{figure}
    \centerline{\includegraphics[width=\textwidth]{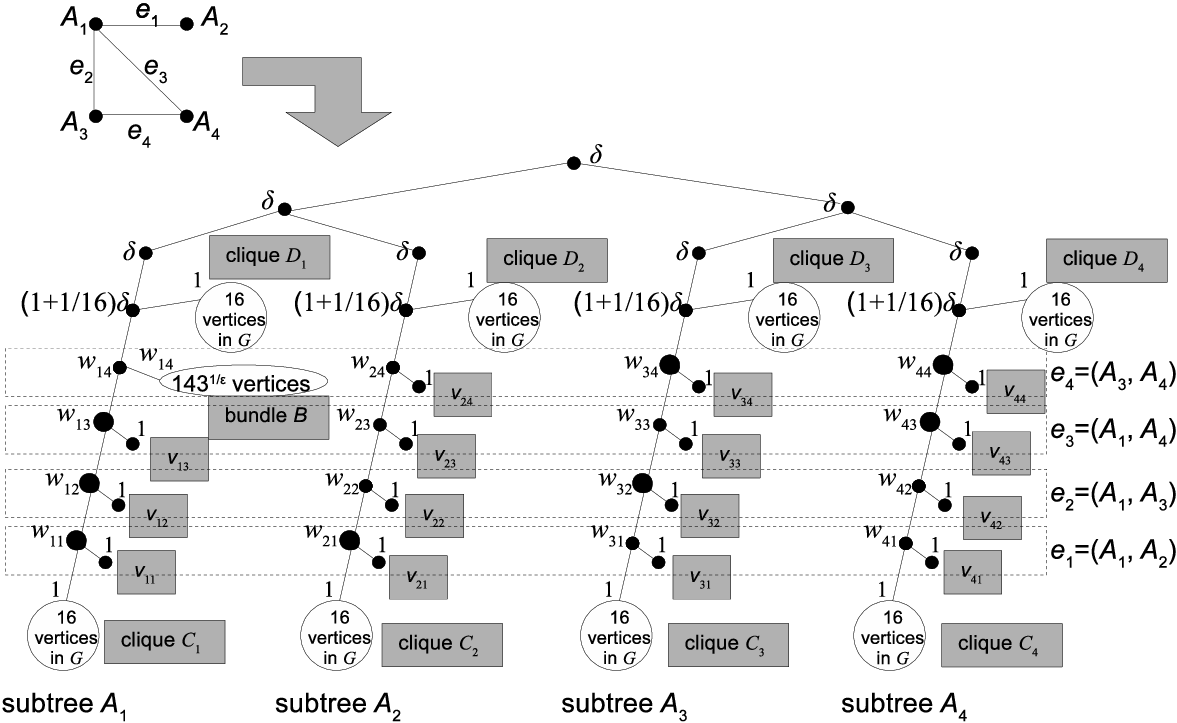}}
    \caption{A toy example illustrating the construction: in this example, $n=m=4$, $W=mn=16$, $\bar{k}=2$, $M=(n(2W+m)-1)^{1/\varepsilon}=143^{1/\varepsilon}$, and the values of $\varepsilon$ and $\delta$ are set sufficiently small and unassigned for clarity. The $w_{ij}$s, defined according to (\ref{eqn:wij}), are as follows: corresponding to the edge $e_1=(A_1,A_2)$, we have $w_{11}=w_{21}=\frac{1-81\delta}{16}$ (shown by larger dots) and $w_{31}=w_{41}=\frac{1-82\delta}{16}$ (shown by smaller dots); corresponding to the edge $e_2=(A_1,A_3)$, we have $w_{12}=w_{32}=\frac{1-84\delta}{17}$ (shown by larger dots) and $w_{22}=w_{42}=\frac{1-85\delta}{17}$ (shown by smaller dots); corresponding to the edge $e_3=(A_1,A_4)$, we have $w_{13}=w_{43}=\frac{1-87\delta}{18}$ (shown by larger dots) and $w_{23}=w_{33}=\frac{1-88\delta}{18}$ (shown by smaller dots); corresponding to the edge $e_4=(A_3,A_4)$, we have $w_{34}=w_{44}=\frac{1-90\delta}{18}$ (shown by larger dots) and $w_{14}=w_{24}=\frac{1-91\delta}{18}$ (shown by smaller dots). In this example, $\varG$ has a $2$-vertex cover $\{A_1,A_3\}$. Corresponding to this, the $k=m+\vark=6$ seeds should be put at $C_1,C_2,C_3,C_4,D_1,D_3$ respectively, so that the vertices in the large bundle $B$ (the one containing $143^{1/\varepsilon}$ vertices) will be eventually infected: firstly, all the vertices in $C_1,C_2,C_3,C_4,D_1,D_3$ will be infected; in the next step, the influence of these infected vertices is just enough to infect $v_{11}$, as $|C_1|w_{11}+|D_1|\cdot(1+\frac1{16})\delta+(|C_2|+|C_3|+|C_4|+|D_3|)\delta=1$; we can also check by calculation that the additional infection of $v_{11}$ will further infect $v_{21}$, and the additional infection of $v_{21}$ will further infect $v_{31},v_{41}$, making all the four vertices corresponding to $e_1$ infected; finally, it is easy to check that the cascade will carry on level-by-level and eventually reach the bundle $B$. In general, each level $i$ corresponding to the edge $e_i=(A_{j_1},A_{j_2})$ contains $n$ vertices $v_{i1},\ldots,v_{in}$, and two of them, $v_{ij_1},v_{ij_2}$, are connected to the tree by a weight heavier than that of the remaining $n-2$ vertices. If the vertices in at least one of $D_{j_1},D_{j_2}$ are infected (corresponding to the case the vertex $A_{j_1}$ or $A_{j_2}$ is included in the vertex cover), after the infection of the vertices in the first $i-1$ levels, the corresponding vertex in $v_{ij_1},v_{ij_2}$ will be infected, which will further infect all the remaining $n-1$ vertices in the $i$-th level. On the other hand, if none of the vertices in $D_{j_1},D_{j_2}$ is infected, even if the cascade reaches the $(i-1)$-th level, no vertex in the $i$-th level will be infected and the cascade will end here without reaching the bundle $B$ which contains most vertices of $G$.}
    \label{fig:reductionHBGM_toy}
  \end{figure}

\paragraph{The Reduction Correctness}
  For a \VC instance $(\varG,\vark)$, consider the \infmax instance $(G,F,\D,k)$ with $k=n+\vark$.
  We aim to show that,
  \begin{enumerate}
    \item If the \VC instance $(\varG,\vark)$ is a \yes instance, then there exists $S\subseteq V$ with $|S|=k$ such that $\sigma(S)\geq M$;
    \item If the \VC instance $(\varG,\vark)$ is a \no instance, then for any $S\subseteq V$ with $|S|=k$ we have $\sigma(S)\leq M^\varepsilon=n(2W+m)-1$.
  \end{enumerate}

\begin{proof}[Proof of (1)]
  Suppose we have a \yes~\VC instance $(\varG,\vark)$ with $\varS\subseteq\varV$ covering all edges in $\varE$.
  In the \infmax instance, we aim to show that at least $M$ vertices will be infected if we choose those $k=n+\vark$ seeds in the following way:
  \begin{itemize}
    \item choose an arbitrary seed in each of the cliques $C_1,\ldots,C_n$ (a total of $n$ seeds are chosen);
    \item for each $A_i\in\varS$, choose an arbitrary seed in the clique $D_i$ (a total of $\vark$ seeds are chosen).
  \end{itemize}

  By such a choice, in the first round of the cascade, all the $W$ vertices in each of $C_1,\ldots,C_n$ and each of those $\vark$ $(D_i)$'s are infected.
  We aim to show that all vertices in $B$ will be infected after at most $3m$ cascade rounds.
  We call the set of $n$ vertices $\{v_{1j},\ldots,v_{nj}\}$ \emph{the $j$-th level}, and we will show that the cascade carries on level by level.
  In particular, we will first show that all vertices in the first level will be infected in at most $3$ rounds.
  Next, given that all vertices in the first $j$ levels are infected, by similar calculations, we can show that all vertices in the $(j+1)$-th level will be infected.

  Consider the first level $\{v_{11},\ldots,v_{n1}\}$.
  Let $e_1=(A_{i_1},A_{i_1'})\in\varE$.
  Since the \VC instance is a \yes instance, either $A_{i_1}\in\varS$ or $A_{i_1'}\in\varS$, or both.
  Assume $A_{i_1}\in\varS$ without loss of generality, then all vertices in $D_{i_1}$ are already infected.
  In the coming round, the vertex $v_{i_11}\in V$ will be infected, as
  \begin{align*}
    f_{v_{i_11}}\left(\bigcup_{i=1}^{n}C_{i}\cup\bigcup_{A_{i}\in\varS}D_{i}\right)&=\delta\left|\bigcup_{i\neq i_1}C_i\cup\bigcup_{i\neq i_1,A_{i}\in\varS}D_{i}\right|+w_{i_11}|C_{i_1}|+\delta\left(1+\frac1W\right)|D_{i_1}|\\
   &=\delta((n-1)+(\vark-1))W+\frac{1-(n+\vark-1)W\delta-\delta}{W}\cdot W\\
   &\qquad+\delta\left(1+\frac1W\right)W\\
   &=1.
  \end{align*}
  If $A_{i_1'}\in\varS$ as well, then $v_{i_1'1}\in V$ will also be infected in the this round, due to the same calculation.
  On the other hand, if $A_{i_1'}\notin\varS$, $v_{i_1'1}$ will be infected in the next round, as
  \begin{align*}
    f_{v_{i_1'1}}\left(\bigcup_{i=1}^{n}C_{i}\cup\bigcup_{A_{i}\in\varS}D_{i}\cup\{v_{i_11}\}\right)&=\delta\left|\bigcup_{i\neq i_1'}C_i\cup\bigcup_{A_{i}\in\varS}D_{i}\cup\{v_{i_11}\}\right|+w_{i_1'1}|C_{i_1'}|\\
   &=\delta((n-1+\vark)W+1)+\frac{1-(n+\vark-1)W\delta-\delta}{W}\cdot W\\
   &=1.
  \end{align*}
  Therefore, both $v_{i_11}$ and $v_{i_1'1}$ will be infected in both cases.

  In the next round, the remaining $n-2$ vertices $\{v_{i_01}\}_{i_0\notin\{i_1,i_1'\};1\leq i_0\leq n}$ will be infected, as we have
  \begin{align*}
    f_{v_{i_01}}\left(\bigcup_{i=1}^{n}C_{i}\cup\bigcup_{A_{i}\in\varS}D_{i}\cup\{v_{i_11},v_{i_1'1}\}\right)
   &=\delta\left|\bigcup_{i\neq i_0}C_i\cup\bigcup_{A_{i}\in\varS}D_{i}\cup\{v_{i_11},v_{i_1'1}\}\right|+w_{i_01}|C_{i_0}|\\
   &=\delta((n-1+\vark)W+2)+\\
   &\qquad\frac{1-(n+\vark-1)W\delta-2\delta}{W}\cdot W\\
   &=1,
  \end{align*}
  in the case $A_{i_0}\notin\varS$ (such that no vertex in $D_{i_0}$ is infected at this moment), and
  \begin{align*}
    &f_{v_{i_01}}\left(\bigcup_{i=1}^{n}C_{i}\cup\bigcup_{i\neq i_0,A_{i}\in\varS}D_{i}\cup\{v_{i_11},v_{i_1'1}\}\right)\\
   =&\delta\left|\bigcup_{i\neq i_0}C_i\cup\bigcup_{i\neq i_0,A_{i}\in\varS}D_{i}\cup\{v_{i_11},v_{i_1'1}\}\right|+w_{i_01}|C_{i_0}|+\delta\left(1+\frac1W\right)|D_{i_0}|\\
   =&\delta((n-1+\vark-1)W+2)+\frac{1-(n+\vark-1)W\delta-2\delta}{W}\cdot W+\delta\left(1+\frac1W\right)W\\
   =&1+\delta>1,
  \end{align*}
  in the case $A_{i_0}\in\varS$ (such that all vertices in $D_{i_0}$ are infected at the first round).
  In conclusion, all the $n$ vertices $\{v_{i1}\}_{1\leq i\leq n}$ will be eventually infected in at most $3$ rounds.

  The analysis of the second level is similar.
  For $e_2=(A_{i_2},A_{i_2'})\in\varE$, we have either $A_{i_2}\in\varS$ or $A_{i_2'}\in\varS$ (or both), making one of $v_{i_22},v_{i_2'2}$ infected (or both), which further makes both $v_{i_22},v_{i_2'2}$ infected (if one of them is not infected previously), and which eventually makes all the $n$ vertices $\{v_{i2}\}_{1\leq i\leq n}$ infected.

  For each $j=1,\ldots,m$ with $e_j=(A_{i_j},A_{i_j'})$, we have either $A_{i_j}\in\varS$ or $A_{i_j'}\in\varS$ (or both).
  Similar as above, after either two or three rounds, all the vertices in $\{v_{ij}\}_{1\leq i\leq n}$ will be infected, if all the vertices in $\{v_{i1}\}_{1\leq i\leq n},\ldots,\{v_{i(j-1)}\}_{1\leq i\leq n}$ are already infected.

  Therefore, we can see that the cascade after the first round carries on in the following order:
  $$v_{i_11}\rightarrow v_{i_1'1}\rightarrow\{v_{i1}\}_{i\neq i_1,i_1'}\rightarrow v_{i_22}\rightarrow v_{i_2'2}\rightarrow\{v_{i2}\}_{i\neq i_2,i_2'}\rightarrow\cdots $$
  $$\rightarrow v_{i_mm}\rightarrow v_{i_m'm}\rightarrow\{v_{im}\}_{i\neq i_m,i_m'}\rightarrow B.$$
  Therefore, we conclude 1 as we already have $M$ infected vertices by just counting those in the bundle $B$.
\end{proof}

For the proof of (2), we present a general proof idea before the formal proof.

To show (2) by contradiction, we assume that we can choose a seed set $S\subseteq V$ such that $|S|=k=n+\vark$ and $\sigma(S)>M^\varepsilon$.
By a careful analysis, we can conclude that the only possible way to choose $S$ is as follows.
  \begin{itemize}
    \item an arbitrary vertex from each of $C_1,\ldots,C_n$ (a total of $n$ vertices are chosen);
    \item an arbitrary vertex from each of $D_{\pi_1},\ldots,D_{\pi_\vark}$ for certain $\{\pi_1,\ldots,\pi_k\}\subseteq\{1,\ldots,n\}$ (a total of $\vark$ vertices are chosen).
  \end{itemize}
The intuitive reason for this is the following:
firstly, choosing $k$ seeds among the $2n$ cliques $C_1,\ldots,C_n,D_1,\ldots,D_n$ is considerably more beneficial, as a seed would cause the infection of $W$ vertices;
secondly, if we cannot choose both $C_i$ and $D_i$, it is always better to choose $C_i$ because the weights $w_{i1},\ldots,w_{im}$ are considerably larger than $\delta(1+1/W)$, if $\delta$ is set sufficiently small.

Since the \VC instance is a \no instance, there exists an edge $e_j=(A_{i_j},A_{i_j'})$ such that no vertex in $D_{i_j}$ and $D_{i_j'}$ is chosen as seed.
By following similar analysis as in the proof of 1, we can see that the cascade would stop at the level $\{v_{ij}\}_{i=1,\ldots,n}$, which concludes (2).

\begin{proof}[Proof of (2)]
  Assume that we can choose seed set $S\subseteq V$ such that $|S|=k=n+\vark$ and $\sigma(S)>M^\varepsilon$.
  First notice that choosing any seeds from $B$ is at most as good as choosing seeds from $C_1\cup\{v_{1j}\}_{1\leq j\leq m-1}$.
  By our assumption $m>n+\vark=k$, we can assume without loss of generality that no seed is chosen in $B$.
  With this assumption, we will prove that none of these $M$ vertices will be infected in the cascade.
  Since the graph $G$ has a total of $N=M+M^\varepsilon$ vertices, this contradicts that $\sigma(S)>M^\varepsilon$.

  Suppose, for the sake of contradiction, a vertex $u\in B$ is infected in round $t$ of the cascade, and there is no infected vertex in $B$ in the first $t-1$ rounds.
  Let $I_u$ be the set of infected vertices before round $t$.
  Since $u$ is infected in round $t$, we have $f_u(I_u)\geq1$, which, by Definition~\ref{defi:HBG}, implies
  $$\sum_{v\in I_u}w(u,v)\geq1.$$
  We analyze the constituents of $I_u$.

  We set $\delta$ to be sufficiently (but still polynomially) small such that
  $$(n-1)(2W+m)\delta+\delta\left(1+\frac1W\right)\ll w_{1m}.\footnote{This is always possible: when $\delta\rightarrow0$, the left-hand side approaches to $0$, while we have $\lim_{\delta\rightarrow0}w_{1m}=\frac{1}{W+m-1}$ for the right-hand side.}$$
  Then the infection of each vertex in $C_1\cup\{v_{1j}\}_{1\leq j\leq m-1}$ has contribution $w_{1m}$ to $f_u(I_u)$, while the net contribution from the infections of all vertices in $V\setminus\{C_1\cup\{v_{1j}\}_{1\leq j\leq m-1}\cup B\}$ is much less than $w_{1m}$.
  On the other hand, even if all the $W+m-1$ vertices in $C_1\cup\{v_{1j}\}_{1\leq j\leq m-1}$ are included in $I_u$, the contribution to $f_u(I_u)$ is
  $$(W+m-1)w_{1m}\leq1-(n+\vark-1)W\delta-(n-1)(m-1)\delta-\delta<1,$$
  which is still not enough.
  Thus, we conclude that $C_1\cup\{v_{1j}\}_{1\leq j\leq m-1}\subseteq I_u$, and the vertices from $V\setminus\{C_1\cup\{v_{1j}\}_{1\leq j\leq m-1}\cup B\}$ should contribute at least $(n+\vark-1)W\delta+(n-1)(m-1)\delta+\delta$ to $f_u(I_u)$.
  From the term $(n+\vark-1)W\delta$, we can see that at least $n+\vark-1$ cliques from the $2n-1$ cliques $C_2,\ldots,C_n,D_1,\ldots,D_n$ must be included in $I_u$.
  Coupled with the observation $C_1\subseteq I_u$, we need at least $n+\vark$ infected cliques from $C_1,\ldots,C_n,D_1,\ldots,D_n$.

  On the other hand, the only way to infect a clique $C_i$ or $D_i$ is to seed one of its vertices.
  To see this for each $D_i$, it is enough to notice that the weight $\delta(1+1/W)$ is extremely small.
  To see this for each $C_i$, notice that only $v_{i1},\ldots,v_{im}$ have non-negligible influence to $C_i$, and
  $$\sum_{j=1}^mw_{ij}<\sum_{j=1}^m\frac1{W-1+j}<m\times\frac1W=\frac1n\ll1.$$
  Therefore, to have $u\in B$ infected in round $t$, the only possible way is to choose $k=n+\vark$ seeds from $n+\vark$ cliques, among all the $2n$ cliques $C_1,\ldots,C_n,D_1,\ldots,D_n$.
  Lastly, it is straightforward to check that the infection of an vertex in $D_i$ is less influential than the infection of an vertex in the corresponding $C_i$ (both $C_i$ and $D_i$ contain the same number of vertices, so their influences to the outside subtrees are the same; however, $C_i$ is connected to $v_{i1},\ldots,v_{im}$ with higher weights than $D_i$).
  Thus, we can assume without loss of generality that $S$ consists of
  \begin{itemize}
    \item an arbitrary vertex from each of $C_1,\ldots,C_n$ (a total of $n$ vertices are chosen);
    \item an arbitrary vertex from each of $D_{\pi_1},\ldots,D_{\pi_\vark}$ for certain $\{\pi_1,\ldots,\pi_k\}\subseteq\{1,\ldots,n\}$ (a total of $\vark$ vertices are chosen).
  \end{itemize}

  Since the \VC instance is a \no instance, for the choice $\varS=\{A_{\pi_1},\ldots,A_{\pi_\vark}\}$, there exists edge $e_j$ that is not covered by $\varS$.
  Let $j^\ast$ be the smallest $j$ such that $e_j$ is not covered by $\varS$.

  We first deal with the case $j^\ast=1$. The case where $j^\ast>1$ is dealt with subsequently.

  If $j^\ast=1$, for $e_1=(A_{i_1},A_{i_1'})$, we have $A_{i_1},A_{i_1'}\notin\varS$.
  In this case, $v_{i_11}$ will not be infected, as
  \begin{align*}
    f_{v_{i_11}}\left(\bigcup_{i=1}^{n}C_{i}\cup\bigcup_{A_{i}\in\varS}D_{i}\right)&=\delta\left|\bigcup_{i\neq i_1}C_i\cup\bigcup_{A_{i}\in\varS}D_{i}\right|+w_{i_11}|C_{i_1}|\\
   &=\delta(n-1+\vark)W+\frac{1-(n+\vark-1)W\delta-\delta}{W}\cdot W\\
   &=1-\delta<1,
  \end{align*}
  and $v_{i_1'1}$ will not be infected for the same reason.
  For $i_0\neq i_1,i_1'$, $v_{i_01}$ will not be infected either, as we have
  \begin{align*}
    f_{v_{i_01}}\left(\bigcup_{i=1}^{n}C_{i}\cup\bigcup_{A_{i}\in\varS}D_{i}\right)&=\delta\left|\bigcup_{i\neq i_0}C_i\cup\bigcup_{A_{i}\in\varS}D_{i}\right|+w_{i_01}|C_{i_0}|\\
   &=\delta(n-1+\vark)W+\frac{1-(n+\vark-1)W\delta-2\delta}{W}\cdot W\\
   &=1-2\delta<1,
  \end{align*}
  in the case $A_{i_0}\notin\varS$, and
  \begin{align*}
    f_{v_{i_01}}\left(\bigcup_{i=1}^{n}C_{i}\cup\bigcup_{A_{i}\in\varS}D_{i}\right)&=\delta\left|\bigcup_{i\neq i_0}C_i\cup\bigcup_{i\neq i_0,A_{i}\in\varS}D_{i}\right|+w_{i_01}|C_{i_0}|+\delta\left(1+\frac1W\right)|D_{i_0}|\\
   &=\delta(n-1+\vark-1)W+\frac{1-(n+\vark-1)W\delta-2\delta}{W}\cdot W\\
   &\qquad+\delta\left(1+\frac1W\right)W\\
   &=1-\delta<1,
  \end{align*}
  in the case $A_{i_0}\in\varS$.
  Thus, none of $\{v_{i1}\}_{1\leq i\leq n}$ will be infected.
  Since $w_{ij_1}>w_{ij_2}$ whenever $j_1<j_2$ for any $i$ (easy to see by observing $w_{ij}\approx\frac1{W-1+j}$), none of $\{v_{ij}\}_{1\leq i\leq n;2\leq j\leq m}$ will be infected.
  In particular, no vertex in $B$ can be infected, which leads to the desired contradiction.

  If $j^\ast>1$, by the similar analysis in the proof of 1 for the \yes instance case, after many cascade rounds, all vertices in $\{v_{ij}\}_{1\leq i\leq n;1\leq j\leq j^\ast-1}$ will be infected.
  For $e_{j^\ast}=(A_{i_{j^\ast}},A_{i_{j^\ast}'})$, we have $A_{i_{j^\ast}},A_{i_{j^\ast}'}\notin\varS$.
  In this case, $v_{i_{j^\ast}j^\ast}$ will not be infected, as
  \begin{align*}
    &f_{v_{i_{j^\ast}j^\ast}}\left(\bigcup_{i=1}^{n}C_{i}\cup\bigcup_{A_{i}\in\varS}D_{i}\cup\{v_{ij}\}_{1\leq i\leq n;1\leq j\leq j^\ast-1}\right)\\
   =&\delta\left|\bigcup_{i\neq i_{j^\ast}}C_i\cup\bigcup_{A_{i}\in\varS}D_{i}\right|+w_{i_{j^\ast}j^\ast}\left|C_{i_{j^\ast}}\cup\{v_{i_{j^\ast}j}\}_{1\leq j\leq j^\ast-1}\right|+\delta\left|\{v_{ij}\}_{i\neq i_{j^\ast},1\leq j\leq j^\ast-1}\right|\\
   =&\delta(n-1+\vark)W\\
   &\qquad+\frac{1-(n+\vark-1)W\delta-(n-1)(j^\ast-1)\delta-\delta}{W-1+j^\ast}\cdot(W+j^\ast-1)+\delta(n-1)(j^\ast-1)\\
   =&1-\delta<1,
  \end{align*}
  and $v_{i_{j^\ast}'j^\ast}$ will not be infected for the same reason.
  Following similar analysis, $v_{i_0j^\ast}$ will not be infected for $i_0\neq i_{j^\ast},i_{j^\ast}'$, and none of $\{v_{ij^\ast}\}_{1\leq i\leq n}$ will be infected.
  By the same observation $w_{ij_1}>w_{ij_2}$ whenever $j_1<j_2$, none of $\{v_{ij}\}_{1\leq i\leq n;j^\ast\leq j\leq m}$ will be infected.
  In particular, no vertex in $B$ can be infected, which again leads to the desired contradiction.
  We conclude (2) here.
\end{proof}

Since $M=\Theta(N)$, (1) and (2) imply Theorem~\ref{hardnessHBGM}.

\section{Stochastic Hierarchical Blockmodel Influence Maximization}
\label{sect::resultSHBGM}
In this section, we will present strong inapproximability results for both pre-sampling and post-sampling versions of stochastic hierarchical blockmodel \infmax.
A major difference between the results in Section~\ref{sect::resultHBGM} and this section is that the strong inapproximability result no longer holds if we assume $\theta_v=1$ for all $v\in V$ in the stochastic hierarchical blockmodel.
In fact, if all the thresholds are fixed to be 1 and $F$ is counting (see Definition~\ref{defi:countingLIF}), $\sigma(\cdot)$ in both Definition~\ref{defi:infmax_pre} and Definition~\ref{defi:infmax_post} become submodular, in which case we can have a simple greedy $(1-1/e)$-approximation algorithm~\cite{KempeKT03,Nemhauser78}.
In particular, assuming $\theta_v=1$ for all $v\in V$ makes post-sampling \infmax trivial: as an infected seed will eventually infect a whole connected component of $G$, the optimal way of choosing $S$ is to choose $k$ seeds from the first $k$ largest connected components, after seeing the sampling $G\sim\mathcal{G}$.
For pre-sampling \infmax, the model becomes the \emph{independent cascade model} \cite{KempeKT03}, which is known to be submodular.

The following two theorems are the same, except that Theorem~\ref{hardnessSHBGM_pre} corresponds to the hardness for pre-sampling model (see Definition~\ref{defi:infmax_pre}), while Theorem~\ref{hardnessSHBGM_post} show the same hardness result for the post-sampling model (see Definition~\ref{defi:infmax_post}) via a randomized Karp's reduction.



\begin{theorem}\label{hardnessSHBGM_pre}
  Consider the pre-sampling stochastic hierarchical blockmodel \infmax problem $(\mathcal{G},F,\D,k)$.
  For any $\varepsilon>0$, even if $F$ is counting and $\D_v$ is a point-mass distribution on certain integer $\theta_v$ for each $v\in V$, it is NP-hard to distinguish between the following two cases:
  \begin{itemize}
      \item \yes: there exists a seed set $S$ with $|S|=k$ such that
      $\displaystyle\E_{G\sim\mathcal{G}}\left[\sigma_{F,\D}^G(S)\right]=\Theta(N);$
      \item \no: for any seed set $S$ with $|S|=k$, we have
      $\displaystyle\E_{G\sim\mathcal{G}}\left[\sigma_{F,\D}^G(S)\right]=O(N^\varepsilon).$
  \end{itemize}
\end{theorem}

\begin{theorem}\label{hardnessSHBGM_post}
  Consider the post-sampling stochastic hierarchical blockmodel \infmax problem $(\mathcal{G},F,\D,k)$.
  For any $\varepsilon>0$ and $c>0$, even if $F$ is counting and $\D_v$ is a point-mass distribution on certain integer $\theta_v$ for each $v\in V$, it is NP-hard to distinguish between the following two cases with probability at least $N^{-c}$ (where the probability is taken over $G\sim\mathcal{G}$):
  \begin{itemize}
      \item \yes: there exists a seed set $S$ with $|S|=k$ such that
      $\displaystyle\sigma_{F,\D}^G(S)=\Theta(N);$
      \item \no: for any seed set $S$ with $|S|=k$, we have
      $\displaystyle\sigma_{F,\D}^G(S)=O(N^\varepsilon).$
  \end{itemize}
\end{theorem}

As a remark to Theorem~\ref{hardnessSHBGM_post}, the theorem says that if we have an oracle that outputs a solution which approximates $\max_{S\subseteq V,|S|\leq k}\sigma(S)$ within a factor of $N^{1-\varepsilon}$ for certain samples $G\sim\mathcal{G}$, and with probability at least $N^{-c}$ we receive a sample $G$ in the set of graphs for which the oracle outputs valid solutions, then we can use this oracle to solve any NP-complete problem as long as we have randomness to sample $G\sim\mathcal{G}$.

We will prove both Theorem~\ref{hardnessSHBGM_pre} and Theorem~\ref{hardnessSHBGM_post} by a reduction from \VC.
Given a \VC instance $(\varG=(\varV,\varE),\vark)$, we will construct a hierarchy tree $T$ which determines $\mathcal{G}$ for both proofs.

\paragraph{The Reduction}
Let $n=|\varV|$ and $m=\varE$ as usual.
Assume $m>n>\vark^2+2$, and $\log_2n$ is an integer.\footnote{Notice that we can assume $n\gg\vark$ is an integer power of $2$ by adding isolated vertices to $\varG$ which are never picked, and we can assume $m>n$ by duplicate each edge (which makes $\varG$ a multi-graph).}
In addition, we assume that $A_1\in\varS$ whenever the \VC instance is a \yes instance.\footnote{This assumption can be made without loss of generality because we can add two extra vertices named $A_1,A_2$ and one extra edge $(A_1,A_2)$ such that one of $A_1,A_2$ much be chosen to cover this edge, and we can assume $A_1$ is chosen.}

We define the following variables used in this section.
$$\delta=\frac1{10mn^2\vark},\qquad\mbox{and}\qquad\Delta=mn^2\delta=\frac1{10\vark},\qquad W=m^{10}n^{10}.$$
Let $M$ be an extremely large number whose value will be decided later.

The construction of $T$ is shown in Figure~\ref{fig:reductionSHBGM}.
$T$ is a full balanced binary tree with $\log_2n$ levels and $n$ leaves.
The weight of all non-leaf nodes is $1/W$, and the weight of all leaves is $1$.
The $i$-th leaf corresponds to $A_i\in\varV$ in the \VC instance.
Recall from Definition~\ref{defi:SHBG} that $\mathcal{G}=(V,T)$ is determined by $T$, and in particular each leaf of $T$ corresponds to a subset of $V$.
As the weight of each leaf is $1$, meaning each edge appear with probability $1$, its corresponding subset of vertices forms a clique in all $G\sim\mathcal{G}$.
We will call the clique corresponding to the $i$-th leaf \emph{the $i$-th clique} in the remaining part of this section.
For each clique $i$, we will first describe the vertices we have constructed in Figure~\ref{fig:reductionSHBGM}, and then define their thresholds.

For positive integers $x,y$, denote by $B(x,y)$ a bundle of $x$ vertices with threshold $y$.
For each $i=1,\ldots,n$, we construct the following vertices for the $i$-th clique:
\begin{itemize}
  \item a bundle of $\vark W^2$ vertices: $B_{i}:=B\left(\vark W^2,\infty\right)$, and
  \item $m(n-2)$ bundles of $W^3$ vertices: $B_{ij\imath}:=B\left(W^3,\theta_{ij\imath}\right)$ for $j=1,\ldots,m$ and $\imath=1,\ldots,n-2$.
\end{itemize}
For $i=1$, we add an extra bundle $C:=\left(M,\theta_{1(m+1)}\right)$.
The thresholds $\{\theta_{ij\imath}\}$ and $\theta_{1(m+1)}$ of those constructed vertices will be defined later.

By our construction, the $1$-st clique has $M+\vark W^2+m(n-2)W^3$ vertices, which is much more than the number of vertices $\vark W^2+m(n-2)W^3$ in each of the remaining cliques.
As a remark, we have constructed $N=M+nm(n-2)W^3+n\vark W^2$ vertices for $G$.
Moreover, for $M$ whose value we have not decided yet, we can make it arbitrarily close to $N$.

Denote by $B_{\cdot j\imath}:=\{B_{ij\imath}\}_{i=1,\ldots,n}$ the $n$ bundles in a horizontal level in Figure~\ref{fig:reductionSHBGM} (for example, in Figure~\ref{fig:reductionSHBGM}, after the top-level $\{B_1,\ldots,B_n\}$, there come levels $B_{\cdot11},B_{\cdot12},\ldots$).
We will call $B_{\cdot j\imath}$ a \emph{level} and abuse the word ``level'' to refer to the vertices in $B_{\cdot j\imath}$.

The correspondence between the \VC instance and the graph we constructed is as follows.
Recall that each vertex $A_i\in\varV$ corresponds to the $i$-th clique.
Now, for each edge $e_j\in\varE$, we have constructed $n-2$ levels $B_{\cdot j1},\ldots,B_{\cdot j(n-2)}$, which are $n(n-2)$ bundles of $W^3$ vertices.
For example, in Figure~\ref{fig:reductionSHBGM}, we have illustrated the $n-2$ levels corresponding to $e_1$ and the $n-2$ levels corresponding to $e_m$, while the levels corresponding to the remaining edges in $\varE$ are omitted.

For each $j=1,\ldots,m$ and $\imath=1,\ldots,n-2$, we denote by $B_{\prec j\imath}$ the union of the first $(j-1)(n-2)+\imath-1$ levels (where the levels are ordered from up to down in Figure~\ref{fig:reductionSHBGM}):
\begin{align*}
  B_{\prec j\imath}:=&\bigcup_{(n-2)j'+\imath'<(n-2)j+\imath}B_{\cdot j'\imath'}\\
  =&B_{\cdot11}\cup B_{\cdot12}\cup\ldots\cup B_{\cdot1(\imath-1)} \cup B_{\cdot1\imath}\cup\ldots\cup B_{\cdot1(n-3)}\cup B_{\cdot1(n-2)}\cup\\
  &B_{\cdot21}\cup B_{\cdot22}\cup\ldots\cup
  B_{\cdot2(\imath-1)} \cup B_{\cdot2\imath}\cup\ldots\cup B_{\cdot2(n-3)}\cup B_{\cdot2(n-2)}\cup\\
  &\cdots\\
  &B_{\cdot j1}\cup B_{\cdot j2}\cup\ldots\cup B_{\cdot j(\imath-1)}.\\
\end{align*}

Next, we define the thresholds $\{\theta_{ij\imath}\}$ and $\theta_{1(m+1)}$.
Denote
$$\omega_{j\imath}:=\left((j-1)(n-2)+(\imath-1)\right)W^3+(n-1)\left((j-1)(n-2)+(\imath-1)\right)W^2,$$
which is the expected number of neighbors of each $b_{ij\imath}\in B_{ij\imath}$ in $B_{\prec j\imath}$.
For each fixed $j$, denote by $i_j,i_j'$ the two indices such that $e_j=(A_{i_j},A_{i_j'})$ with $i_j<i_j'$, and all $\theta_{ij\imath}$'s are defined as follows.
$$
\left[\begin{array}{cccc}
  \theta_{1j1} & \theta_{2j1} & \cdots & \theta_{nj1}\\
  \theta_{1j2} & \theta_{2j2} & \cdots & \theta_{nj2}\\
  \vdots & \vdots & \ddots & \vdots\\
  \theta_{1jn} & \theta_{2jn} & \cdots & \theta_{njn}
\end{array}\right]:=\left[\begin{array}{cccc}
  \omega_{j1}+(1-\Delta)W^2 & \omega_{j1}+(1-\Delta)W^2 & \cdots & \omega_{j1}+(1-\Delta)W^2 \\
  \omega_{j2}+(1-\Delta)W^2 & \omega_{j2}+(1-\Delta)W^2 & \cdots & \omega_{j2}+(1-\Delta)W^2 \\
  \vdots & \vdots & \ddots & \vdots\\
  \omega_{jn}+(1-\Delta)W^2 & \omega_{jn}+(1-\Delta)W^2 & \cdots & \omega_{jn}+(1-\Delta)W^2 \\
\end{array}\right]+\\
$$
$$
\begin{blockarray}{cccccccccc}
 & & &  &\mbox{Column }i_j &  &\mbox{Column }i_j' &  &  &  \\
\begin{block}{[cccccccccc]}
  1W^2 & 2W^2 & 3W^2 & \cdots & 0 & \cdots & 0 & \cdots & (n-3)W^2 & (n-2)W^2\\
  (n-2)W^2 & 1W^2 & 2W^2 & \cdots & 0 & \cdots & 0 & \cdots & (n-4)W^2 & (n-3)W^2\\
  (n-3)W^2 & (n-2)W^2 & 1W^2 & \cdots & 0 & \cdots & 0 & \cdots & (n-5)W^2 & (n-4)W^2\\
  \vdots & \vdots & \vdots & \vdots & \vdots & \vdots & \vdots & \vdots & \vdots & \vdots\\
  2W^2 & 3W^2 & 4W^2 & \cdots & 0 & \cdots & 0 & \cdots & (n-2)W^2 & 1W^2\\
\end{block}
\end{blockarray}
$$
Notice that for different $\imath_1,\imath_2\in\{1,\ldots,n-2\}$, $\left(\theta_{1j\imath_1}-\omega_{j\imath_1},\theta_{2j\imath_1}-\omega_{j\imath_1},\ldots,w_{nj\imath_1}-\omega_{j\imath_1}\right)$ is a permutation of $\left(w_{1j\imath_2}-\omega_{j\imath_2},w_{2j\imath_2}-\omega_{j\imath_2},\ldots,w_{nj\imath_2}-\omega_{j\imath_2}\right)$.
Specifically, for the second matrix above, excluding the $i_j$-th and the $i_j'$-th columns, the first row is an arithmetic progression $1W^2,2W^2,(n-2)W^2$, and the $(\imath+1)$-th row is obtained by cyclically shifting the $\imath$-th row to the right by $1$ unit.

Finally, for the threshold $\theta_{1(m+1)}$ of each vertex in the bundle $C$. We define
$$\theta_{1(m+1)}:=m(n-2)W^3+(n-1)m(n-2)W^2+(1-\Delta) W^2.$$
As we will see later, $\theta_{1(m+1)}$ is slightly less than the expected number of neighbors of each $c\in C$ in $V\setminus C$, by an amount of $\Theta(\Delta W^2)$.

\begin{figure}
    \centerline{\includegraphics[width=\textwidth]{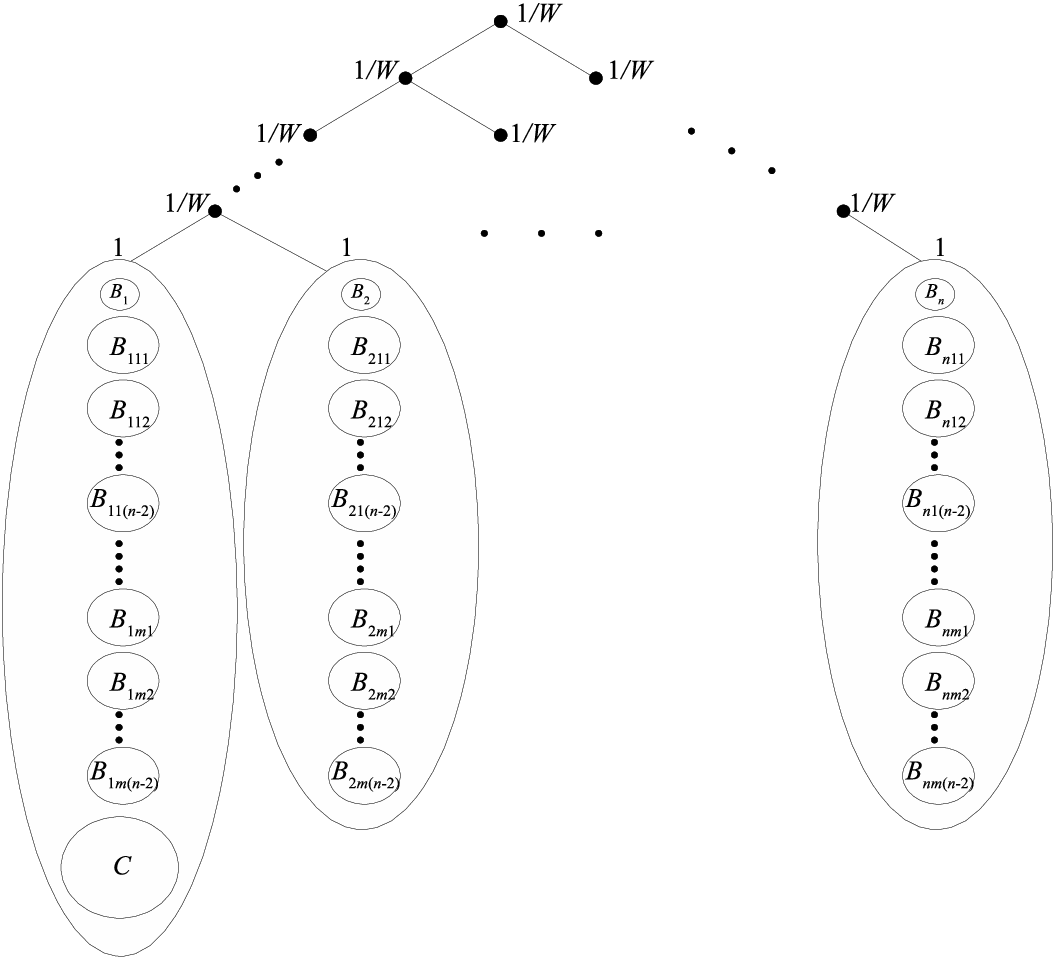}}
    \caption{The construction of the hierarchy tree $T$.}
    \label{fig:reductionSHBGM}
\end{figure}

\paragraph{The High-level Ideas}
Before presenting rigorous arguments, we provide high level ideas of the reduction in this subsection.

We have constructed the hierarchy tree $T$, which corresponds to a graph distribution $\mathcal{G}$ (refer to Definition~\ref{defi:SHBG}).
In the next subsection, we will show that a sample $G\sim\mathcal{G}$ can simulate the corresponding \VC instance with high probability.
In particular, we will say such samples are ``good'' samples, which we will define rigorously, and we will prove that a sample is good with probability $1-o(1)$.

Given a \VC instance $(\varG,\vark)$, we consider the \infmax instance $(G,F,\D,k)$, where $G$ is a good sample, $F,\D$ are as defined in Theorem~\ref{hardnessSHBGM_pre} (or Theorem~\ref{hardnessSHBGM_post}), and $k=\vark W^2$.

Suppose we have a good sample $G$.
If the \VC instance is a \yes instance,
we can find $\varS\subseteq\varV$ with $|\varS|=\vark$ such that $\varS$ covers all edges in $\varE$.
For each $A_i\in\varS$, we choose $W^2$ seeds from the bundle $B_i$, so a total of $\vark W^2=k$ seeds are chosen.

Similar to what happens in Section~\ref{sect::resultHBGM}, the cascade will flow level-by-level.
In particular, for the first edge $e_1\in\varE$ and $i_1,i_1'$ such that $e_1=(A_{i_1},A_{i_1'})$, the vertices in the bundles $B_{i_111}$ and $B_{i_1'11}$ have the lowest threshold in the level $B_{\cdot11}$.
On the other hand, by our choice of $k$ seeds, we have chosen $W^2$ seeds from one (or both) of $B_{i}$ and $B_{i'}$.
Calculations show that these seeds are just enough to infect all vertices in $B_{i_111}$ and $B_{i_1'11}$.
The infection of these vertices will eventually infected the entire level $B_{\cdot11}$,
and similar analysis shows that the levels $B_{\cdot12},B_{\cdot13},...$ will be infected one-by-one.
Finally, the cascade can reach the huge bundle $C$, and most vertices in $G$ will be infected.

If the \VC instance is a \no instance, we can assume all seeds are chosen from $\{B_1,\ldots,B_n\}$, as it is always a better idea to choose seeds from vertices having higher thresholds in a clique.\footnote{Rigorously, this may not be true in the post-sampling case, where the seed-picker can see the sample $G$. The vertices not in $\{B_1,\ldots,B_n\}$ may happen to have more neighbors across cliques, and the seed-picker can take advantage of this. We will reason about this later. However, for now, we assume all seeds are chosen from $\{B_1,\ldots,B_n\}$.}
We say that the $i$-th clique is activated if we have chosen almost $W^2$ seeds from $B_i$, or more than this number.
We can draw an analogy between activating the $i$-th clique in \infmax and picking the set $A_i$ in \VC.

Since the \VC instance is a \no instance, certain element $e_{j^\ast}$ is not covered, and we will show that the cascade will stop at one of the $n-2$ levels $B_{\cdot j^\ast1},\ldots,B_{\cdot j^\ast(n-2)}$.
Intuitively, the thresholds of vertices in these levels shift cyclically by our construction, and there exists a level whose vertices' thresholds are shifted to the position such that the cascade fails on all leaves.
In particular, even if we put all $k=\vark W^2$ seeds in a single bundle $B_i$, there exists a level $\imath$ such that $\theta_{ij^\ast\imath}$ is large enough, making the cascade still fail on leaf $i$.
On the other hand, there are only two leaves $i_{j^\ast},i_{j^\ast}'$ having lowest $\theta_{ij^\ast\imath}$ in all levels $\imath=1,\ldots,n-2$, which are exactly those $i_{j^\ast},i_{j^\ast}'$ with $e_{j^\ast}=(A_{i_{j^\ast}},A_{i_{j^\ast}'})$.
However, we have very few seeds (considerably fewer than $W^2$) on the $i_{j^\ast}$-th and the $i_{j^\ast}'$-th cliques, by our assumption that $e_{j^\ast}$ is not covered.

Since the cascade will fail on a certain intermediate level, it cannot reach the huge bundle $C$.
By making $C$ contain most vertices in $G$ (i.e., making $M$ large enough), we can see that the number of infected vertices corresponding to a \yes~\VC instance is significantly higher, which implies both Theorem~\ref{hardnessSHBGM_pre} and Theorem~\ref{hardnessSHBGM_post}.

In the next two subsections, we will rigorously prove the correctness of our reduction.

\paragraph{Good Samplings}
In this subsection, we define ``good'' samplings $G\sim\mathcal{G}$ which are useful in the reduction from \VC, in the sense that $G$ successfully simulates the \VC instance, and we show that a sample $G\sim\mathcal{G}$ is good with a high probability.

Firstly, consider a $W^3$ sized bundle $B_{ij\imath}$, and an arbitrary vertex $v$ not in the $i$-th clique.
Over all the samplings $G\sim\mathcal{G}$, $v$'s expected number of neighbors in $B_{ij\imath}$ is
$$\E_{G\sim\mathcal{G}}\left[|\Gamma(v)\cap B_{ij\imath}|\right]=\frac1W\cdot W^3=W^2.$$

Secondly, consider a set $D_i$ of $\delta W^2$ vertices in the $i$-th clique, and a set $D_{-i}$ of $(\vark+1)W^2$ vertices that are not in the $i$-th clique,
the expected total number of edges between $D_i$ and $D_{-i}$ is
$$\E_{G\sim\mathcal{G}}\left[|\{(u,v):u\in D_i, v\in D_{-i}\}|\right]=\frac1W\cdot \delta W^2\cdot(\vark+1) W^2=\delta(\vark+1)W^3.$$

We define a sampling $G\sim\mathcal{G}$ to be ``good'' if the above two numbers roughly concentrate on their expectations.
\begin{definition}\label{good_defi}
  A sampling $G\sim\mathcal{G}$ is \emph{good} if the following holds.
  \begin{enumerate}
    \item For all $i=1,\ldots,n$, $j=1,\ldots,m$ and $\imath=1,\ldots,n-2$, and any vertex $v$ not in the $i$-th clique,
    $$(1-\delta)W^2<|\Gamma(v)\cap B_{ij\imath}|<(1+\delta)W^2.$$
    \item For any set $D_i$ of $\delta W^2$ vertices in the $i$-th clique, and any set $D_{-i}$ of $(\vark+1)W^2$ vertices that are not in the $i$-th clique,
    the number of edges between $D_i$ and $D_{-i}$ is less than $W^{3.6}$:
    $$\left|\{(u,v):u\in D_i, v\in D_{-i}\}\right|<W^{3.6}.$$
  \end{enumerate}
\end{definition}
The following lemma shows that a sampling $G\sim\mathcal{G}$ is good with high probability.

\begin{lemma}\label{good_lemma}
  A sampling $G\sim\mathcal{G}$ is good with probability more than $1-e^{-\sqrt W}$.
\end{lemma}
\begin{proof}
  We apply Chernoff-Hoeffding inequality and union bounds to show this lemma.
  In a random sample $G\sim\mathcal{G}$, for each $i=1,\ldots,n$; $j=1,\ldots,m$; $\imath=1,\ldots,n-2$ and $v$, requirement 1 in Definition~\ref{good_defi} fails with probability
  $$\Pr\left[\left|W^2-|\Gamma(v)\cap B_{ij\imath}|\right|\geq\delta W^2\right]\leq2\exp\left(-\frac12\left(\delta W^2\right)^2\frac1{W^3}\right)<e^{-W^{0.6}},$$
  where the last inequality is due to $(\delta W^2)^2=\frac1{\vark^2}m^{38}n^{36}>W^{3.6}$.

  For each $D_i$ and $D_{-i}$, requirement 2 in Definition~\ref{good_defi} fails with probability
  $$\Pr\left[\left|\{(u,v):u\in D_i, v\in D_{-i}\}\right|\geq W^{3.6}\right]\leq\exp\left(-\frac12\frac{\left(W^{3.6}-\delta(\vark+1)W^3\right)^2}{\delta W^2\cdot(\vark+1) W^2}\right)<e^{-W^3}.$$
  By a union bound, the probability that a sample $G\sim\mathcal{G}$ is not good is
  \begin{align*}
  \Pr[\mbox{not good}]&<nm(n-2)Ne^{-W^{0.6}}+\binom{N}{\delta W^2}\binom{N}{\vark W^2}e^{-W^3}\\
  &<N^2e^{-W^{0.6}}+N^{\delta W^2+\vark W^2}e^{-W^3}\\
  &=e^{2\log N}e^{-W^{0.6}}+e^{(\delta W^2+\vark W^2)\log N}e^{-W^3}\\
  &<e^{-\sqrt W},\tag{as $N=\text{poly}(W)$, which implies $\log N=o(W^c)$ for arbitrary $c>0$}
  \end{align*}
  which immediately implies the lemma.
\end{proof}

\paragraph{The Reduction Correctness}
In this section, we show that \infmax on a good sample $G\sim\mathcal{G}$ simulates the \VC problem.
\begin{lemma}\label{yesno_SHBGM}
  Consider \infmax with $k=\vark W^2$ seeds. For any good sample $G\sim\mathcal{G}$,
  \begin{enumerate}
    \item if the \VC instance is a \yes instance, a total of $\vark W^2+nm(n-2)W^3+M$ vertices can be infected by properly choosing the $k$ seeds;
    \item if the \VC instance is a \no instance, at most $N-M$ vertices can be infected for any choices of the $k$ seeds.
  \end{enumerate}
\end{lemma}
\begin{proof}[Proof of (1)]
  Suppose the \VC instance is a \yes instance.
  Let $\varS$ be the choice of $\vark$ vertices in \VC instance that covers all edges in $\varE$.
  As mentioned earlier, we can assume $A_1\in\varS$.
  For each $A_i\in\varS$, we choose $W^2$ seeds from the bundle $B_i$, so a total of $\vark W^2=k$ seeds are chosen.

  We show that all vertices in the level $B_{\cdot11}$ will be infected.
  Consider $e_1=(A_{i_1},A_{i_1'})$ with $i_1<i_1'$.
  By the fact the \VC instance is a \yes instance and the way we choose the seeds, $W^2$ vertices in either $B_{i_1}$ or $B_{i_1'}$, or both, are seeded.
  Assume without loss of generality that $W^2$ vertices from $B_{i_1}$ are seeded, then all the vertices in the bundle $B_{i_111}$, having threshold $\theta_{i_111}=\omega_{11}+(1-\Delta)W^2+0=(1-\Delta)W^2<W^2$ will be infected.
  As for the vertices in $B_{i_1'11}$, they will be infected in the same way if $W^2$ vertices from $B_{i_1'}$ are also seeded.
  On the other hand, if no vertex in $B_{i_1'}$ is seeded, all vertices in $B_{i_1'11}$ will be infected due to the influence of $B_{i_111}$.
  This is because 1) each vertex in $B_{i_1'11}$ has more than $(1-\delta)W^2$ infected neighbors in $B_{i_111}$ by requirement 1 of Definition~\ref{good_defi}, and 2) each vertex in $B_{i_1'11}$ has threshold $(1-\Delta)W^2<(1-\delta)W^2$.
  In the next $n-2$ iterations, by a careful calculation and based on requirement 1 of Definition~\ref{good_defi}, all vertices in the remaining $n-2$ bundles $\{B_{i11}\}_{i\neq i_1,i_1'}$ will be infected in the following order:
  \begin{equation}\label{eqn:level1_yes}
  B_{111}\rightarrow B_{211}\rightarrow\cdots B_{(i_1-1)11}\rightarrow B_{(i_1+1)11}\rightarrow\cdots B_{(i_1'-1)11}\rightarrow B_{(i_1'+1)11}\rightarrow\cdots\rightarrow B_{n11}.
  \end{equation}
  Therefore, the entire level $B_{\cdot11}$ will be infected.

  By similar analysis, we will show that the next level $B_{\cdot12}$ will be infected after the previous level $B_{\cdot11}$.
  Again, assume without loss of generality that $W^2$ seeds in $B_{i_1}$ are chosen. (Remember that the first $n-2$ levels are for edge $e_1\in\varE$, so we are still working on $e_1$.)
  Each vertex in $B_{i_112}$ has $(W^2+W^3)$ infected neighbor in the $i_1$-th clique, and has more than $(n-1)(1-\delta)W^2$ infected neighbors in $\{B_{i11}\}_{i\neq i_1}$, which is a total of more than $W^3+nW^2-(n-1)\delta W^2$ neighbors.
  Moreover, each vertex in $B_{i_112}$ has threshold $\theta_{i_112}=\omega_{12}+(1-\Delta)W^2+0=W^3+(n-1)W^2+(1-\Delta)W^2=W^3+nW^2-\Delta W^2$ which is less than the number of infected neighbors, as $-\Delta<-(n-1)\delta$.
  Therefore, all vertices in $B_{i_112}$ will be infected.
  As for the vertices in $B_{i_1'12}$, following the analysis in the last paragraph, they will be infected at the same iteration if $W^2$ vertices in $B_{i_1'}$ are seeded, and they will be infected at the next iteration due to the extra influence from $B_{i_112}$ if not.
  Finally, the remaining $n-2$ bundles $\{B_{i12}\}_{i\neq i_1,i_1'}$ will be infected in the following order:
  \begin{equation}\label{eqn:level2_yes}
  B_{212}\rightarrow B_{312}\rightarrow\cdots B_{(i_1-1)12}\rightarrow B_{(i_1+1)12}\rightarrow\cdots B_{(i_1'-1)12}\rightarrow B_{(i_1'+1)12}\rightarrow\cdots B_{n12}\rightarrow B_{112},
  \end{equation}
  which is similar to (\ref{eqn:level1_yes}), but is cyclically shifted to the left by $1$ unit, due to our cyclic construction of the thresholds.
  Thus, we have shown that the level $B_{\cdot12}$ will be infected after the previous level $B_{\cdot11}$.

  Following the same analysis, we can conclude that all levels will be infected in the following order:
  $$B_{\cdot11}\rightarrow B_{\cdot12}\rightarrow\cdots\rightarrow B_{\cdot1(n-2)}\rightarrow$$
  $$B_{\cdot21}\rightarrow B_{\cdot22}\rightarrow\cdots\rightarrow B_{\cdot2(n-2)}\rightarrow$$
  $$\cdots$$
  $$B_{\cdot m1}\rightarrow B_{\cdot m2}\rightarrow\cdots\rightarrow B_{\cdot m(n-2)}.$$

  Lastly, each vertex $c\in C$ has $W^2+m(n-2)W^3$ infected neighbors in the $1$-st clique (notice that we assume $A_1\in\varS$, which implies $W^2$ vertices in $B_1$ are seeded, which contributes $W^2$ infected neighbors), and more than $(n-1)\cdot m(n-2)(1-\delta)W^2$ infected neighbors from the other $n-1$ cliques, which is a total of $m(n-2)W^3+(n-1)m(n-2)W^2+W^2-(n-1)m(n-2)\delta W^2$ neighbors.
  In addition, $c$ has threshold $\theta_{1(m+1)}=m(n-2)W^3+(n-1)m(n-2)W^2+(1-\Delta)W^2$, which is less than the number of infected neighbors, as we have $-\Delta=-mn^2\delta<-(n-1)m(n-2)\delta$.
  Consequently, all vertices in $C$ will be infected.
  By summing up the total number of infected vertices, we conclude the first part of this lemma.
\end{proof}

\begin{proof}[Proof of (2)]
  Suppose the \VC instance is a \no instance.
  For those $n\vark W^2$ vertices in $\{B_i\}_{i=1,\ldots,n}$ having threshold $\infty$, they will not be infected unless being seeded, which means at least $(n-1)\vark W^2$ of them will not be infected.
  To show that the total number of infected vertices cannot exceed $N-M$, it is enough to show that at most $(n-1)\vark W^2$ vertices can be infected in the bundle $C$ of $M$ vertices.
  We will show the following stronger claim.
  \begin{proposition}\label{prop:no1}
    If the \VC instance is a \no instance, all vertices in $C$ will not be infected unless being seeded.
  \end{proposition}

  To show Proposition~\ref{prop:no1}, we show that the cascade will stop at an intermediate level.
  We will first identify this level, and then show this claim in Proportion~\ref{prop:no2}.

  Consider an arbitrary seed set $S$ (with $|S|=k$).
  Let $S_i$ be the seeds chosen from the $i$-th clique, and $k_i=|S_i|$ so that $\sum_{i=1}^nk_i=k$.
  We say that the $i$-th clique is activated if $k_i\geq(1-9\Delta)W^2$.
  Since
  $(\vark+1)(1-9\Delta)W^2=\vark W^2+\left(1-\frac9{10}-\frac{9}{10\vark}\right)W^2>k,$
  at most $\vark$ cliques can be activated.

  If we draw an analogy between activating a clique and picking a vertex in \VC, by the fact that the \VC instance is a \no instance, there exists $j^\ast$ where $e_{j^\ast}=(A_{i_{j^\ast}},A_{i_{j^\ast}'})$ such that both $i_{j^\ast}$-th and $i_{j^\ast}'$-th cliques are not activated.
  For the ease of illustration, assume without loss of generality that $i_{j^\ast}=n-1$ and $i_{j^\ast}'=n$.
  Since we have assumed $n>\vark^2+2$, there exists $\imath^\ast\leq n-1-\vark$ such that the $\imath^\ast$-th, the $(\imath^\ast+1)$-th, ..., and the $(\imath^\ast+\vark-1)$-th cliques are not activated.
  (If we have an activated clique within any $\vark$ consecutive cliques in the first $n-2$ cliques, the total number of activated cliques is at least $\frac{n-2}\vark>\vark$, which is a contradiction.)
  We will show that the cascade stops at the level $B_{\cdot j^\ast\imath^\ast}$.
  That is, there are only $o(W^3)$ infected vertices in
  $$\left(\bigcup_{(n-2)j+\imath\geq(n-2)j^\ast+\imath^\ast}B_{\cdot j\imath}\right)\cup C=V\setminus\left(B_1\cup\cdots\cup B_n\cup B_{\prec j^\ast\imath^\ast}\right).$$
  We will show that this is true even in the case that all vertices in the previous $(n-2)(j^\ast-1)+\imath^\ast-1$ levels (i.e., those in $B_{\prec j^\ast\imath^\ast}$) are infected.
  \begin{proposition}\label{prop:no2}
    There are only $o(W^3)$ infected vertices in the level $B_{\cdot j^\ast\imath^\ast}$, given that all vertices in $B_{\prec j^\ast\imath^\ast}$ and at most $\vark W^2$ vertices elsewhere (i.e., in $V\setminus B_{\prec j^\ast\imath^\ast}$) are infected.
  \end{proposition}
  The ``$\vark W^2$ vertices elsewhere'' mentioned in Proposition~\ref{prop:no2} refer to the $k=\vark W^2$ seeds.
  Notice that the seed-picker may choose the seeds outside $B_{\prec j^\ast\imath^\ast}$, and Proposition~\ref{prop:no2} holds even if all vertices in $B_{\prec j^\ast\imath^\ast}$ are infected and the $k$ seeds are all outside $B_{\prec j^\ast\imath^\ast}$.

  Before proving Proposition~\ref{prop:no2}, we remark that Proposition~\ref{prop:no2} immediately implies Proposition~\ref{prop:no1}: the vertices in the later levels $B_{\cdot j^\ast(\imath^\ast+1)},B_{\cdot j^\ast(\imath^\ast+2)},\ldots,B_{\cdot m(n-2)}$ have thresholds even higher than the thresholds of vertices in $B_{\cdot j^\ast\imath^\ast}$, and the thresholds increase by $\Theta(W^3)$ for each next level.

  Proposition~\ref{prop:no2} can be proved by just a sequence of calculations.

\begin{proof}[Proof of Proposition~\ref{prop:no2}]
  Suppose all vertices in $B_{\prec j^\ast\imath^\ast}$ and at most $\vark W^2$ vertices elsewhere are infected after a certain cascade iteration $t$.
  We will first show that less than $\delta W^2$ not-seeded vertices can be infected in each bundle $B_{ij^\ast\imath^\ast}$ for $i=1,\ldots,n$ in the next cascade iteration $t+1$.
  Specifically, we will show this separately for (i) the $2$ bundles $B_{nj^\ast\imath^\ast}$ and $B_{(n-1)j^\ast\imath^\ast}$, (ii) the $\vark$ bundles $B_{\imath^\ast j^\ast\imath^\ast},B_{(\imath^\ast+1)j^\ast\imath^\ast},\ldots,B_{(\imath^\ast+\vark-1)j^\ast\imath^\ast}$, and (iii) the remaining $n-2-\vark$ bundles.
  Then, we will show the same claim for later iterations.

  (i) For each vertex in the bundle $B_{nj^\ast\imath^\ast}$, by requirement 1 of Definition~\ref{good_defi}, the number of infected neighbors among the vertices in $B_{\prec j^\ast\imath^\ast}$ is less than
  \begin{equation}\label{eqn:prop2a}
   \underbrace{((n-2)(j^\ast-1)+\imath^\ast-1)W^3}_{\text{from the }n\text{-th clique}}+\underbrace{(n-1)\cdot((n-2)(j^\ast-1)+\imath^\ast-1)\cdot(1+\delta)W^2}_{\text{from the other }n-1\text{ cliques}}<\omega_{j^\ast\imath^\ast}+\Delta W^2.
  \end{equation}

  For each vertex in the bundle $B_{nj^\ast\imath^\ast}$, we have already counted the number of infected neighbors in $B_{\prec j^\ast\imath^\ast}$.
  Next, we consider the infected neighbors in $V\setminus B_{\prec j^\ast\imath^\ast}$.
  There are at most $\vark W^2$ of them by our assumption, and they are the seeds $S=\bigcup_{i=1}^nS_n$.

  The number of infected neighbors among seed set $S_n$ contributes at most $k_n<(1-9\Delta)W^2$, as we have assumed the $n$-th clique is not activated.
  Summing up this and (\ref{eqn:prop2a}), the total number of infected neighbors in $B_{\prec j^\ast\imath^\ast}\cup S_n$ is at most $\omega_{j^\ast\imath^\ast}+(1-8\Delta)W^2$.
  Since by our construction $\theta_{nj^\ast\imath^\ast}=\omega_{j^\ast\imath^\ast}+(1-\Delta)W^2+0$, to have $\delta W^2$ not-seeded vertices infected, the number of edges between each of these $\delta W^2$ vertices and $\bigcup_{i=1}^{n-1}S_i$ should be more than
  $$7\Delta W^2,7\Delta W^2-1,7\Delta W^2-2,\ldots,7\Delta W^2-\delta W^2+1$$
  respectively.
  This requires a total of
  $$\sum_{t=0}^{\delta W^2-1}(7\Delta W^2-t)>\delta W^2(7\Delta W^2-\delta W^2+1)>W^{3.6}$$
  edges, where the last inequality is based on the fact $\delta W^2=\frac1{10\vark}m^{19}n^{18}\gg W^{1.6}$.
  Since $\sum_{i=1}^{n-1}k_i<(\vark+1)W^2$, this is a contradiction to requirement 2 of Definition~\ref{good_defi}.

  For exactly the same reason, we can only have less than $\delta W^2$ not-seeded vertices infected in the bundle $B_{(n-1)j^\ast\imath^\ast}$, as $\theta_{(n-1)j^\ast\imath^\ast}=\theta_{nj^\ast\imath^\ast}$.

  (ii) Next, we consider these $\vark$ bundles: $B_{\imath^\ast j^\ast\imath^\ast},B_{(\imath^\ast+1)j^\ast\imath^\ast},\ldots,B_{(\imath^\ast+\vark-1)j^\ast\imath^\ast}$, whose corresponding cliques $\imath^\ast,\imath^\ast+1,\ldots,\imath^\ast+\vark-1$ are not activated by our assumption.
  Based on our construction, the vertices in these bundles have thresholds
  $$\omega_{j^\ast\imath^\ast}+(1-\Delta)W^2+1W^2,\omega_{j^\ast\imath^\ast}+(1-\Delta)W^2+2W^2,\ldots,\omega_{j^\ast\imath^\ast}+(1-\Delta)W^2+\vark W^2$$
  respectively, which are all more than $\theta_{nj^\ast\imath^\ast}$.
  By the same arguments, we can show that having $\delta W^2$ not-seeded vertices infected in any of these bundles requires even more edges, which contradicts requirement 2 of Definition~\ref{good_defi}.

  (iii) For each of the remaining $n-2-\vark$ bundles $B_{ij^\ast\imath^\ast}$ with $i\neq\imath,\imath+1,\ldots,\imath+\vark-1,n-1,n$, although the corresponding $i$-th clique may be activated, the threshold $\theta_{ij^\ast\imath^\ast}$ is at least $\omega_{j^\ast\imath^\ast}+(1-\Delta)W^2+(\vark+1)W^2$.
  The number of seeds chosen in the $i$-th clique $k_i$ cannot offset the term $(\vark+1)W^2$.
  Therefore, applying the same arguments shows us that less than $\delta W^2$ not-seeded vertices can be infected in each of these bundles.

  We have shown that less than $\delta W^2$ not-seeded vertices can be infected in each bundle $B_{ij^\ast\imath^\ast}$ in iteration $t+1$.
  To show this claim for future iterations, assume for the sake of contradiction that 1) at iteration $t^\ast>t+1$, less than $\delta W^2$ not-seeded vertices are infected in each bundle $B_{ij^\ast\imath^\ast}$, and 2) at iteration $t^\ast+1$, for certain $i^\ast$ we have at least $\delta W^2$ not-seeded vertices infected in the bundle $B_{i^\ast j^\ast\imath^\ast}$.
  Denote by $D_{-i^\ast}$ the set of those vertices outside the $i$-th clique which are infected during the iterations $t+1,t+2,\ldots,t^\ast$, and $D_{i^\ast}$ be the set of those vertices in the $i$-th clique which are infected during the iterations $t+1,t+2,\ldots,t^\ast,t^\ast+1$.
  Following the same arguments, for some $\delta W^2$ vertices from $D_{i^\ast}$, the number of edges between each of these $\delta W^2$ vertices and $D_{-i^\ast}\cup S$ should be more than
  $$7\Delta W^2,7\Delta W^2-1,7\Delta W^2-2,\ldots,7\Delta W^2-\delta W^2+1$$
  respectively, whose summation is more than $W^{3.6}$.
  On the other hand, since $|D_{-i^\ast}|<(n-1)\cdot\delta W^2<W^2$, we have $|D_{-i^\ast}\cup S|<(1+\vark)W^2$, which again contradicts to requirement 2 of Definition~\ref{good_defi}.
  Therefore, we conclude Proposition~\ref{prop:no2}.
\end{proof}
As we have remarked that Proposition~\ref{prop:no2} implies Proposition~\ref{prop:no1}, we conclude the second part of Lemma~\ref{yesno_SHBGM}.
\end{proof}

Finally, by making $M$ sufficiently large, both Theorem~\ref{hardnessSHBGM_pre} and Theorem~\ref{hardnessSHBGM_post} follow from Lemma~\ref{good_lemma} and Lemma~\ref{yesno_SHBGM}.

\section{Hierarchical Blockmodel with One-Way Influence}
\label{sect::oneway}
In this section, we consider a variant to the hierarchical blockmodel in which the influence between any two vertex-blocks can only be ``one-way''.
To each node in the hierarchy tree, a \emph{sign} is assigned deciding the directions of the edges between the two vertex-blocks associated to its two children.
For example, let $t$ be a node in the hierarchy tree, and $t_L,t_R$ be its left child and right child respectively.
If $t$ has a positive sign, then all edges between $V(t_L)$ and $V(t_R)$ are from $V(t_L)$ to $V(t_R)$; otherwise, these edges are from $V(t_R)$ to $V(t_L)$.
In this manner, the influence between $V(t_L)$ and $V(t_R)$ is one-way.

In \infmax, the seed-picker needs to decide not only the choice of those $k$ seeds, but also the sign at each tree node.
That is, the algorithm to \infmax problem should also output the optimal directions of influence between each pair of vertex-blocks.

Our algorithm also works in the more restrictive, but, perhaps, more practical setting where the signs for all tree nodes are fixed as input and the seed-picker only needs to decide the choice of $k$ seeds.
The directed influence between two communities may be observed in our real life for multiple reasons.
In some scenarios (e.g.,  Twitter), the network itself is directed.  Status differences between members of different communities could create a uniform direction of influence.  
Another reason of directed influence may be government regulations.
For example, in the cellphone market, many Chinese users adopt iPhone products due to the influence of American users, while Huawei cellphones, adopted by many Chinese users, are banned in the United States of America.

\subsection{A Dynamic Programming Algorithm}
We present a dynamic-programming-based algorithm for \infmax for this variant of the hierarchical blockmodel, when the thresholds of the vertices are deterministic.
Our algorithm makes use of the following observation:
\emph{for a tree node $t$, the influence from the infected vertices in the vertex-block $V(t)$ to each vertex in $V\setminus V(t)$ only depends on the \textbf{number} of infected vertices in $V(t)$.}
This is formally described in Definition~\ref{defi:oneway} and Lemma~\ref{lem:oneway} below.

\begin{definition}\label{defi:oneway}
  Given a set $I\subseteq V$ of infected vertices and a vertex $v\in V\setminus I$, the \emph{influence} from $I$ to $v$ is defined by $\sum_{u\in I}w(u,v)$, where $w(u,v)$ is the weight of the edge $(u,v)$ which is the weight of the deepest node $t\in V_T$ such that $V(t)$ contains both $u$ and $v$.
\end{definition}
By our definition, if the influence from the set of all infected vertices to an uninfected vertex $v$ exceeds $\theta_v$, $v$ will be infected.
\begin{lemma}\label{lem:oneway}
  Consider an arbitrary node $t\in V_T$. The influence from a set of infected vertices $I_1\subseteq V(t)$ in $V(t)$ to a vertex $u\in V\setminus V(t)$ only depends on $|I_1|$. Moreover, for any $v_1,v_2\in V(t)$ and an arbitrary set of infected vertices outside $V(t)$, $I_2\subseteq V\setminus V(t)$, the influences from $I_2$ to $v_1$ and $v_2$ are the same.
\end{lemma}
\begin{proof}
For any $v_1,v_2\in V(t)$ and $u\in V\setminus V(t)$, let $t_{v_1},t_{v_2},t_u$ be the leaves such that $v_1\in V(t_{v_1})$, $v_2\in V(t_{v_2})$ and $u\in V(t_u)$.
The least common ancestor of $t_{v_1}$ and $t_u$ is the same as the least common ancestor of $t_{v_2}$ and $t_u$, which is the least common ancestor of $t$ and $t_u$.
This implies that the edges $(v_1,u)$ and $(v_2,u)$ have the same weight, and the lemma follows easily from this observation.
\end{proof}

For each tree node $t\in V_T$, each $i=1,\ldots,k$, and each $\nu=0,1,\ldots,|V|$, define $H[t,i,\nu]$ be the smallest positive real number $\gamma$ satisfying the following:
\begin{itemize}
  \item given that the threshold of each vertex is updated to $\theta_v\leftarrow\theta_v-\gamma$, where we assume the vertex with $\theta_v-\gamma\leq0$ is infected immediately, we can choose $i$ seeds in $V(t)$ such that at least $\nu$ vertices in $V(t)$ will be infected (due to the influence of these $i$ seeds).
\end{itemize}
Intuitively, this means we can infect $\nu$ vertices by $i$ seeds, given that the influence from infected vertices outside $V(t)$ is $H[t,i,\nu]$.
Correspondingly, let $\Sigma[t,i,\nu]\subseteq V(t)$ store the seeding strategy that allocate $i$ seeds in $V(t)$ such that, given that the influence from certain set of infected vertices in $V\setminus V(t)$ to each vertex in $V(t)$ is $H[t,i,\nu]$, those $i$ seeds infect at least $\nu$ vertices in $V(t)$.

If $t$ is a leaf, the subgraph induced by $V(t)$ is a clique in which all the $|V(t)|(|V(t)|-1)$ edges have the equal weight.
Obviously, the optimal strategy is to place the $i$ seeds on those vertices with the highest thresholds.
We propose Algorithm~\ref{alg:init} to calculate $\Sigma[t,i,\nu]$ and $H[t,i,\nu]$ for each leaf $t$.
\begin{algorithm}[ht!]
\KwIn{vertex set $V(t)$, weight of each edge $w(t)$, threshold set $\{\theta_v\}_{v\in V(t)}$, integers $i,\nu$}%
\KwOut{$\Sigma[t,i,\nu]$ and $H[t,i,\nu]$ for leaf $t$} %
set $\Sigma[t,i,v]$ be the $i$ vertices in $V(t)$ having the highest thresholds (set $\Sigma[t,i,v]=V(t)$ if $i\geq |V(t)|$)\;
\For{each vertex $v\in V(t)$} {
    update $\theta_v\leftarrow \theta_v-i\cdot w(t)$
}
\eIf{$\nu\leq|\{\theta_v:\theta_v\leq0\}|+i$}{
    set $H[t,i,\nu]=0$
}{
    set $H[t,i,\nu]$ be the $(\nu-i)$-th smallest threshold in $\{\theta_v\}_{v\in V(t)}$
}
\Return{$\Sigma[t,i,\nu]$ and $H[t,i,\nu]$}
\caption{Initialization for a Leaf $t$}\label{alg:init}
\end{algorithm}

If $t$ is not a leaf, we aim to find a recurrence between $H[t,i,\nu]$ and $H[t_L,i_L,\nu_L],H[t_R,i_R,\nu_R]$.
Suppose the sign of $t$ is positive, and there are $\nu_L$ infected vertices in $V(t_L)$.
Their influence to $V(t_R)$ is $\nu_L\cdot w(t)$ where $w(t)$ is the weight of $t$ reflecting the weight of all edges from $V(t_L)$ to $V(t_R)$.
We have a similar observation in the case that the sign of $t$ is negative.

By considering all decompositions $i=i_L+i_R$ and $\nu=\nu_L+\nu_R$, if the sign of $t$ is positive, we have
\begin{equation}\label{eqn:oneway_plus}
H^+[t,i,\nu]=\min_{i_L=0,\ldots,i;\quad\nu_L=0,\ldots,\nu}\Big\{\max\big(H[t_L,i_L,\nu_L],H[t_R,i-i_L,\nu-\nu_L]-\nu_L\cdot w(t)\big)\Big\};
\end{equation}
if the sign of $t$ is negative, we have
\begin{equation}\label{eqn:oneway_minus}
H^-[t,i,\nu]=\min_{i_R=0,\ldots,i;\quad\nu_R=0,\ldots,\nu}\Big\{\max\big(H[t_L,i-i_R,\nu-\nu_R]-\nu_R\cdot w(t),H[t_R,i_R,\nu_R]\big)\Big\},
\end{equation}
where we set $H[t,i,\nu]=\infty$ if $\nu>|V(t)|$.
Finally, we decide the sign of $t$:
\begin{equation}\label{eqn:oneway_sign}
H[t,i,\nu]=\min\big(H^+[t,i,\nu],H^-[t,i,\nu]\big).
\end{equation}

The recurrence between $\Sigma[t,i,v]$ and $\Sigma[t_L,i_L,\nu_L],\Sigma[t_R,i_R,\nu_R]$ can be obtained in a natural way.
The sign of $t$, $\sign(t)\in\{+,-\}$, is defined naturally by (\ref{eqn:oneway_sign}).
If $\sign(t)=+$, we have $\Sigma[t,i,\nu]=\Sigma[t_L,i_L^\ast,\nu_L^\ast]\cup\Sigma[t_R,i-i_L^\ast,\nu-\nu_L^\ast]$, where $(i_L^\ast,\nu_L^\ast)$ is the minimizer for (\ref{eqn:oneway_plus}); if $\sign(t)=-$, we have $\Sigma[t,i,\nu]=\Sigma[t_L,i-i_R^\ast,\nu-\nu_R^\ast]\cup\Sigma[t_R,i_R^\ast,\nu_R^\ast]$, where $(i_R^\ast,\nu_R^\ast)$ is the minimizer for (\ref{eqn:oneway_minus}).

Define the \emph{height} of $t \in V_T$ be the length of the path to $t$'s deepest descendant.
The following Algorithm~\ref{alg:dp} solves \infmax for the hierarchical blockmodel with one-way influence.
\begin{algorithm}
\KwIn{hierarchical blockmodel $G=(V,T)$, threshold set $\{\theta_v\}_{v\in V}$, integer $k$} 
\KwOut{1) $S \subseteq V$ such that $|S| = k$ and $S$ maximizes $\sigma(S)$, and 2) the sign of each internal node $t$: $\sign(t)$} 
\For{each height $i = 0, 1, \ldots, h$} {
  \For{each node $t \in V_T$ with height $i$} {
    \eIf{$t$ is a leaf}{
      initialize $\Sigma[t,i,\nu]$ and $H[t,i,\nu]$ by Algorithm~\ref{alg:init} for all $i=0,1,\ldots,k$ and $\nu=0,1\ldots,N$
    }{
      \For{for each $i=0,1,\ldots,k$ and $\nu=0,1\ldots,N$}{
        $\displaystyle H^+[t,i,\nu]=\min_{i_L=0,\ldots,i;\nu_L=0,\ldots,\nu}\Big\{\max\big(H[t_L,i_L,\nu_L],H[t_R,i-i_L,\nu-\nu_L]-\nu_L\cdot w(t)\big)\Big\}$\;
        $\displaystyle H^-[t,i,\nu]=\min_{i_R=0,\ldots,i;\nu_R=0,\ldots,\nu}\Big\{\max\big(H[t_L,i-i_R,\nu-\nu_R]-\nu_R\cdot w(t),H[t_R,i_R,\nu_R]\big)\Big\}$\;
        $\displaystyle H[t,i,\nu]=\min\big(H^+[t,i,\nu],H^-[t,i,\nu]\big)$\;
        set $\displaystyle\sign(t)=\argmin_{s\in\{+,-\}}H^s[t,i,\nu]$\;
        \eIf{$\sign(t)=+$}{
            set $\Sigma[t,i,\nu]=\Sigma[t_L,i_L^\ast,\nu_L^\ast]\cup\Sigma[t_R,i-i_L^\ast,\nu-\nu_L^\ast]$, where $(i_L^\ast,\nu_L^\ast)$ minimizes $H^+[t,i,\nu]$
        }{
            set $\Sigma[t,i,\nu]=\Sigma[t_L,i-i_R^\ast,\nu-\nu_R^\ast]\cup\Sigma[t_R,i_R^\ast,\nu_R^\ast]$, where $(i_R^\ast,\nu_R^\ast)$ minimizes $H^-[t,i,\nu]$
        }
      }
    }
  }
}
set $\nu^\ast$ be the maximum $\nu$ such that $H[r,k,\nu]=0$, where $r$ is the root of $T$\;
\Return{$\Sigma[r,k,\nu^\ast]$ and $\sign(t)$ for each internal node $t$}
\caption{Dynamic Programming Algorithm for Hierarchical Blockmodel \infmax with One-Way Influence}\label{alg:dp}
\end{algorithm}
It is straightforward to check that Algorithm~\ref{alg:dp} runs in time $O\left(N^3k^2\right)$.

\begin{remark}\label{remark:dp}
Algorithm~\ref{alg:dp} can be easily adapted to the variant of the \infmax problem where each $\sign(t)$ is fixed as input (instead of being a part of the output).
Instead of computing both $H^+[t,i,\nu]$ and $H^-[t,i,\nu]$, and setting $H[t,i,\nu]=\min\big(H^+[t,i,\nu],H^-[t,i,\nu]\big)$, we only need to have $H[t,i,\nu]=H^{\sign(t)}[t,i,\nu]$. Corresponding, we have either $\Sigma[t,i,\nu]=\Sigma[t_L,i_L^\ast,\nu_L^\ast]\cup\Sigma[t_R,i-i_L^\ast,\nu-\nu_L^\ast]$ or  $\Sigma[t,i,\nu]=\Sigma[t_L,i-i_R^\ast,\nu-\nu_R^\ast]\cup\Sigma[t_R,i_R^\ast,\nu_R^\ast]$ depending on $\sign(t)$ which is now given by the input.
\end{remark}

\subsection{Further Discussions}
We have seen inapproximability results in Section~\ref{sect::resultHBGM} and Section~\ref{sect::resultSHBGM} for \infmax on the (stochastic) hierarchical blockmodel.
Our algorithm in this section reveals the intrinsic reason why these problems are difficult.

In the hard \infmax instances in Figure~\ref{fig:reductionHBGM} and Figure~\ref{fig:reductionSHBGM}, we constructed the hierarchy tree by creating $n$ branches corresponding to the $n$ vertices in \VC.
In the case the \VC instance is a \yes instance, the influence of the properly chosen seeds passes through these $n$ branches ``back-and-forth'' frequently:
the infected vertices in branch $A_i$ make vertices in branch $A_j$ infected, while these newly infected vertices in $A_j$ may have backward influence to $A_i$, and cause more infected vertices in $A_i$.
This bidirectional effect is not considered in Algorithm~\ref{alg:dp}, and is exactly why \infmax is hard.
On the other hand, when there is no such bidirectional effect, even if the algorithm needs to decide the optimal directions at all internal nodes (with exponentially many choices $2^{\Theta(|V_T|)}$), \infmax becomes easy on the hierarchial blockmodel, as our algorithm in this section suggests.

As mentioned in the related work section, Angell and Schoenebeck \cite{angell2016don} show that a generalization of this algorithm works well empirically.
This perhaps indicates that the bidirectional influence is, in the average case, not often so important in realistic settings.

\section{2-Quasi-Submodular Influence Maximization}
\label{sect:hardness_2submod}
We prove the following theorem in this section which says that, for any fixed 2-quasi-submodular $f$, there exists a constant $\tau$ depending on $f$ such that \infmax with the universal local influence model $I^G_f$ is NP-hard to approximate to within factor $N^\tau$, where $N$ is the number of vertices of the graph.


\begin{theorem}\label{hardnessproof}
  Consider the \infmax problem with the universal local influence model $I^G_f$ for any fixed 2-quasi-submodular $f$. There exists a constant $\tau$ depending on $f$ such that it is NP-hard to distinguish between the following two cases:
  \begin{itemize}
      \item \yes: there exists a seed set $S$ with $|S|=k$ such that $\sigma_f^G(S)=\Theta(N)$;
      \item \no: for any seed set $S$ with $|S|=k$, we have $\sigma_f^G(S)=O(N^{1-\tau})$.
  \end{itemize}
\end{theorem}

The sequence notation $(a_i)_{i=0,1,2\ldots}$ is used to represent $f$ in this section.
Because $f$ is 2-quasi-submodular, we have $a_0=0$ and $a_2>2a_1$.
We denote $p^\ast=\lim_{i\rightarrow\infty}a_i$, which exists because $(a_i)$ is increasing and bounded by $1$ (see Definition~\ref{defi:ulim} and Definition~\ref{defi:2quasisubmodular}).
We consider two cases: $a_1>0$ and $a_1=0$.
We note that we have $a_2>0$ by the 2-quasi-submodular assumption.
In the case $a_1>0$, we will first assume the graph is directed, and later we will show that this assumption is not essential.

The remaining part of this section is organized as follows: Sect.~\ref{sect:proofsketch} provides a sketch of the proof of Theorem~\ref{hardnessproof} for the case $a_1>0$, with arguments presented in an intuitive level,
Sect.~\ref{sect:proof1} to Sect.~\ref{sect:directededgegadgets} prove the theorem rigorously for the case $a_1>0$, and Sect.~\ref{sect:proof3} prove the theorem rigorously for the case $a_1=0$.
Finally, in a similar style to the result in \cite{li2017influence}, we prove a variant of Theorem~\ref{hardnessproof} in Appendix saying that the inapproximability also holds if only $N^\gamma$ (for some fixed $\gamma\in(0,1)$) vertices admit the fixed 2-quasi-submodular function $f$ while the remaining vertices admit certain fixed non-zero submodular function $g$.

\subsection{Proof Sketch of Theorem~\ref{hardnessproof} for $a_1>0$}
\label{sect:proofsketch}
We prove the theorem by a reduction from the \textsc{SetCover} problem.
\begin{definition}\label{defi:SetCover}
  Given a universe $U$ of $n$ elements, a set of $K$ subsets $A=\{A_i\mid A_i\subseteq U\}$, and a positive integer $k$, the \textsc{SetCover} problem asks if we can choose $k$ subsets $\{A_{i_1},\ldots,A_{i_k}\}\subseteq A$ such that $A_{i_1}\cup\cdots\cup A_{i_k}=U$.
\end{definition}

We construct a graph $G$ which consists of two parts: the set cover part and the verification part,
where the set cover part simulates \textsc{SetCover} and the verification part verifies if all the elements in the \textsc{SetCover} instance are covered.
The construction is shown in Fig.~\ref{fig:highlevel_short}.
We first assume that the graph $G$ is directed, and then we show that this assumption is not essential by constructing a \emph{directed edge gadget} to simulate directed edges.

\begin{figure}
\centering
  \includegraphics[width=0.7\textwidth]{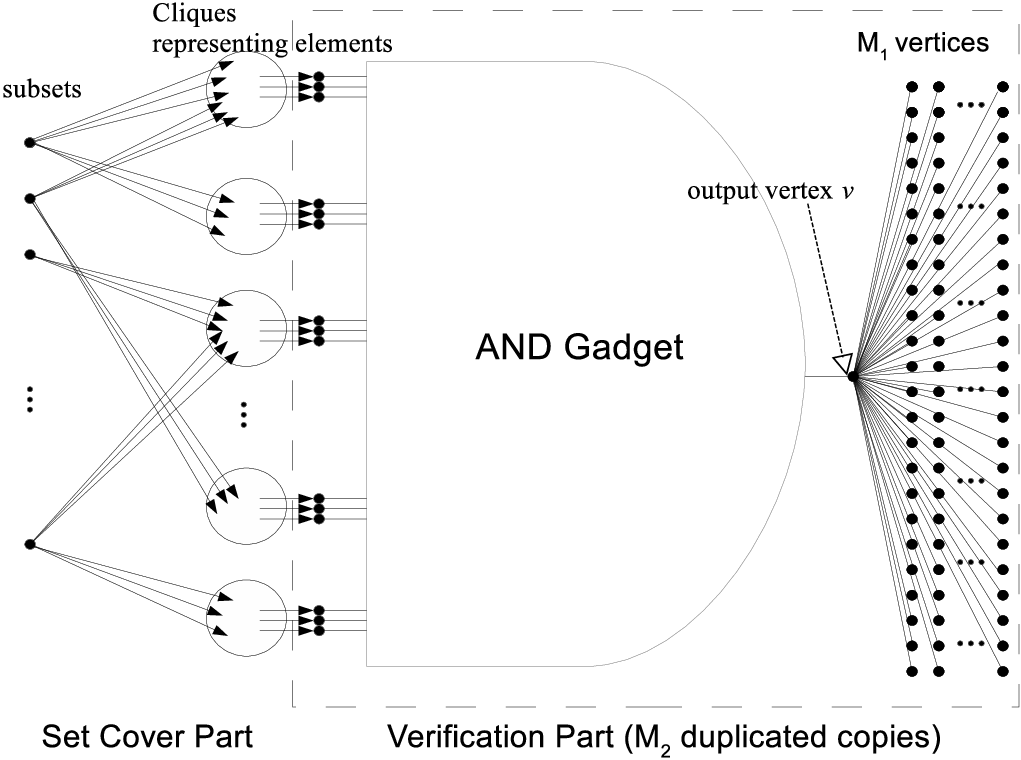}
  \caption{The high-level structure of the reduction for the proof of Theorem~\ref{hardnessproof}}
  \label{fig:highlevel_short}
\end{figure}

Given a \textsc{SetCover} instance, in the set cover part, we use a single vertex to represent a subset $A_i$ and a clique of size $m$ to represent each element in $U$.
If an element is in a subset, we create $m$ directed edges from the vertex representing the subset to each the $m$ vertices in the clique representing the element.
If a vertex representing a subset is picked, then all vertices in the cliques corresponding to the elements contained in this subset will be infected with probability close to $p^\ast$, by choosing $m$ large enough.
We call such cliques as being activated.
In a \yes instance of \textsc{SetCover}, we can choose $k$ seeds such that all cliques are activated.

In the verification part, we construct a \emph{AND gadget}, simulating the logical AND operation, to verify if all the cliques are activated.
The AND gadget takes $n$ inputs, each of which is a set of vertices from each of the $n$ cliques.
The output of the AND gadget is a vertex $v$, such that it will only be infected with a positive constant probability if all the $n$ cliques are activated.

We connect the output vertex $v$ of this AND gadget to a huge bundle of $M_1$ vertices, such that a constant fraction of those $M_1$ vertices will be infected only if all the cliques are activated (which corresponds to the case the \textsc{SetCover} is a \yes instance).
By making $M_1$ large enough, we can achieve a hardness of approximation ratio $N^\tau$.
To avoid the seed-picker bypassing the set cover game by directed seeding the output vertex $v$, we duplicate the verification part by $M_2$ times for some sufficiently large $M_2$.

Finally, we replace all directed edges in Fig.~\ref{fig:highlevel_short} by directed edge gadgets, including those connecting the vertices representing subsets and the cliques representing elements, and those connecting the set cover part and the verification part.
To complete the proof of Theorem~\ref{hardnessproof}, we present the construction of the AND gadget and the directed edge gadget in the next few subsections.

\subsubsection{The Probability Filter Gadget}
In this section, we present the construction of a gadget called \emph{probability filter gadget}, which is the key component in the constructions of both AND gadget and directed edge gadgets mentioned above.

Given a set of vertices that will be infected with a same probability $x$, the probability filter gadget tests if $x$ is larger than certain threshold $p_1$.
It outputs a vertex infected with probability almost $0$ if $x<p_1$, and with certain non-negligible probability $p_2$ if $x>p_1$.

\paragraph{The probability scaling down gadget}
Firstly, we need to construct the \emph{probability scaling down gadget} which takes a vertex $u$ with infection probability $p_u$ as input, and output a vertex $v$ such that $v$ is infected with probability $p_v=\alpha p_u$, where $\alpha\leq p^\ast$ is an adjustable parameter.
The construction of this gadget is shown in Fig.~\ref{fig:scalingdown_short}: we add many paths of different lengths from $u$ to $v$, and we can achieve $p_v=\alpha p_u$ by adjusting the number of paths and the length of each path.

\begin{figure}
\centering
\begin{minipage}{\textwidth}
  \centering
  \includegraphics[width=0.5\textwidth]{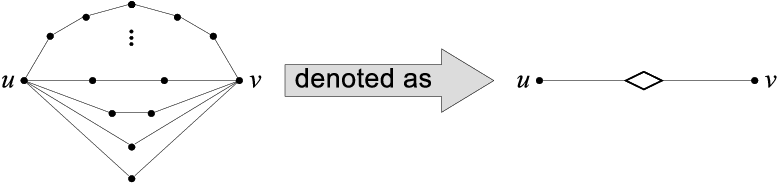}
\captionof{figure}{The probability scaling down gadget}
\label{fig:scalingdown_short}
\end{minipage}
\begin{minipage}{.5\textwidth}
  \centering
  \includegraphics[width=.8\linewidth]{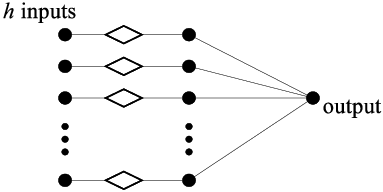}
  \captionof{figure}{The probability separation block}
  \label{fig:filter_short}
\end{minipage}%
\begin{minipage}{.5\textwidth}
  \centering
  \includegraphics[width=.45\linewidth]{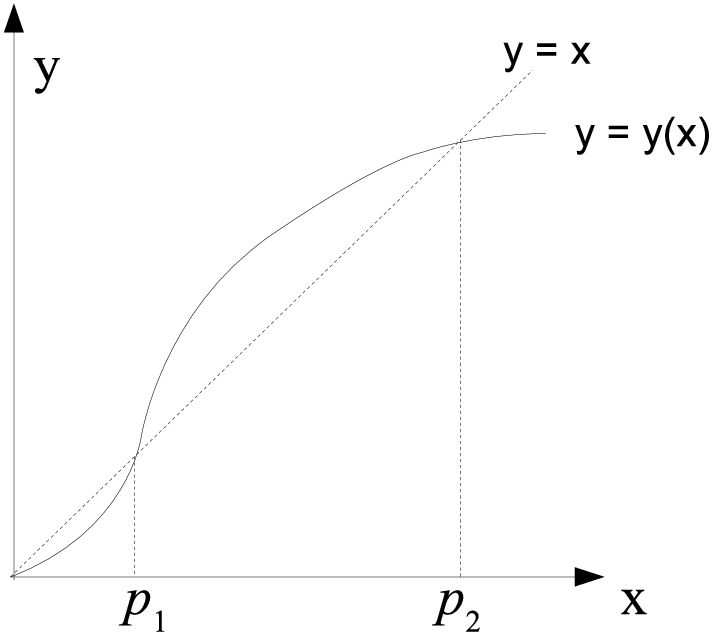}
  \captionof{figure}{The output probability $y$ versus the input probability $x$}
  \label{fig:filtercurve_short}
\end{minipage}
\end{figure}

\paragraph{The probability separation block}
Next, we construct a \emph{probability separation block}, which is the building block to the probability filter gadget.
The probability separation block takes $h$ vertices as input and outputs one vertex such that
\begin{enumerate}
  \item if each input is infected independently with a same probability that is greater than certain threshold $p_1$, then the output vertex will be infected with a slightly higher probability;
  \item if each input is infected independently with a same probability that is less than $p_1$, then the output vertex will be infected with a slightly lower probability.
\end{enumerate}

The construction of the probability separation block is shown in Fig.~\ref{fig:filter_short}, in which the $h$ inputs' infection probabilities are scaled down by a certain factor $\alpha$ by the probability scaling down gadgets, and then they are connected to the output vertex.
It is exactly the 2-quasi-submodularity of $f$ which enables us to adjust the two parameters $h$ and $\alpha$ such that (1) and (2) above hold.

Suppose each of the $h$ vertices in the input are infected with probability $x$,
and let $y=y(x)$ be the probability that the output vertex is infected.
We claim that we can tune the values of $\alpha$ and $h$ such that the graph of $y(x)$ looks like Fig.~\ref{fig:filtercurve_short}.

By considering the number of infected neighbors of the output vertex, we have
$$
y=\sum_{i=1}^h\binom{h}{i}a_i(\alpha x)^i(1-\alpha x)^{h-i},
$$
which is
$y=h\alpha a_1x+\frac{h(h-1)}2\alpha^2(a_2-2a_1)x^2+o(x^2)$
for sufficiently small $x$.
Choosing a sufficiently small constant $\delta>0$ and choosing $\alpha,h$ to satisfy $h\alpha a_1=1-\delta$, we have
$$y-x=-\delta x+\frac{h(h-1)}2\alpha^2(a_2-2a_1)x^2+o(x^2).$$
Since $y-x=-\delta x+o(x)$, we can see that $y<x$ for small enough $x$.

On the other hand, for sufficiently large $h$ and sufficiently small $\delta$ (and adjusting $\alpha$ such that $h\alpha a_1=1-\delta$ still holds), the second order derivative of $y-x$, which is $$\frac{d^2(y-x)}{dx^2}=h(h-1)\alpha^2(a_2-2a_1)+o(1)\approx\frac1{a_1^2}(a_2-2a_1)>0,$$ can be considerably more significant than its first order derivative $-\delta$.
Therefore, $y-x$, starting from $0$ at $x=0$ and being negative for very small $x$, will soon become positive after $x$ increases.
This proves our claim.
Notice that the 2-quasi-submodularity of $f$ makes sure $a_2-2a_1>0$.


\paragraph{The probability filter gadget}
The probability filter gadget consists of $\ell$ layers such that the $i$-th layer consists of $h^{\ell-i}$ probability separation blocks,
where the output vertices of every $h$ probability separation blocks in the $i$-th layer are the inputs of a single probability separation block in the $(i+1)$-th layer.
Because there are $h^{\ell-1}$ probability separation blocks in the first layer,
the probability filter gadget takes $\Lambda=h^\ell$ vertices as input.
The probability filter gadget outputs a single vertex after $\ell$ layers.

From Fig.~\ref{fig:filtercurve_short}, if we make $\ell$ large enough, we conclude that the probability filter gadget does the following job, which tests if the input vertices are infected with a probability larger than the threshold value $p_1$.
\begin{enumerate}
  \item if each vertex in the $\Lambda$ inputs is infected independently with a same probability less than $p_1$, then the vertex on the output end will be infected with a probability close to $0$;
  \item if each vertex in the $\Lambda$ inputs is infected independently with a same probability in $(p_1,p_2]$, then the vertex on the output end will be infected with a probability close to $p_2$.
\end{enumerate}


\subsubsection{The AND Gadget and the Directed Edge Gadget}
Both the AND gadget and the directed edge gadget can be constructed by using a single probability filter gadget as the core.

\begin{figure}
\centering
\begin{minipage}{.42\textwidth}
  \centering
  \includegraphics[width=\linewidth]{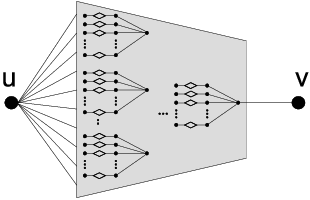}
  \captionof{figure}{The directed edge gadget $\langle u,v\rangle$}
  \label{fig:directededgegadget_short}
\end{minipage}
\begin{minipage}{.45\textwidth}
  \centering
  \includegraphics[width=\linewidth]{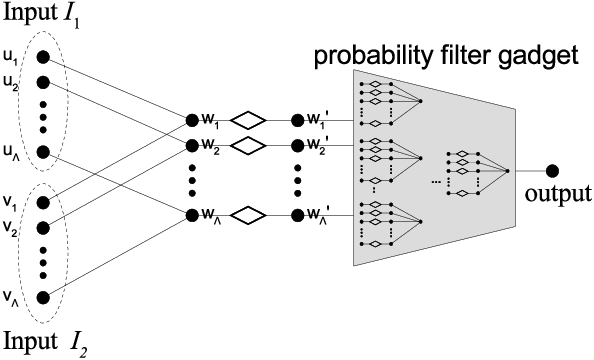}
  \captionof{figure}{The AND gadget with two inputs}
  \label{fig:AND_2inputs}
\end{minipage}
\end{figure}

\paragraph{The directed edge gadget}
The construction of the directed edge gadget $\langle u,v\rangle$ is shown in Fig.~\ref{fig:directededgegadget_short}.
It uses a single probability filter gadget, whose input vertices are connected to $u$, and whose output vertex is connected to $v$.
By adjusting $h$ and $\alpha$ making $p_1$ small enough\footnote{if we further check the calculations in the subsection where we construct the probability separation block, we can see that $p_1$ can be made arbitrarily small, by choosing small enough $\delta=1-h\alpha a_1$. Detailed calculations and justifications are in the later sections.}, we can make the output $v$ infected with noticeable probability (almost $p_2$) if $u$ is infected.
On the other hand, if $v$ is infected, then the expected number of infected vertices among those $h$ vertices on the input end of the top layer probability separation block is $h\alpha a_1=1-\delta<1$, which suggests that the cascade process will die out after a few layers from right to left.
In particular, the influence of $v$ cannot be passed to $u$.

\paragraph{The AND gadget}
The AND gadget in Fig.~\ref{fig:highlevel_short} takes $n$ sets of vertices as input.
It tests if all cliques are activated, that is, if each vertex in each input set is infected with probability almost $p^\ast$.

Here, we first construct a smaller AND gadget which only takes two input sets.
Let $I_1=\{u_1,u_2,\ldots,u_{\Lambda}\}$ and $I_2=\{v_1,v_2,\ldots,v_{\Lambda}\}$ be the two input sets.
The AND gadget should do the following:
\begin{enumerate}
  \item if each vertex in $I_1$ and $I_2$ is infected with probability $p^\ast$, the AND gadget outputs a vertex which is infected with a notable probability;
  \item if all vertices in at least one of $I_1,I_2$ are infected with probability $0$, the AND gadget outputs a vertex which is infected with a negligible probability.
\end{enumerate}

We create $\Lambda$ vertices $w_1,w_2,\ldots,w_{\Lambda}$ and create two edges $(u_i,w_i),(v_i,w_i)$ for each $i=1,2,\ldots,\Lambda$.
In case (1), each $w_i$ will be infected with probability $q_1=a_2(p^\ast)^2+2a_1p^\ast(1-p^\ast)$; in case (2), each $w_i$ will be infected with probability at most $q_2=a_1p^\ast$.
Obviously, $q_1>q_2$, and the AND gadget needs to ``amplify'' the gap between $q_1$ and $q_2$.

This naturally reminds us the probability filter gadget.
In particular, if the threshold $p_1$ of the probability filter gadget is in between: $q_1>p_1>q_2$, we can just make $\{w_1,\ldots,w_\Lambda\}$ the inputs of the probability filter gadget, and we are done.
However, by our discussion about probability separation block in the last subsection, $p_1$ is only guaranteed to exist, which may not be in $(q_2,q_1)$.
To settle this, we use probability scaling down gadgets to rescale the infection probability of $w_i$ such that $p_1$ will be in between after rescaling $q_1,q_2$.\footnote{It seems worrying that $q_1$ and $q_2$ may be both less than $p_1$, in which case the construction fails as we can only scale probabilities ``down''. However, as we have remarked, we can make $p_1$ arbitrarily small such that $p_1\ll q_2<q_1$}
Fig.~\ref{fig:AND_2inputs} shows the construction of this AND gadget.

To construct the AND gadget allowing $n$ input sets, we can use this AND gadget as a building block and construct an \emph{AND circuit} with $\log_2n$ levels of AND gadgets.
The last level contains a single AND gadget, whose output is connected to the $M_1$ vertices on the right-hand side of Fig.~\ref{fig:highlevel_short}.
For each AND gadget in Level $i$, its output become one input of a certain AND gadget in Level $i+1$.
The inputs of the AND gadgets in Level $1$ are exactly those associated to the $n$ cliques representing elements of the \VC instance.

We conclude the proof sketch here. In the remaining sections, we present the full proof of Theorem~\ref{hardnessproof} which realizes the intuitions and ideas in this section.

\subsection{Proof of Theorem~\ref{hardnessproof} for $a_1>0$ with Directed Graphs}
\label{sect:proof1}
We first define the following AND gadget which simulates the logical AND operation.
The construction of this AND gadget is deferred to Section~\ref{sect:othergadgets}---\ref{sect:ANDgadgets}.
We note that the nonsubmodularity property $a_2>2a_1$ plays an important role in the construction of the AND gadget.
In particular, the construction of the AND gadget uses a smaller gadget called the ``probability filter gadget'' as a building block (see Figure~\ref{fig:gadgetrelation}), and 2-quasi-submodularity is essential for constructing the probability filter gadget (refer to Section~\ref{subsect:pfg} for details).

\begin{definition}\label{defi:AND}
  An \emph{$(I, \Lambda, p_0, p_2, \varepsilon_1,\varepsilon_2, f)$-AND gadget} takes $I$ sets which each contains $\Lambda$ vertices as input, and outputs one vertex such that
  \begin{enumerate}
    \item if all the vertices in all $I$ sets are infected independently with probabilities less than $\frac{11}{10}p_0$, and moreover the infection probabilities of the vertices in at least one input set are less than $\frac12p_0$, then the output vertex will be infected with probability less than $\varepsilon_1$;
    \item if all the vertices in all $I$ sets are infected independently with probabilities in the interval $(p_0,\frac{11}{10}p_0)$, the output vertex will be infected with probability in $(p_2-\varepsilon_2,p_2]$,
  \end{enumerate}
\end{definition}

We remark that the choices for both factors of $p_0$ in 1 of the above definition, $\frac{11}{10}$ and $\frac12$, are only required to be close enough to $1$ and $0$ respectively.
We aim to simulate the case where at least one of the inputs is not ``active'' (being far from the threshold $p_0$) and the other ones are not ``too active'' (being at most somewhere around the threshold $p_0$), in which case the AND gadget outputs ``false'' (such that the output vertex is infected with negligible probability $\varepsilon_1$).

With the choice of the seven parameters satisfying the relation in the below lemma, we can construct the AND gadget.
\begin{lemma}\label{ANDlemma}
  Given any 2-quasi-submodular function $f$ with $a_1>0$, any constant threshold $p_0>0$ and any $I=2^\ell$ that is an integer power of $2$, there exists a constant $p_2>0$ depending on $p_0$ and $f$ such that for any $\varepsilon_1>0$ and any constant $\varepsilon_2>0$, we can construct an $(I, \Lambda, p_0, p_2, \varepsilon_1,\varepsilon_2, f)$-AND gadget with $\Lambda=O\left((1/\varepsilon_1)^{c_1}I^{c_2}\right)$,
  and the numbers of vertices and edges in this AND gadget are both $O\left((1/\varepsilon_1)^{c_1}I^{c_2+1}\right)$,
  where $c_1$ and $c_2$ are two constants.
\end{lemma}

The following lemma is needed in the next section for the proof of Theorem~\ref{hardnessproof} for undirected graphs.
\begin{lemma}\label{ANDlemma_p2}
  Given any 2-quasi-submodular function $f$ with $a_1>0$ and any $I=2^\ell$ that is an integer power of $2$, there exists $p_2>0$ such that for any $\varepsilon_1>0$ and any constant $\varepsilon_2>0$, we can construct an $(I, \Lambda, p^\ast(p_2-\varepsilon_2), p_2, \varepsilon_1,\varepsilon_2, f)$-AND gadget.
  We have $\Lambda=O\left((1/\varepsilon_1)^{c_1}I^{c_2}\right)$ and the AND gadget contains $O\left((1/\varepsilon_1)^{c_1}I^{c_2+1}\right)$ vertices and $O\left((1/\varepsilon_1)^{c_1}I^{c_2+1}\right)$ edges,
  where $c_1$ and $c_2$ are two constants.
\end{lemma}

Notice that Lemma~\ref{ANDlemma} does not imply Lemma~\ref{ANDlemma_p2}:
in Lemma~\ref{ANDlemma}, we first fix the third parameter $p_0$, and the existence of the fourth parameter $p_2$ relies on the third; in Lemma~\ref{ANDlemma_p2}, we simultaneously fix the third and the fourth parameters.

The construction of the AND gadget and the proof of Lemma~\ref{ANDlemma} and Lemma~\ref{ANDlemma_p2} are deferred to Section~\ref{sect:ANDgadgets}.
In this section, we aim to prove Theorem~\ref{hardnessproof} for $a_1>0$ with directed graphs and assuming Lemma~\ref{ANDlemma}, while we do not need Lemma~\ref{ANDlemma_p2} at this moment.
We remark that the construction of AND gadget requires no directed edges, although we consider directed graph in this section.

\subsubsection{A Reduction from {\sc SetCover}}
We prove the theorem by a reduction from \textsc{SetCover}.


Without loss of generality, we will assume $K =O(n)$.\footnote{One way to justify this assumption is to consider \textsc{VertexCover}, which can be viewed as a special case of \textsc{SetCover} by viewing vertices as subsets and edges as elements. In a connected graph, the number of vertices $K $ never exceeds $O(n)$, if $n$ is the number of edges.}
We will also assume that each element in $U$ is covered by at least one subset $A_i$ in \textsc{SetCover} (otherwise we know for sure the instance is a {\sf NO} instance).
In addition, we assume the number of elements $n=|U|$ is an integer power of $2$, as we can add elements into $U$ and let these elements be included in all sets $A_i$ in the case $n$ is not an integer power of $2$.

We construct a graph $G$ with $N$ vertices which consists of two parts: the set cover part and the verification part,
where the set cover part simulates the \textsc{SetCover} instance and the verification part verifies if all the elements in the \textsc{SetCover} instance are covered.
The construction is shown in Figure~\ref{fig:highlevel}.

\begin{figure}
  \centerline{\includegraphics[width=\textwidth]{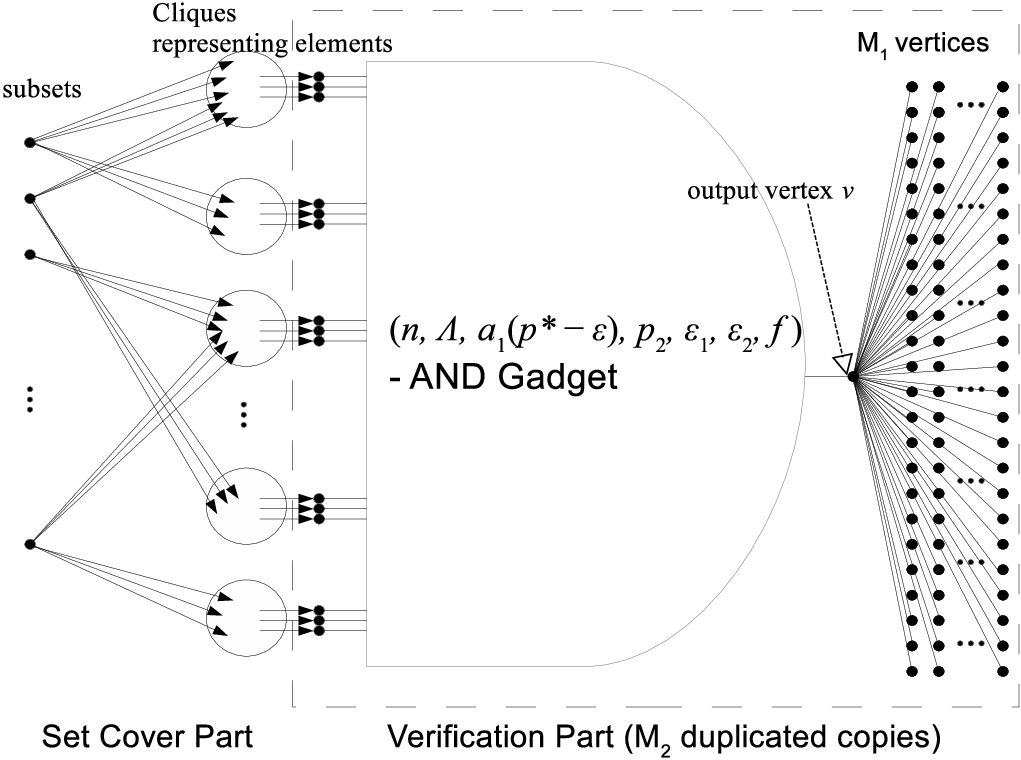}}
  \caption{The high-level structure of the reduction}
  \label{fig:highlevel}
\end{figure}

Define $\varepsilon=2\left(p^\ast-a_{\lfloor a_1n\rfloor}\right)$ which approaches to $0$ as $n\rightarrow\infty$ if $a_1>0$.
According to Lemma~\ref{ANDlemma}, for $p_0=a_1(p^\ast-\varepsilon)$ and $I=n$, there exists a constant $p_2>0$, such that if we set $\varepsilon_1=\frac1n$ and $\varepsilon_2=\frac1{100}p_2$, we can construct an $(n,\Lambda,p^\ast-\varepsilon,p_2,\varepsilon_1,\varepsilon_2,f)$-AND gadget, where $\Lambda=O\left((1/\varepsilon_1)^{c_1}n^{c_2}\right)=O\left(n^{c_1+c_2}\right)$.
We will use this AND gadget later.
Define $M_1=n^{c_1+c_2+10}$, $M_2=n^2$, and $m=M_2\Lambda$.

\paragraph{The Set Cover Part}
Given a \textsc{SetCover} instance, we use a single vertex to represent a subset $A_i$ and a clique of size $m$ to represent each element in $U$.
If an element is in a subset, we create $m$ directed edges from the vertex representing the subset to each the $m$ vertices in the clique representing the element.

\paragraph{The Verification Part}
We construct the $\left(n,\Lambda,a_1(p^\ast-\varepsilon),p_2,\varepsilon_1,\varepsilon_2,f\right)$-AND gadget mentioned.
We associate each of the $n$ cliques to one of the $n$ inputs of this AND gadget, such that a matching is formed between the $n$ cliques and the $n$ inputs.
For each of the $n$ cliques and its associated input, we choose $\Lambda$ vertices from the clique, and connect them to the $\Lambda$ vertices of the associated input by $\Lambda$ directed edges.
We create $M_1$ vertices and let the output vertex $v$ of the AND gadget be connected to these $M_1$ vertices with undirected edges.
Then, we duplicate the AND gadget and the attached $M_1$ vertices to a total of $M_2$ copies such that the vertices at the input ends of the AND gadgets in all these $M_2$ copies are connected from the different vertices in the $n$ cliques as inputs.
This, in particular, justifies our choice of clique size $m=M_2\Lambda$.

\paragraph{The Size of the Construction}
To show that the reduction is in polynomial time, it is enough to show that the number of vertices $N$ in the graph $G$ we constructed is a polynomial of $n$.
According to Lemma~\ref{ANDlemma}, the AND gadget has $O\left((1/\varepsilon_1)^{c_1}n^{c_2+1}\right)=O\left(n^{c_1+c_2+1}\right)$ vertices.
We have
$$N=K +mn+M_2\left(O\left(n^{c_1+c_2+1}\right)+M_1\right)=K +mn+\Theta\left(n^{c_1+c_2+12}\right)=\Theta\left(n^{c_1+c_2+12}\right),$$
where $K +mn$ is the size for the set cover part and $M_2\left(O\left(n^{c_1+c_2+1}\right)+M_1\right)$ is the size for the verification part.

Finally, noticing that $N=\Theta\left(n^{c_1+c_2+12}\right)$ and letting $\tau=\frac1{c_1+c_2+12}$ (which depends on $c_1,c_2$, and $c_1,c_2$ depends only on $f$), the lemma below immediately concludes Theorem~\ref{hardnessproof} for the case $a_1>0$ with directed edges.
\begin{lemma}\label{YESNOcompare}
  If the \textsc{SetCover} instance is a {\sf YES} instance, by choosing $k$ seeds appropriately, we can infect at least $\Theta\left(n^{c_1+c_2+12}\right)$ vertices in expectation in the graph $G$ we have constructed; if it is a {\sf NO} instance, we can infect at most $O\left(n^{c_1+c_2+11}\right)$ vertices in expectation for any choice of $k$ seeds.
\end{lemma}
\begin{proof}
  If the \textsc{SetCover} instance is a {\sf YES} instance, we are able to choose $k$ subsets $\{A_{i_1},\ldots,A_{i_k}\}\subseteq A$ such that $A_{i_1}\cup\cdots\cup A_{i_k}=U$.
  We choose the $k$ vertices corresponding to these $k$ subsets as seeds.

  We say that a clique representing an element is \emph{activated} if all its $m$ vertices are infected with probabilities more than $p^\ast-\varepsilon$.
  If a vertex representing a subset is seeded, for each clique representing the element it covers, each of the $m$ vertices in this clique will be infected with probability $a_1$.
  Thus, $ma_1$ vertices will be infected in expectation.
  According to Chernoff-Hoeffding inequality, with probability at least $1-\exp\left(-\frac18a_1^2m\right)$, there are more than $\frac12a_1m$ infected vertices in the clique.
  If this happens, in the next cascade iteration, each vertex in the clique has more than $\frac12a_1m$ infected neighbors, so it will be infected with probability at least $a_{\lfloor\frac12a_1m\rfloor}\geq a_{\lfloor a_1n\rfloor}>p^\ast-\varepsilon$ (notice that $\frac12m=\Theta(n^{c_1+c_2+2})\gg n$).
  Therefore, if a vertex representing a subset is seeded and a clique representing an element is in this subset, then this clique is activated with probability at least $1-\exp\left(-\frac18a_1^2n\right)$.

  By our choice of $k$ seeds, each of the clique is activated with probability at least $1-\exp\left(-\frac18a_1^2n\right)$.
  By a union bound, all the $n$ cliques will be activated with probability at least
  $$p_{\text{activated}}=1-K n\exp\left(-\frac18a_1^2n\right)=\Theta(1).$$

  In the highly likely case where all the $n$ cliques are activated, all the vertices at the input ends of all the AND gadgets will be infected with probability more than $a_1(p^\ast-\varepsilon)$.
  Since the parameter $p_0=a_1(p^\ast-\varepsilon)$ is set for the AND gadget, the output vertex $v$ falls into case (2) in Definition~\ref{defi:AND}, which means it will be infected with probability more than $p_2-\varepsilon_2$.
  Therefore, all the $M_1$ vertices connected to $v$ in each of the $M_2$ copies will be infected with probability at least $a_1(p_2-\varepsilon_2)$, so the expected total number of infected vertices is at least $p_{\text{activated}}\cdot a_1(p_2-\varepsilon_2)M_1M_2=\Theta\left(n^{c_1+c_2+12}\right)$.

  On the other hand, if the \textsc{SetCover} instance is a {\sf NO} instance, consider any choice of $k$ seeds with $k_1$ of them in the $K$ vertices representing subsets, $k_2$ of them in the $n$ cliques, and the remaining $k_3=k-k_1-k_2$ of them in the verification part.
  We first show that at least one clique will not be activated.

  The $k_3$ vertices in the verification part play no role in activating the cliques, as the $n$ cliques are connected to the verification part by directed edges.
  As for the $k_2$ vertices in the cliques, since we assume each element in $U$ is in at least one subset, infecting any vertex in any clique is at most as good as infecting the vertex representing the subset covering the element that the clique represents.
  Therefore, when analyzing the activation of cliques, we can reason as if these $k_2$ seeds are among the $K$ subsets.
  Since the \textsc{SetCover} instance is a {\sf NO} instance, and we have picked $k_1+k_2\leq k$ subsets, at least one clique will not be activated.

  Among the $M_2$ AND gadgets, at most $k_2$ of them take the input vertices which are connected from the $k_2$ seeds in the cliques.
  Since these $k_2$ seeds are infected with probability $1$ making these input vertices infect with probability $a_1$ which may be larger than $\frac{11}{10}a_1(p^\ast-\varepsilon)$, the outputs of these $k_2$ AND gadgets are unknown as it falls into neither case (1) nor case (2).
  We have also assumed $k_3$ seeds are selected in the verification parts, so we also do not know the outputs of another (at most) $k_3$ AND gadgets.

  For the remaining $M_2-k_2-k_3$ AND gadgets, they fall into case (1)
  by the fact that at least one clique is not activated and our setting $p_0=a_1(p^\ast-\varepsilon)$ for the AND gadget.
  Since we have set the AND gadget parameter $\varepsilon_1=\frac1n$, the output vertex $v$ will be infected with probability less than $\frac1n$, which will infect at most $a_1\frac{M_1}n$ vertices in expectation among the $M_1$ vertices on the right-hand side of Figure~\ref{fig:highlevel}.
  Notice that each AND gadget has $O(n^{c_1+c_2+1})$ vertices by Lemma~\ref{ANDlemma}, and the set cover part has $K+nm$ vertices.
  In this case, even if all the $K+nm+(M_2-k_2-k_3)\cdot O\left(n^{c_1+c_2+1}\right)=O\left(n^{c_1+c_2+3}\right)$ vertices in the set cover part and the $(M_2-k_2-k_3)$ AND gadgets are infected, the total number of infected vertices cannot exceed $O\left(n^{c_1+c_2+3}\right)+M_2\cdot a_1\frac{M_1}n=O\left(n^{c_1+c_2+11}\right)$.

  Finally, for those remaining $k_2+k_3$ AND gadgets whose outputs are unknown,
  even if all vertices in these $k_2+k_3$ copies of AND gadgets and their attached $M_1$ vertices are infected, this total number is still $(k_2+k_3)\cdot \left(O\left(n^{c_1+c_2+1}\right)+M_1\right)=(k_2+k_3)\cdot O\left(n^{c_1+c_2+10}\right)=O\left(n^{c_1+c_2+11}\right)$.
  Therefore, if the \textsc{SetCover} instance is a {\sf NO} instance, we can infect at most $O\left(n^{c_1+c_2+11}\right)$ vertices in $G$.
\end{proof}

\subsection{Proof of Theorem~\ref{hardnessproof} for $a_1>0$ with Undirected Graphs}
\label{sect:proof2}
To prove Theorem~\ref{hardnessproof} for undirected graphs, we will need the following \emph{directed edge gadget} which simulates directed edges, and the construction of this gadget also requires the property $a_2>2a_1$.
This is because the directed edge gadget also uses probability filter gadgets as building blocks.

\begin{definition}\label{defi:directededge}
  A \emph{$(\Upsilon,\epsilon,b,f)$-directed edge gadget} $\langle u,v\rangle$ takes one vertex $u$ as input and output one vertex $v$ such that the following properties hold.
  \begin{enumerate}
    \item\emph{directed property:} If $u$ is connected to each of the $\Upsilon$ vertices $v_1,\ldots,v_{\Upsilon}$ by a directed edge gadget $\langle u,v_i\rangle$, and $v_1,\ldots,v_{\Upsilon}$ are already infected, then $u$ will be infected with probability less than $\epsilon$.
    \item If the input $u$ is infected, then the output $v$ will be infected with probability $b$. Moreover, $b>0$.
  \end{enumerate}
\end{definition}

The size of a directed edge gadget is given by the following lemma.
\begin{lemma}\label{directededgelemma1}
  For any 2-quasi-submodular function $f$ with $a_1>0$, any positive integer $\Upsilon$ and any $\epsilon>0$, there exists $b\in(0,1)$ such that we can construct a $(\Upsilon,\epsilon,b,f)$-directed edge gadget with $\Theta\left(\Upsilon^d(1/\epsilon)^d\right)$ vertices and $\Theta\left(\Upsilon^d(1/\epsilon)^d\right)$ edges, where $d>1$ is a constant depending only on $f$.
\end{lemma}

We also need the following lemma.
\begin{lemma}\label{directededgelemma3}
  Given an $(I,\Lambda,p_0,p_2,\varepsilon_1,\varepsilon_2,f)$-AND gadget, for any $\Upsilon$ and $\epsilon$, we can construct a $(\Upsilon,\epsilon,b,f)$-directed edge gadget with $b\in\left(p_2-\frac12\varepsilon_2,p_2\right]$ using $\Theta\left(\Upsilon^d(1/\epsilon)^d\right)$ vertices and $\Theta\left(\Upsilon^d(1/\epsilon)^d\right)$ edges, where $d>1$ is a constant depending only on $f$.
\end{lemma}

The construction of the directed edge gadget and the proofs of Lemma~\ref{directededgelemma1} and Lemma~\ref{directededgelemma3} are deferred to Section~\ref{sect:directededgegadgets}.

\subsubsection{A Reduction from {\sc SetCover}}
According to Lemma~\ref{ANDlemma_p2}, for $I=n$, there exists a constant $p_2>0$, such that if we set $\varepsilon_1=\frac1n$ and $\varepsilon_2=\frac1{100}p_2$, we can construct an $(n,\Lambda,p^\ast(p_2-\varepsilon_2),p_2,\varepsilon_1,\varepsilon_2,f)$-AND gadget, where $\Lambda=O\left((1/\varepsilon_1)^{c_1}n^{c_2}\right)=O\left(n^{c_1+c_2}\right)$.
Define $M_2=n^2$ and $m=M_2\Lambda$ as before, and we will define $M_1$ later.

Applying Lemma~\ref{directededgelemma3}, we can construct a $(mn,m^{-2},b,f)$-directed edge gadget such that $b\in\left(p_2-\frac12\varepsilon_2,p_2\right]$.
The numbers of vertices and edges in this directed edge gadget are both
$$\Theta\left((mn)^d\left(1/m^{-2}\right)^d\right)=\Theta\left(m^{3d}n^d\right)=\Theta\left(n^{(3c_1+3c_2+7)d}\right).$$
We will use this directed edge gadget exclusively in the construction.

Finally, define $M_1=n^{(30c_1+30c_2+70)d}$.

We will construct an undirected graph $G$ similar to the one in the last section, with some modifications.
We make the following two modifications:
\begin{enumerate}
  \item We replace all directed edges in Figure~\ref{fig:highlevel} by $(mn,m^{-2},b,f)$-directed edge gadgets. These consist of 1) the directed edges connecting between the $K$ vertices representing subsets and the $n$ cliques representing elements and 2) the directed edges connecting between the set cover part and the verification part.
  \item We use the $(n,\Lambda,p^\ast(p_2-\varepsilon_2),p_2,\varepsilon_1,\varepsilon_2,f)$-AND gadgets in the verification part instead of the $(n,\Lambda,a_1(p^\ast-\varepsilon),p_2,\varepsilon_1,\varepsilon_2,f)$-AND gadgets.
\end{enumerate}
For the remaining parts of the construction, all the edges, including the ones in the clique, the ones in the AND gadget, and the ones connected to the $M_1$ vertices on the right-hand side of Figure~\ref{fig:highlevel}, are undirected edges.
In particular, we recall that the edges in the AND gadget are undirected.

\paragraph{The Size of the Construction}
We show that the total number of vertices in $G$ is still of polynomial size.
In the set cover part, we have created at most $K mn$ directed edge gadgets between the $K$ vertices and the $mn$ vertices in the $n$ cliques.
The total size of the set cover part is at most $K+mn+K mn\cdot\Theta\left(n^{(3c_1+3c_2+7)d}\right)=O\left(n^{(3c_1+3c_2+7)d+c_1+c_2+4}\right)$.

In the verification part, the AND gadget contains $O\left(n^{c_1+c_2+1}\right)$ vertices by Lemma~\ref{ANDlemma_p2}, the total number of vertices is $M_2(O(n^{c_1+c_2+1})+M_1)=\Theta\left(n^{(30c_1+30c_2+70)d+2}\right)$.

Since $d>1$ by Lemma~\ref{ANDlemma_p2}, $N$ is dominated by the number of vertices in the verification part: $N=\Theta\left(n^{(30c_1+30c_2+70)d+2}\right)$.

Finally, with $\tau=\frac1{(30c_1+30c_2+70)d+2}$, Theorem~\ref{hardnessproof} for undirected graphs follows immediately from the following lemma.
\begin{lemma}\label{YESNOcompare2}
  If the \textsc{SetCover} instance is a {\sf YES} instance, by choosing $k$ seeds appropriately, we can infect $\Theta\left(n^{(30c_1+30c_2+70)d+2}\right)$ vertices in expectation in the graph $G$; if it is a {\sf NO} instance, we can infect at most $O\left(n^{(30c_1+30c_2+70)d+1}\right)$ vertices in expectation for any choice of $k$ seeds.
\end{lemma}
\begin{proof}
  If the \textsc{SetCover} instance is a {\sf YES} instance, we are able to choose $k$ subsets $\{A_{i_1},\ldots,A_{i_k}\}\subseteq A$ such that $A_{i_1}\cup\cdots\cup A_{i_k}=U$.
  We choose the $k$ vertices corresponding to these $k$ subsets as the seeds.
  Since the \textsc{SetCover} instance is a {\sf YES} instance, for each clique, each vertex is connected from a seed by a directed edge gadget, which will be infected with probability $b$.
  In each clique, $bm$ vertices will be infected in expectation, and the remaining vertices in the clique will be infected with probability at least $a_{\lfloor bm\rfloor}$, which has limit $p^\ast$ as $n\rightarrow\infty$.

  By the same analysis in the proof of Lemma~\ref{YESNOcompare}, with a high probability $p_{\text{activated}}=\Theta(1)$, all the $n$ cliques will be activated such that all vertices in the clique will be infected with probability $p^\ast-\varepsilon$ for certain $\varepsilon=o(1)$.
  By our construction and Lemma~\ref{directededgelemma3}, each of the $mn$ vertices that are passed into the input of the AND gadget will be infected with probability $$(p^\ast-\varepsilon)b\in\left((p^\ast-\varepsilon)\left(p_2-\frac12\varepsilon_2\right),(p^\ast-\varepsilon)p_2\right]\subseteq\left(p^\ast(p_2-\varepsilon_2),\frac{11}{10}p^\ast(p_2-\varepsilon_2)\right),$$
  by noticing that $\varepsilon=o(1)$ and $\varepsilon_2=\frac1{100}p_2<\frac1{10}p_2$ is a constant.
  Thus, the AND gadget falls into case (2) of Definition~\ref{defi:AND}, so the output vertex $v$ of the AND gadget will be infected with probability more than $p_2-\varepsilon_2$.
  Therefore, each of the $M_1$ vertices will be infected with probability $p_{\text{activated}}a_1(p_2-\varepsilon_2)$, and the expected total number of infected vertices in those $M_2$ copies of $M_1$ vertices is already $p_{\text{activated}}a_1(p_2-\varepsilon_2)M_1M_2=\Theta\left(n^{(30c_1+30c_2+70)d+2}\right)$.

  If the \textsc{SetCover} instance is a {\sf NO} instance, consider any choice of the $k$ seeds with $k_1$ seeds in the $K$ vertices representing subsets, $k_2$ seeds in the directed edge gadgets connecting the $K$ vertices and $nm$ vertices in the $n$ cliques, $k_3$ seeds in the $n$ cliques, $k_4$ seeds in the directed edge gadgets between the $n$ cliques in the set cover part and the inputs of the AND gadget in the verification part, and the remaining $k_5=k-k_1-k_2-k_3-k_4$ seeds in the verification parts.
  We first aim to show that at least one clique will not be activated with high probability.

  When analyzing cliques' activation, it is easy to see that putting $k_2$ seeds on the directed edge gadgets is at most as good as putting them on the corresponding vertices representing the subsets.
  Similarly, putting $k_4$ seeds on the directed edge gadgets connecting the set cover part and the verification part is at most as good as putting them on the corresponding vertices in the cliques, and having $k_3+k_4$ seeds in the cliques is at most as good as having them in the $K$ vertices representing the subsets covering the elements that those cliques represent.
  Thus, we can reason as if we have selected $k_1+k_2+k_3+k_4$ subsets in the \textsc{SetCover} problem.
  Since the \textsc{SetCover} instance is a {\sf NO} instance, those $k_1+k_2+k_3+k_4\leq k$ seeds cannot cover all the cliques.
  As for the $k_5$ seeds in the verification part, their influences on each vertex in the $n$ cliques is at most $m^{-2}$ based on Definition~\ref{defi:directededge}, which has remote effect to the cliques, and we will discuss it later.

  To show that at least one clique is not activated, it remains to show that the clique not covered by those $k_1+k_2+k_3+k_4$ vertices cannot be activated.
  For each of those vertices representing subsets that are not picked, since it is connected to at most $mn$ vertices ($m$ vertices in each of the $n$ cliques) by the $(mn,m^{-2},b,f)$-directed edge gadgets, it will be infected with probability at most $m^{-2}$ by Definition~\ref{defi:directededge}.
  For each vertex in each uncovered clique, it may only be infected due to 1) the influence from one of the $K$ vertices which is not seeded and which is infected with probability at most $m^{-2}$, or 2) the influence from the $k_5$ seeds from the verification parts.
  In particular, it will be infected due to (1) with probability $bm^{-2}$, and it will be infected due to (2) with probability $m^{-2}$.
  By a union bound, the probability that there exist infected vertices in an uncovered clique is at most
  $m\cdot(bm^{-2}+m^{-2})=O(m^{-1})$.
  Since there can be at most $n$ uncovered cliques, the probability that all uncovered cliques contain no infected vertex is at least
  $p_{\text{no}}=1-n\cdot O\left(m^{-1}\right)>1-O\left(\frac1n\right).$
  Therefore, with the probability above, there exists at least one clique which is not activated.

  In the case that not all cliques are activated, since all the vertices in a not activated clique are infected with probability $0$, the corresponding input vertices to the AND gadget are also infected with probability $0b=0$.
  The output vertices $v$ in the AND gadgets therefore fall into case (1) in at least $M_2-k_3-k_4-k_5$ copies.
  Thus, in each of the corresponding $M_2-k_3-k_4-k_5$ copies of the $M_1$ vertices bundle (on the rightmost of Figure~\ref{fig:highlevel}), the expected number of infected vertices is at most $\varepsilon_1\cdot M_1=O\left(n^{(30c_1+30c_2+70)d-1}\right)$.
  In this case, even if all the vertices in the entire set cover part, the $mn$ directed edge gadgets connecting the two parts, all the $M_2$ AND-gadgets, and the remaining $k_3+k_4+k_5$ copies of the $M_1$ vertices bundles, the total number of infected vertices is at most
  \begin{align}\label{eqn:yesnocompare}
    &K mn\cdot\Theta\left(n^{(3c_1+3c_2+7)d}\right)+mn\cdot\Theta\left(n^{(3c_1+3c_2+7)d}\right)+M_2\cdot O\left(n^{c_1+c_2+1}\right)\nonumber\\
    &\qquad+(k_3+k_4+k_5)\left(O\left(n^{c_1+c_2+1}\right)+M_1\right)+(M_2-k_3-k_4-k_5)O\left(n^{(30c_1+30c_2+70)d-1}\right)\nonumber\\
    =&O\left(n^{(3c_1+3c_2+7)d+c_1+c_2+4}\right)+\Theta\left(n^{(3c_1+3c_2+7)d+3}\right)+O\left(n^{c_1+c_2+3}\right)\nonumber\\
    &\qquad+O\left(n^{(30c_1+30c_2+70)d+1}\right)+O\left(n^{(30c_1+30c_2+70)d+1}\right)\nonumber\\
    =&O\left(n^{(30c_1+30c_2+70)d+1}\right).
  \end{align}
  Finally, even assuming all vertices in $G$ are infected in the case that all cliques are activated (which happens with probability $1-p_{\text{no}}<O\left(\frac1n\right)$), the expected number of infected vertices is at most
  $$p_{\text{no}}\cdot O\left(n^{(30c_1+30c_2+70)d+1}\right)+(1-p_{\text{no}})N=O\left(n^{(30c_1+30c_2+70)d+1}\right),$$
  which concludes the lemma.
\end{proof}

\subsection{Constructions of Some Other Required Gadgets}
\label{sect:othergadgets}
Before constructing the AND gadget and the directed edge gadget, we need some other gadgets.
In this section and the next two sections, graph with undirected edges are considered.

We will construct the \emph{probability scaling down gadget} and the \emph{probability filter gadget}, which are used to construct the AND gadget and the directed edge gadget.
The relation of these gadgets are shown in Figure~\ref{fig:gadgetrelation}.

\begin{figure}
\centerline{\includegraphics[width=\textwidth]{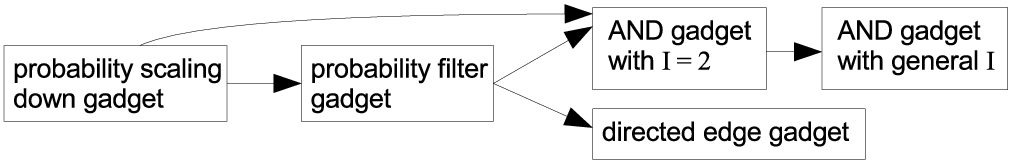}}
\caption{The relation of all the gadgets defined}
\label{fig:gadgetrelation}
\end{figure}

\subsubsection{Probability Scaling Down Gadget}
We first define and construct the following \emph{probability scaling down gadget} which is an essential component of both the AND gadget and the directed edge gadget.
\begin{definition}\label{defi:psdg}
  The \emph{$(\alpha,\varepsilon,f)$-probability scaling down gadget} takes one vertex $u$ as input and output a vertex $v$ such that
  \begin{itemize}
    \item if $u$ is infected with probability $p_u$, $v$ will be infected with probability $p_v\in(\alpha p_u-\varepsilon,\alpha p_u]$.
  \end{itemize}
\end{definition}

\begin{lemma}\label{lem:psdg}
  For any 2-quasi-submodular function $f$ with $a_1>0$, any constant $\varepsilon>0$ and any $\alpha$ with $0<\alpha\leq p^\ast$, there exists an $(\alpha,\varepsilon,f)$-probability scaling down gadget with constant numbers of vertices and edges.
\end{lemma}
\begin{proof}
To construct this gadget, we iteratively add paths from $u$ to $v$, where a path of length $\ell$ consists of $\ell-1$ vertices $w_1,\ldots,w_{\ell-1}$ and $\ell$ edges $(u,w_1),(w_1,w_2),\ldots,(w_{\ell-1},v)$.
Given $p_u$, by repeatedly adding paths from $u$ to $v$, we are increasing $p_v$.
In each iteration $i$, we add a path of length $\ell_i$ from $u$ to $v$, where $\ell_i$ is the minimum length to maintain $p_v\leq\alpha p_u$.
That is, either it is true that $p_v>\alpha p_u$ if a path of length $\ell_i-1$ was added, or $\ell_i=2$ which is already the minimum length a path can ever be.
The iterative process ends if $p_v\in(\alpha p_u-\varepsilon,\alpha p_u]$, and it is straightforward to check that such process will end as long as $\alpha\in(0,p^\ast]$.
Figure~\ref{fig:scalingdown} illustrates the probability scaling down gadget.
\begin{figure}
\centerline{\includegraphics{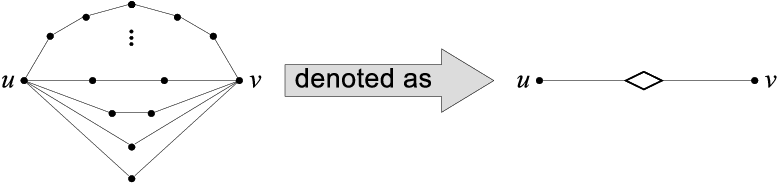}}
\caption{The probability scaling down gadget}
\label{fig:scalingdown}
\end{figure}

The size of the probability scaling down gadget depends on the influence function $f$ and the small constant $\varepsilon$.
Since $f$ is fixed in advance, the size of this gadget is constant.
\end{proof}

\begin{remark}\label{symmetricremark}
The probability scaling down gadget is symmetric.
Given $p_v=\alpha p_u$, then $p_u=\alpha p_v$ if $v$ becomes the input and $u$ becomes the output.
\end{remark}

\subsubsection{Probability Filter Gadget}
\label{subsect:pfg}
Based on the probability scaling down gadget, we can construct the following \emph{probability filter gadget}.
\begin{definition}\label{defi:probabilityfiltergadget}
A \emph{$(\Lambda,p_1,p_2,\varepsilon_1,\varepsilon_2,f)$-probability filter gadget} takes $\Lambda$ vertices as input, and outputs a vertex such that
  \begin{enumerate}
    \item if each vertex in the $\Lambda$ inputs is infected independently with a same probability less than $p_1$, then the vertex on the output end will be infected with a probability less than $\varepsilon_1$;
    \item if each vertex in the $\Lambda$ inputs is infected independently with a same probability in $(p_1,p_2]$, then the vertex on the output end will be infected with a probability in $(p_2-\varepsilon_2,p_2]$.
\end{enumerate}
\end{definition}

We aim to show the following lemma in this subsection.
\begin{lemma}\label{filtersize}
  Given any 2-quasi-submodular influence function $f$ with $a_1>0$, any constant $\varepsilon_2>0$, any $\varepsilon_1>0$, and any ratio $r>0$, we can construct a $(\Lambda,p_1,p_2,\varepsilon_1,\varepsilon_2,f)$-probability filter gadget with $p_2/p_1>r$ and $\Lambda=O((1/\varepsilon_1)^c)$, and this probability filter gadget contains $O((1/\varepsilon_1)^c)$ vertices and $O((1/\varepsilon_1)^c)$ edges, where $c$ is a constant.
\end{lemma}

To construct the probability filter gadget, we first construct the gadget shown in Figure~\ref{fig:filter}, which is the building block of this gadget.
We will call this building block \emph{probability separation block}.
As shown in the figure, this building block takes $h$ vertices as input and outputs one vertex.
Particularly, we apply $h$ probability scaling down gadgets to ``scale down'' the probabilities of all input vertices' infection by a factor of $\alpha$, and then connect those vertices to the output vertex.
\begin{figure}
\centerline{\includegraphics{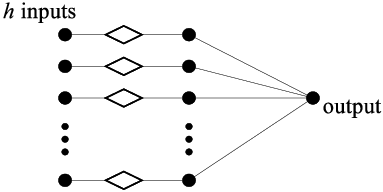}}
\caption{The probability separation block}
\label{fig:filter}
\end{figure}

The probability filter gadget consists of $\ell$ layers such that the $i$-th layer consists of $h^{\ell-i}$ such probability separation blocks,
where the output vertices of every $h$ probability separation blocks in the $i$-th layer are the input of a probability separation block in the $(i+1)$-th layer.
Because there are $h^{\ell-1}$ probability separation blocks in the first layer,
the probability filter gadget takes $\Lambda=h^\ell$ vertices as input.
The probability filter gadget outputs a single vertex after $\ell$ layers.
We will tune the value of $\alpha$, $h$ and $\ell$ such that the two properties in Definition~\ref{defi:probabilityfiltergadget} hold for certain thresholds $p_1$ and $p_2$.

For each probability separation block, suppose each of the $h$ vertices in the input are infected with probability $x$ independently,
and let $y=y(x)$ be the probability that the output vertex is infected.
We aim to tune the value of $\alpha$ and $h$ such that the graph of $y(x)$ looks like Figure~\ref{fig:filtercurve}.

By considering the number of infected neighbors of the output vertex, it is straightforward to see that
\begin{equation}\label{FilterEquation}
y=\sum_{i=1}^h\binom{h}{i}a_i(\alpha x)^i(1-\alpha x)^{h-i}.
\end{equation}
For sufficiently small $x$, we have
$$y=h\alpha a_1x+\frac{h(h-1)}2\alpha^2(a_2-2a_1)x^2+o(x^2).$$
Choosing a sufficiently small constant $\delta>0$ and choosing $\alpha$ ($h$ will be set in the future) to satisfy $h\alpha a_1=1-\delta$, we have
$$y-x=-\delta x+\frac{h(h-1)}2\alpha^2(a_2-2a_1)x^2+o(x^2).$$
Since $y-x=-\delta x+o(x)$, we can see that $y<x$ for small enough $x$. On the other hand, for sufficiently large $h$ and sufficiently small $\delta$ (and adjusting $\alpha$ such that $h\alpha a_1=1-\delta$ still holds\footnote{According to Definition~\ref{defi:psdg} and Lemma~\ref{lem:psdg}, given the scale $\alpha^\ast$ for which we want to adjust to, we can construct a probability scaling down gadget such that the actual scale $\alpha$ is arbitrarily close to $\alpha^\ast$. Although we cannot make the adjustment exact, a close enough approximation would still satisfy our purpose here, as all we want is $\delta$ to be small enough, or $h\alpha a_1$ to be close enough to $1$.}), we have
$$\frac{h(h-1)}2\alpha^2=\frac12h^2\alpha^2-\frac h2\alpha^2=\frac{(1-\delta)^2}{2a_1^2}-\frac{(1-\delta)^2}{2ha_1^2}>\frac1{3a_1^2}.$$
We can see from the following that $y>x$ after a while as $x$ increases.
  \begin{align*}
  x_1=\frac{6a_1^2}{a_2-2a_1}\delta\qquad\Longrightarrow\qquad y(x_1)-x_1&>-\delta x_1+\frac{a_2-2a_1}{3a_1^2}x_1^2+o(x_1^2)\\
   &=\frac{6a_1^2}{a_2-2a_1}\delta^2+o(\delta^2)\\
   &>0.
  \end{align*}
Notice that the 2-quasi-submodularity of $f$ makes sure $a_2>2a_1$ such that $x_1$ is positive.

We have seen that $y<x$ for small enough $x$, and $y>x$ after $x$ increases.
There must be a threshold $p_1$ such that $y=x$ at $x=p_1$ by the Intermediate Value Theorem.
On the other hand, $y$ is upper bounded by $p^\ast$ while $x$ can be as large as $1$, so $y\leq x$ for sufficiently large $x$.
The Intermediate Value Theorem suggests there exists another threshold $x=p_2>p_1$ such that $y=x$.
Consequently, Figure~\ref{fig:filtercurve} indeed represents the graph of $y(x)$ for the proper choices of $\alpha$ and $h$.
\begin{figure}
\centerline{\includegraphics[scale=0.5]{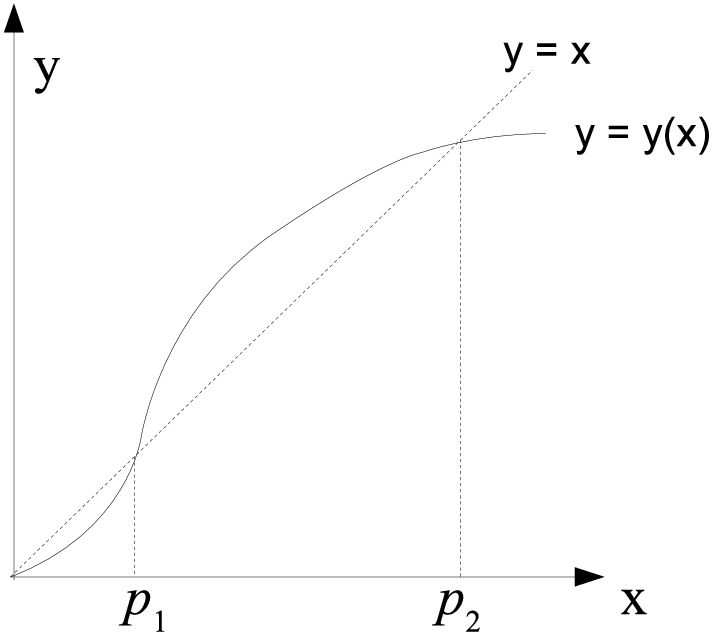}}
\caption{The output probability $y$ versus the input probability $x$}
\label{fig:filtercurve}
\end{figure}

Finally, from the graph in Figure~\ref{fig:filtercurve}, we can see that the infection probability of the output vertices in the $i$-th layer increases as $i$ increases, if all the $\Lambda=h^\ell$ input vertices are infected with an independent probability larger than $p_1$.
In contrast, the infection probability of the output vertices in the $i$-th layer decreases as $i$ increase, if all the $\Lambda=h^\ell$ input vertices are infected with an independent probability less than $p_1$.
By setting $\ell$ large enough, we can make both (1) and (2) in Definition~\ref{defi:probabilityfiltergadget} hold.

Before we move on, we show some properties of the thresholds $p_1$ and $p_2$, and our objective is to show the following proposition which is a part of Lemma~\ref{filtersize}.
\begin{proposition}\label{lemmaratio}
  For any large ratio $r>0$, we can find $h$ and $\alpha$ such that
  $p_2/p_1>r$.
\end{proposition}
By the calculation above, the proposition below follows immediately.
\begin{proposition}\label{lemmap1}
  $p_1<\frac{6a_1^2}{a_2-2a_1}\delta$.
\end{proposition}
We also have the following lower bound for $p_2$.
\begin{proposition}\label{lemmap2}
  By choosing $h$ sufficiently large and $\delta$ sufficiently small, we have $p_2>a_1\gamma$ for any $\gamma$ such that
  $$a_2(1-e^{-\gamma}-\gamma e^{-\gamma})-a_1(\gamma-\gamma e^{-\gamma})>0.$$
\end{proposition}
\begin{proof}
  By replacing all $a_3,a_4,\ldots,a_h$ to $a_2$ in Equation~(\ref{FilterEquation}), we have
  \begin{align*}
    y&\geq\sum_{i=1}^h\binom{h}{i}a_2(\alpha x)^i(1-\alpha x)^{h-i}-\binom{h}{1}(a_2-a_1)\alpha x(1-\alpha x)^{h-1}\\
    &=a_2\left(\sum_{i=0}^h\binom{h}{i}(\alpha x)^i(1-\alpha x)^{h-i}-(1-\alpha x)^h\right)-h(a_2-a_1)\alpha x(1-\alpha x)^{h-1}\\
    &=a_2-a_2(1-\alpha x)^h-h(a_2-a_1)\alpha x(1-\alpha x)^{h-1}\\
    &=a_2-a_2\exp(h\ln(1-\alpha x))-h(a_2-a_1)\alpha x\exp((h-1)\ln(1-\alpha x))\\
    &\geq a_2-a_2\exp(-h\alpha x)-h(a_2-a_1)\alpha x\exp(-\alpha x(h-1))\tag{concavity of $\ln$ function}.
  \end{align*}
  Letting $x=a_1\gamma$, we have
  \begin{align*}
    y-x&\geq a_2-a_2\exp(-\gamma(1-\delta))-(a_2-a_1)(1-\delta)\gamma\exp\left(\gamma(1-\delta)\left(\frac1h-1\right)\right)-a_1\gamma\tag{since $x=a_1\gamma$ and $h\alpha a_1=1-\delta$}\\
    &>a_2(1-e^{-\gamma}-\gamma e^{-\gamma})-a_1(\gamma-\gamma e^{-\gamma})-\epsilon,
  \end{align*}
  where in the last step, for any $\epsilon>0$, we can find small enough $\delta$ and large enough $h$ to make the inequality holds.
  Rigorously, we have $1-\delta\rightarrow1$ and $\frac1h\rightarrow0$ for $\delta\rightarrow0$ and $h\rightarrow\infty$. The expression in the second last step is a continuous function, which has limit $a_2(1-e^{-\gamma}-\gamma e^{-\gamma})-a_1(\gamma-\gamma e^{-\gamma})$, and the last step is obtained by the definition of limit.

  Therefore, $y-x>0$ for any $x>p_1$ with $x=a_1\gamma$, where $\gamma$ satisfies
  $$a_2(1-e^{-\gamma}-\gamma e^{-\gamma})-a_1(\gamma-\gamma e^{-\gamma})>0,$$
  which implies the proposition.
\end{proof}

We remark that there always exists $\gamma$ satisfying the inequality in Proposition~\ref{lemmap2}.
To see this, we show that $\Phi(\gamma):=a_2(1-e^{-\gamma}-\gamma e^{-\gamma})-a_1(\gamma-\gamma e^{-\gamma})>0$ when $\gamma$ is sufficiently small.
By straightforward calculations, we have $\Phi(0)=\Phi'(0)=0$ and $\Phi''(0)=a_2-2a_1>0$, which means $\Phi(0)=0$ and $\Phi$ is increasing on $[0,\gamma_0)$ for some small $\gamma_0$, which further implies that $\Phi$ is positive on $[0,\gamma_0)$.

Proposition~\ref{lemmap1} implies that we can construct the probability filter gadget with arbitrarily small $p_1$ by setting $\delta$ small.
On the other hand, Proposition~\ref{lemmap2} implies that $p_2$ can be made larger than some number depending only on $a_1$ and $a_2$, which in particular can be considerably larger than $p_1$, which yields Proposition~\ref{lemmaratio}.

Finally, we are ready to show Lemma~\ref{filtersize}.
\begin{proof}[Proof of Lemma~\ref{filtersize}]
  The possibility of this construction is straightforward, as the construction is already made explicit in this section.
  It remains to show that the gadget contains $O((1/\varepsilon_1)^c)$ vertices and $O((1/\varepsilon_1)^c)$ edges, and $\Lambda=O((1/\varepsilon_1)^c)$.

  Since $\varepsilon_2$ is a constant, we only need constantly many layers such that the input probability $x$ increases to more than $p_2-\varepsilon_2$, if $x$ is initially larger than $p_1$.

  To investigate how many layers are needed to make $x$ decreases to less than $\varepsilon_1$ in the case $x$ is initially smaller than $p_1$,
  recall that in each layer of the probability filter gadget, the input probability $x$ is updated to $y$ such that $y-x=-\delta x+o(x)$ for sufficiently small $x$, so each time $x$ is decreased by a factor of $(1-\delta)$.
  After a constant number of layers, $x$ will be sufficiently small such that the term $o(x)$ is negligible, and after another $\frac{\log(1/\varepsilon_1)}{\log(1/(1-\delta))}$ layers, $x$ will decrease by a factor of $(1-\delta)^{\frac{\log(1/\varepsilon_1)}{\log(1/(1-\delta))}}=\varepsilon_1$, which makes the value of $x$ much smaller than $\varepsilon_1$.
  Therefore, we need at most $\ell=O(\log(1/\varepsilon_1))$ layers.
  Let $\chi_v,\chi_e$ be the number of vertices and edges respectively in a probability separation block shown in Figure~\ref{fig:filter}, and they are both constants according to Lemma~\ref{lem:psdg}.
  The total number of vertices in a probability filter gadget is
  $$\sum_{i=1}^\ell\chi_v\cdot h^{\ell-i}=\chi_v\frac{h^\ell-1}{h-1}=\Theta\left(h^\ell\right)=O((1/\varepsilon_1)^c),$$
  and the total number of edges has the same asymptotic bound by the same calculation above, with $\chi_v$ changed to $\chi_e$.
  Thus, we conclude that the gadget contains $O((1/\varepsilon_1)^c)$ vertices and $O((1/\varepsilon_1)^c)$ edges.

  For $\Lambda$, we have $\Lambda=h^\ell=O((1/\varepsilon_1)^c)$ by our construction, which concludes the last part of the lemma.
\end{proof}

\subsection{Construction of the AND Gadget with $I=2$}
\label{sect:ANDgadgetwithtwoinputs}
In this section, we construct the AND gadget with parameter $I=2$.
The AND gadget makes use of a single probability filter gadget with the same choices of parameters $\Lambda$, $p_2$, $\varepsilon_1$, $\varepsilon_2$ and $f$.
The AND gadget takes two sets $I_1,I_2$ of vertices as inputs, and each set has $\Lambda=h^\ell$ vertices.
Let $I_1=\{u_1,u_2,\ldots,u_{\Lambda}\}$ and $I_2=\{v_1,v_2,\ldots,v_{\Lambda}\}$.
We create $\Lambda$ vertices $w_1,w_2,\ldots,w_{\Lambda}$ and create two edges $(u_i,w_i),(v_i,w_i)$ for each $i=1,2,\ldots,\Lambda$.
We apply the probability scaling down gadgets to create another $\Lambda$ vertices $w_1',w_2',\ldots,w_{\Lambda}'$ such that $p(w_i')=\beta p(w_i)$ for each $i=1,2,\ldots,\Lambda$, where $\beta$ is set to the value such that
$$\beta\varphi_T^+(p_0)<p_2,\qquad\beta\varphi_T^-(p_0)>p_1,\qquad\mbox{and}\qquad\beta\varphi_F^+(p_0)<p_1,$$
where
\begin{eqnarray*}
  \varphi_T^+(p_0)&=&(a_2-2a_1)\left(\frac{11}{10}p_0\right)^2+2a_1\left(\frac{11}{10}p_0\right),\\
  \varphi_T^-(p_0)&=&(a_2-2a_1)p_0^2+2a_1p_0,\\
  \varphi_F^+(p_0)&=&\frac{11}{20}(a_2-2a_1)p_0^2+\frac{16}{10}a_1p_0.
\end{eqnarray*}
The construction is shown in Figure~\ref{fig:ANDwithtwoinput}.

\begin{figure}
\centerline{\includegraphics{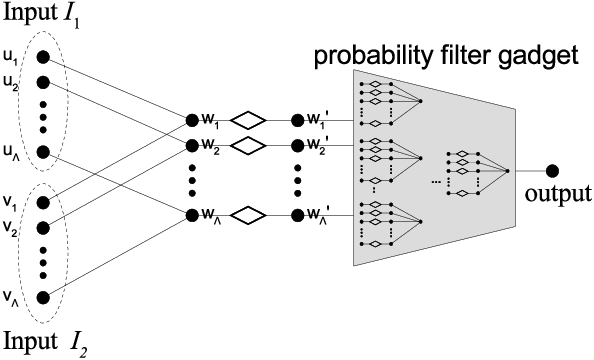}}
\caption{AND gadget with $I=2$}
\label{fig:ANDwithtwoinput}
\end{figure}

Notice that if all $u_i$ and $v_i$ are infected with an independent probability in the interval $(p_0,\frac{11}{10}p_0)$, that is, the inputs $I_1,I_2$ fall into case (2) in Definition~\ref{defi:AND}, $w_i$ will be infected with probability
\begin{align*}
p(w_i)&=a_2p(u_i)p(v_i)+a_1p(u_i)(1-p(v_i))+a_1p(v_i)(1-p(u_i))\\
&=(a_2-2a_1)p(u_i)p(v_i)+a_1p(u_i)+a_1p(v_i),
\end{align*}
which is in the interval $\left(\varphi_T^-(p_0),\varphi_T^+(p_0)\right)$.

On the other hand, if one of $u_i$ and $v_i$ is infected with probability less than $\frac12p_0$ and the other one is infected with probability less than $\frac{11}{10}p_0$, that is, the inputs $I_1,I_2$ fall into case (1) in Definition~\ref{defi:AND}, $w_i$ will be infected with probability
\begin{align*}
p(w_i)&=(a_2-2a_1)p(u_i)p(v_i)+a_1p(u_i)+a_1p(v_i)\\
&<\frac{11}{20}(a_2-2a_1)p_0^2+\frac{16}{10}a_1p_0\\
&=\varphi_F^+(p_0).
\end{align*}
Given that $\left(\beta\varphi_T^-(p_0),\beta\varphi_T^+(p_0)\right)\subseteq(p_1,p_2)$ and $\beta\varphi_F^+(p_0)<p_1$,
it is now straightforward to check that the two properties (1) and (2) in Definition~\ref{defi:AND} hold for $I=2$, since the probability filter gadget will ``filter'' the two probabilities such that one goes to a value less than $\varepsilon_1$ and the other goes into $(p_2-\varepsilon_2,p_2]$.

By our construction of probability scaling down gadget, the factor must satisfy $\beta\leq p^\ast$.
It seems worrying that $\left(\varphi_T^-(p_0),\varphi_T^+(p_0)\right)$ and $\varphi_F^+(p_0)$ will be both scaled down to smaller than $p_1$ even if we take maximum $\beta=p^\ast$.
Indeed, Proposition~\ref{lemmaratio} and Proposition~\ref{lemmap1} ensure that this cannot happen, as we can always make $p_1$ small enough by making $\delta$ small enough.
We remark here that the choice of $\delta$ depends on $p_0$ and $p^\ast$ (it needs to be considerably smaller than some polynomial of $p_0$ such that $\left(\varphi_T^-(p_0),\varphi_T^+(p_0)\right)$ and $\varphi_F^+(p_0)$ can be scaled down to different sides of $p_1$), where $p^\ast$ depends only on $f$.

Now we prove the following lemma, which is a special case of Lemma~\ref{ANDlemma} with $I=2$.
\begin{lemma}\label{ANDlemmawithtwoinputs}
  Given any 2-quasi-submodular function $f$ with $a_1>0$ and any constant threshold $p_0>0$, there exists a constant $p_2>0$ depending on $p_0$ and $f$ such that for any constant $\varepsilon_2>0$ and any $\varepsilon_1>0$, we can construct a $(2, \Lambda, p_0, p_2, \varepsilon_1,\varepsilon_2, f)$-AND gadget with $\Lambda=O\left((1/\varepsilon_1)^{c_1}\right)$,
  and the number of vertices and edges in this AND gadget are both $O\left((1/\varepsilon_1)^{c_1}\right)$,
  where $c_1$ is a constant.
\end{lemma}
\begin{proof}
  The existence of this AND gadget is shown by the explicit construction in this section.

  To show that $p_2$ only depends on $p_0$ and $f$, notice that it depends on $h$, $\delta$ and $f$ (in particular, $a_1$ and $a_2$ only) according to Proposition~\ref{lemmap2}.
  Additionally, $h,\alpha$ are selected such that $\delta=1-h\alpha a_1$ is small enough, and we have remarked just now that $\delta$ depends on $p_0$ and $f$.
  Therefore, $p_2$ only depends on $p_0$ and $f$, as the graph $y=y(x)$ determines the value of $p_2$.

  For the size of this AND gadget and the input size $\Lambda$, the size of this AND gadget is the size of a probability filter gadget plus $3\Lambda$ for those $u_i,v_i,w_i$, and the size of each of both input sets is $\Lambda$.
  Therefore, Lemma~\ref{filtersize} implies the second part of this lemma.
\end{proof}

\begin{remark}[Remark of Lemma~\ref{ANDlemmawithtwoinputs}]
Lemma~\ref{ANDlemmawithtwoinputs} shows that when constructing a $(2,\Lambda,p_0,p_2,\varepsilon_1,\varepsilon_2,f)$-AND gadget, we are free to set up the parameter $p_0$, and the parameter $p_2$ will be determined.
\emph{After $p_2$ is determined}, we are still free to choose $\varepsilon_1,\varepsilon_2$, and $\Lambda$ will be then determined.
In fact, the two parameters $\varepsilon_1,\varepsilon_2$ decides the number of layers needed in the probability filter gadget, and we can achieve (1) and (2) in Definition~\ref{defi:AND} for \emph{any} valid function $y(x)$ with two intersections to the line $y=x$ as it is in Figure~\ref{fig:filtercurve}.
That is the reason why we can choose $\varepsilon_1,\varepsilon_2$ after $p_2$ is determined.
In particular, for the same function $y(x)$ but different $\varepsilon_1,\varepsilon_2$, we just need the AND gadgets with different numbers of layers in their inner probability filter gadgets.
We will make use of this observation to construct AND gadgets with the same parameters $p_0,p_2,f$ but different $\varepsilon_1,\varepsilon_2$ in the next section.
\end{remark}

To conclude this section, we show that we can also construct a $(2,\Lambda,p_2-\varepsilon_2,p_2,\varepsilon_1,\varepsilon_2,f)$-AND gadget and a $(2,\Lambda,p^\ast(p_2-\varepsilon_2),p_2,\varepsilon_1,\varepsilon_2,f)$-AND gadget which will be used in the next section.
Notice that Lemma~\ref{ANDlemmawithtwoinputs} does not imply the possibility of constructing this AND gadget, as $p_2$'s existence is supposed to depend on the third parameter, which now become $p_2-\varepsilon_2$ and $p^\ast(p_2-\varepsilon_2)$, two constants related to $p_2$.

\begin{lemma}\label{ANDlemmap2}
  Given any 2-quasi-submodular influence function $f$ and any constant threshold $p_0>0$, we can construct a $(2,\Lambda,p_0,p_2,\varepsilon_1,\varepsilon_2,f)$-AND gadget, a $(2,\Lambda,p_2-\varepsilon_2,p_2,\varepsilon_1,\varepsilon_2,f)$-AND gadget and a $(2,\Lambda,p^\ast(p_2-\varepsilon_2),p_2,\varepsilon_1,\varepsilon_2,f)$-AND gadget with the same parameters $\Lambda,p_2,\varepsilon_1,\varepsilon_2$.
\end{lemma}
\begin{proof}
  The three AND gadgets are only different at the third parameter, which is the input threshold determining which of the two cases (1) and (2) in Definition~\ref{defi:AND} the inputs fall into.
  By our construction, we can use the same structure for the three AND gadgets, except that we use three different scaling down factors $\beta_1,\beta_2,\beta_3$ for the different thresholds $p_0$, $p_2-\varepsilon_2$ and $p^\ast(p_2-\varepsilon_2)$.
  In particular, the three probability filter gadgets inside the three AND gadgets can be exactly the same, provided that the ``gap'' $p_2/p_1$ is large enough such that
  \begin{itemize}
    \item $\left(\beta_1\varphi_T^-({p}_0),\beta_1\varphi_T^+({p}_0)\right)$ and $\beta_1\varphi_F^+({p}_0)$ are on the different sides of $p_1$,
    \item $\left(\beta_2\varphi_T^-(p_2-\varepsilon_2),\beta_2\varphi_T^+(p_2-\varepsilon_2)\right)$ and $\beta_2\varphi_F^+(p_2-\varepsilon_2)$ are on the different sides of $p_1$, and
    \item $\left(\beta_3\varphi_T^-(p^\ast(p_2-\varepsilon_2)),\beta_3\varphi_T^+(p^\ast(p_2-\varepsilon_2))\right)$ and $\beta_3\varphi_F^+(p^\ast(p_2-\varepsilon_2))$ are on the different sides of $p_1$.
  \end{itemize}
  We know that this is always possible by Proposition~\ref{lemmaratio}.

  As the same probability filter gadget is used in the two AND gadgets, the four parameters $\Lambda,p_2,\varepsilon_1,\varepsilon_2$, which are inherited from the probability filter gadget by our construction, are identical for the three AND gadgets.
\end{proof}

\subsection{Construction of the AND Gadget with General $I$ of an Integer Power of $2$}
\label{sect:ANDgadgets}
In this section, we construct the AND gadget in Definition~\ref{defi:AND} with general $I$ that is an integer power of $2$.

A $(I,\Lambda,p_0,p_2,\varepsilon_1,\varepsilon_2,f)$-AND gadget is a $(\log_2I)$-level AND circuit using 2-set-input AND gadgets constructed in the previous section as building block.
We will use three different types of 2-set-input AND gadgets.
\begin{itemize}
  \item Type $A$: $(2,\Lambda_0,p_0,p_2,\frac13(p_2-\varepsilon_2),\varepsilon_2,f)$-AND gadget.
  \item Type $B$: $(2,\Lambda_0,p_2-\varepsilon_2,p_2,\frac13(p_2-\varepsilon_2),\varepsilon_2,f)$-AND gadget.
  \item Type $C$: $(2,\Lambda_C,p_2-\varepsilon_2,p_2,\varepsilon_1,\varepsilon_2,f)$-AND gadget.
\end{itemize}
Lemma~\ref{ANDlemmap2} indicates that we can construct $A$ and $B$, and by Lemma~\ref{ANDlemmawithtwoinputs} $\Lambda_0$ is a constant since $\frac13(p_2-\varepsilon_2)$ is a constant.
By Lemma~\ref{ANDlemmawithtwoinputs} and its remark, we can construct $C$ based on $B$ by adjusting the number of layers in the inner probability filter gadget, and $\Lambda_C=O\left((1/\varepsilon_1)^{c_1}\right)$ for some constant $c_1$.

Figure~\ref{fig:AND} shows the construction of this AND gadget.
The type and the number of AND gadgets in each of the $\log_2I$ levels are set as follows:
\begin{itemize}
  \item Level $(\log_2I)$: A single AND gadget of Type $C$ is constructed.
  \item Level $(\log_2I-1)$: $2$ groups of $\Lambda_C$ Type $B$ AND gadgets are constructed, and the output vertices in each group are connected to each of the input ends $I_1,I_2$ of the AND gadget in Level~$(\log_2I)$.
  \item Level $(\log_2I-2)$: $2^2$ groups of $\Lambda_0\Lambda_C$ Type $B$ AND gadgets are constructed, and the output vertices in each group are connected to each of the input ends $I_1,I_2$ of the AND gadgets in each of the $2$ groups in Level~$(\log_2I-1)$.
  \item Level $(\log_2I-3)$: $2^3$ groups of $\Lambda_0^2\Lambda_C$ Type $B$ AND gadgets are constructed, and the output vertices in each group are connected to each of the input ends $I_1,I_2$ of the AND gadgets in each of the $2^2$ groups in Level~$(\log_2I-2)$.
  \item $\cdots$
  \item Level $2$: $2^{\log_2I-2}$ groups of $\Lambda_0^{\log_2I-3}\Lambda_C$ Type $B$ AND gadgets are constructed, and the output vertices in each group are connected to each of the input ends $I_1,I_2$ of the AND gadgets in each of the $2^{\log_2I-3}$ groups in Level~$3$.
  \item Level $1$: $2^{\log_2I-1}$ groups of $\Lambda_0^{\log_2I-2}\Lambda_C$ Type $A$ AND gadgets are constructed, and the output vertices in each group are connected to each of the input ends $I_1,I_2$ of the AND gadgets in each of the $2^{\log_2I-2}$ groups in Level~$2$.
\end{itemize}
Finally, the two input sets $I_1,I_2$ in each of the $2^{\log_2I-1}=\frac I2$ AND gadget groups in Level~$1$ form two of the $I$ input sets for the $(I,\Lambda,p_0,p_2,\varepsilon_1,\varepsilon_2,f)$-AND gadget we are constructing,
and the output vertex of the Type~$C$ AND gadget in Level~$(\log_2I)$ is the output of the $(I,\Lambda,p_0,p_2,\varepsilon_1,\varepsilon_2,f)$-AND gadget.

\begin{figure}
\centerline{\includegraphics[scale=0.6]{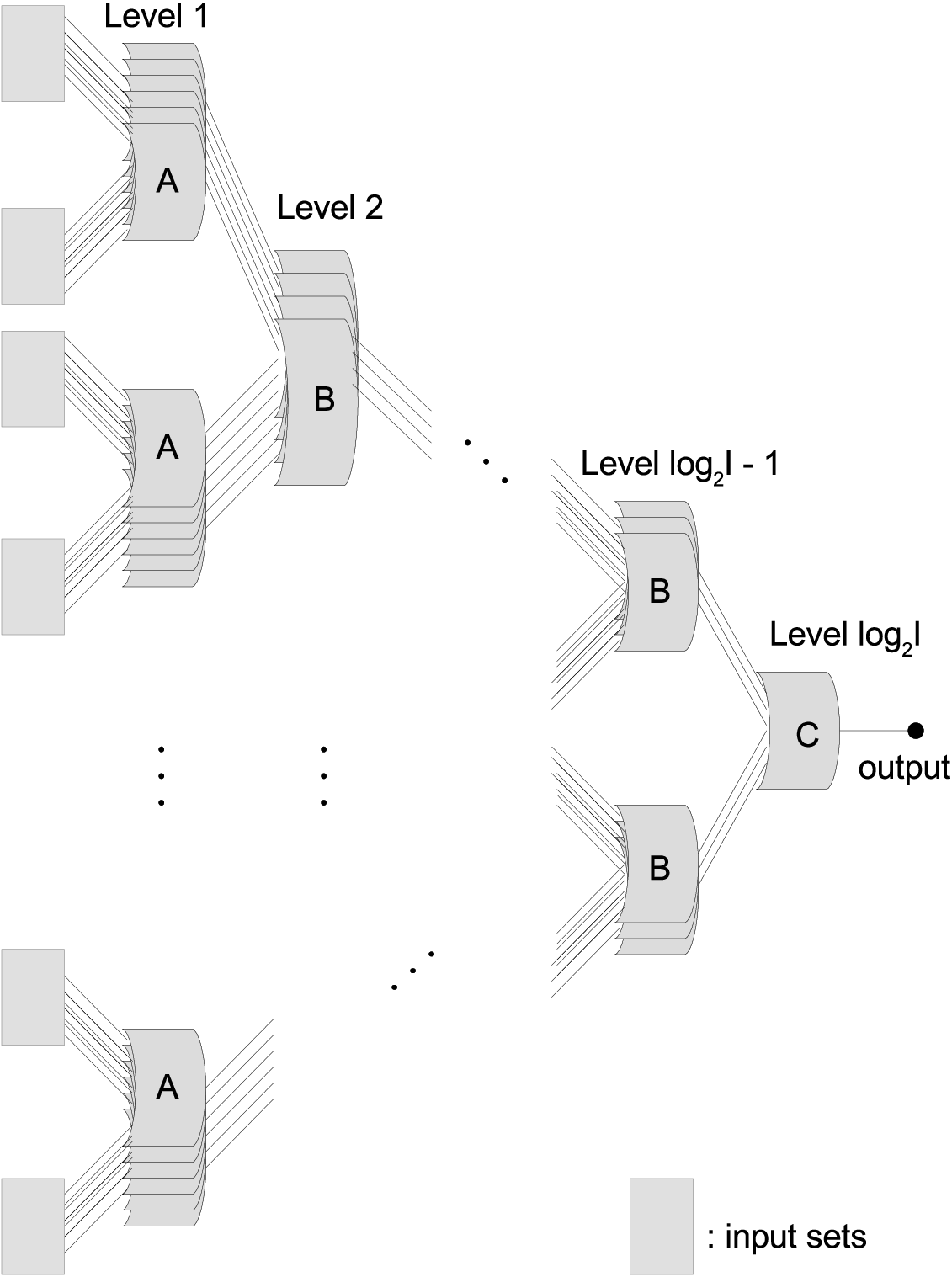}}
\caption{The $(I,\Lambda,p_0,p_2,\varepsilon_1,\varepsilon_2,f)$-AND gadget}
\label{fig:AND}
\end{figure}

We now show that (1) and (2) in Definition~\ref{defi:AND} hold.
\begin{enumerate}
  \item If all the vertices in all $I$ input sets are infected with independent probabilities less than $\frac{11}{10}p_0$, and the infection probabilities of the vertices in at least one set are less than $\frac12p_0$, then the Type~$A$ AND gadgets in at least one group in Level~$1$ will output vertices with infection probabilities less than $\frac13(p_2-\varepsilon_2)$.
      Since the threshold (the third parameter) of Type~$B$ AND gadgets is set to $(p_2-\varepsilon_2)$ and $\frac13(p_2-\varepsilon_2)<\frac12(p_2-\varepsilon_2)$, the Type~$B$ AND gadgets in at least one group in each of Level $2,3,\ldots,\log_2I-1$ will output vertices with infection probabilities less than $\frac13(p_2-\varepsilon_2)$.
      Finally, at least one of the two input sets for the Type~$C$ AND gadget in Level~$(\log_2I)$ will be infected with probabilities less than $\frac13(p_2-\varepsilon_2)$, which is less than $\frac12(p_2-\varepsilon_2)$.
      Thus, the output of the entire $(I,\Lambda,p_0,p_2,\varepsilon_1,\varepsilon_2,f)$-AND gadget is a vertex with infection probabilities less than $\varepsilon_1$, which implies (1) in Definition~\ref{defi:AND}.
  \item If all the vertices in all $I$ input sets are infected with independent probabilities in $(p_0,\frac{11}{10}p_0)$, all the Type~$A$ AND gadgets in Level~$1$ will output vertices with infection probabilities in $(p_2-\varepsilon_2,p_2]$.
      Since $(p_2-\varepsilon_2,p_2]\subseteq\left(p_2-\varepsilon_2,\frac{11}{10}(p_2-\varepsilon_2)\right)$ for small enough $\varepsilon_2$,\footnote{If the parameter $\varepsilon_2$ in the $(I,\Lambda,p_0,p_2,\varepsilon_1,\varepsilon_2,f)$-AND gadget we are constructing is not small enough to satisfy this, we can replace $\varepsilon_2$ with another smaller $\varepsilon_2'$ and instead construct a $(I,\Lambda,p_0,p_2,\varepsilon_1,\varepsilon_2',f)$-AND gadget. Notice that the description 2 of Definition~\ref{defi:AND} implies that a $(I,\Lambda,p_0,p_2,\varepsilon_1,\varepsilon_2',f)$-AND gadget is also a valid $(I,\Lambda,p_0,p_2,\varepsilon_1,\varepsilon_2,f)$-AND gadget for $\varepsilon_2'<\varepsilon_2$.} all the Type~$B$ AND gadgets in each of Level $2,3,\ldots,\log_2I-1$ will output vertices with infection probabilities in $(p_2-\varepsilon_2,p_2]$.
      Finally, the Type~$C$ AND gadget in Level~$(\log_2I)$ will output a vertex with infection probability in $(p_2-\varepsilon_2,p_2]$.
\end{enumerate}

Finally, we prove Lemma~\ref{ANDlemma} and Lemma~\ref{ANDlemma_p2} in Section~\ref{sect:proof1}.
\begin{proof}[Proof of Lemma~\ref{ANDlemma}]
  The existence of the $(I,\Lambda,p_0,p_2,\varepsilon_1,\varepsilon_2,f)$-AND gadget is proved by the explicit construction above.
  It remains to show that the number of vertices and edges in this AND gadget is $O\left((1/\varepsilon_1)^{c_1}I^{c_2+1}\right)$, and the input size is $\Lambda=O\left((1/\varepsilon_1)^{c_1}I^{c_2}\right)$.

  By Lemma~\ref{ANDlemmawithtwoinputs}, the number of vertices and edges in the Type $A$ and $B$ AND gadgets are constants, since the parameter $\frac13(p_2-\varepsilon_2)$ is a constant.
  Let $\chi$ be a constant upper bound for these.
  As for Type~$C$ AND gadget, it has $O\left((1/\varepsilon)^{c_1}\right)$ vertices and edges by Lemma~\ref{ANDlemmawithtwoinputs}.
  Since there are $2^{\log_2I-i}\Lambda_0^{\log_2I-i-1}\Lambda_C$ AND gadgets in Level~$i$ and $\Lambda_C=O\left((1/\varepsilon)^{c_1}\right)$ as mentioned, the total number of vertices and edges have the following bound.
  $$O\left((1/\varepsilon)^{c_1}\right)+\sum_{i=1}^{\log_2I-1}\chi\cdot2^{\log_2I-i}\Lambda_0^{\log_2I-i-1}\Lambda_C<\chi\Lambda_C\cdot(2\Lambda_0)^{\log_2I}=O\left((1/\varepsilon_1)^{c_1}I^{c_2+1}\right),$$
  where $c_2=\log_2\Lambda_0$ is a constant.

  As for $\Lambda$, there are $\Lambda_0^{\log_2I-2}\Lambda_C$ AND gadgets in each of the $\frac I2$ groups in Level~$1$, and each of these AND gadgets takes $\Lambda_0$ vertices as one of the two inputs.
  Therefore, we have
  $$\Lambda=\Lambda_0\cdot\Lambda_0^{\log_2I-2}\Lambda_C=O\left((1/\varepsilon_1)^{c_1}I^{c_2}\right),$$
  which concludes the last part of the lemma.
\end{proof}

\begin{proof}[Proof of Lemma~\ref{ANDlemma_p2}]
  Based on Lemma~\ref{ANDlemmap2}, by changing all the Type $A$ $(2,\Lambda_0,p_0,p_2,\frac13(p_2-\varepsilon_2),\varepsilon_2,f)$-AND gadgets in Level 1 to the Type $A'$ $(2,\Lambda_0,p^\ast(p_2-\varepsilon_2),p_2,\frac13(p_2-\varepsilon_2),\varepsilon_2,f)$-AND gadgets, we obtain an $(I,\Lambda,p^\ast(p_2-\varepsilon_2),p_2,\varepsilon_1,\varepsilon_2,f)$-AND gadget.

  The size of the AND gadget only changes by a constant, as the only difference between the two AND gadgets are the different probability scaling down gadgets used for different $\beta$ for $A$ and $A'$.
  Since the probability scaling down gadget has a constant size, we conclude the second half of the lemma.
\end{proof}

\subsection{Construction of Directed Edge Gadget}
\label{sect:directededgegadgets}
The $(\Upsilon,\epsilon,b,f)$-directed edge gadget in Definition~\ref{defi:directededge} can be constructed by modifying the number of layers in the $(\Lambda,p_1,p_2,\varepsilon_1,\varepsilon_2,f)$-probability filter gadget in Definition~\ref{defi:probabilityfiltergadget}.
While still keeping the parameter $h$ and $\alpha$ such that $a_1h\alpha=1-\delta$ in the probability separation block of the probability filter gadget, we modify the number of layers in the circuit to $L=\frac{\log(\Upsilon/\epsilon)}{\log(1/(1-\delta))}+1$.

To construct a directed edge gadget $\langle u,v\rangle$, we connect $u$ to all the $h^L$ inputs to the circuit, and let $v$ be the output.
The construction of directed edge gadget is shown in Figure~\ref{fig:directededgegadget}.

\begin{figure}
\centerline{\includegraphics{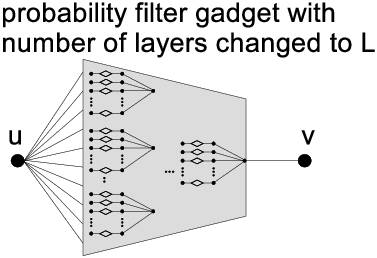}}
\caption{The directed edge gadget $\langle u,v\rangle$}
\label{fig:directededgegadget}
\end{figure}

To show property (1) in Definition~\ref{defi:directededge}, suppose $u$ is connected to $\Upsilon$ infected vertices $v_1,v_2,\ldots,v_{\Upsilon}$ by the directed edge gadgets.
If the vertices in the $i$-th layer are infected with probability $x_i$, then the vertices in the $(i-1)$-th layer will be infected with probability $x_{i-1}=a_1\alpha x_i$, which can be easily seen from Figure~\ref{fig:filter} and by noticing the symmetric property of probability scaling down gadgets mentioned in Remark~\ref{symmetricremark}.
Therefore, each vertices in the first level that are adjacent to $u$ will be infected with probability $(a_1\alpha)^L$.
Since there are $h^L$ vertices in the first level and $u$ is assumed to be connected to $\Upsilon$ vertices by the directed edge gadgets, the expected number of $u$'s infected neighbors is
$$\mathbb{E}[\text{num of infected neighbors}]=\Upsilon h^L(a_1\alpha)^L=\Upsilon(1-\delta)^L=\epsilon(1-\delta)<\epsilon,$$
where recall that we have set
$$L=\frac{\log(\frac{\Upsilon}\epsilon)}{\log\frac1{1-\delta}}+1.$$
Therefore, by Markov's inequality, the probability that $u$ has infected neighbor(s) is less than $\epsilon$, which means $u$ will be infected with probability less than $\epsilon$.

For (2), suppose $u$ is connected to $v$ by a directed edge gadget $\langle u,v\rangle$ and $u$ is already infected.
Then all the $h^L$ inputs of the inner probability filter gadget will be infected with probability $a_1$ independently, and $v$ will be infected with probability in $(p_2-\varepsilon_2,p_2]$ if $\delta$ is set small enough such that $a_1$ passes the threshold $p_1$.
In particular, $b>0$.

Lastly, we prove Lemma~\ref{directededgelemma1} and Lemma~\ref{directededgelemma3}.
\begin{proof}[Proof of Lemma~\ref{directededgelemma1}]
  The possibility of the construction is already made explicit.

  Let $\lambda$ be the upper bound of the number of vertices and edges in a probability separation block in the probability filter gadget (which is a constant), the total number of vertices in a directed edge gadget is
  $$\sum_{i=0}^{L-1}\lambda h^i=\lambda\frac{h^L-1}{h-1}=\Theta\left(h^L\right)=\Theta\left(h^{\frac{\log\Upsilon}{\log\frac1{1-\delta}}+\frac{\log(\frac1\epsilon)}{\log\frac1{1-\delta}}+1}\right)=\Theta\left(\Upsilon^d(1/\epsilon)^d\right),$$
  and the total number of edges is
  $$\underbrace{h^L}_{\text{number of edges from }u\text{ to the probability scaling down gadget}}+\sum_{i=0}^{L-1}\lambda h^i=\Theta\left(h^L\right)=\Theta\left(\Upsilon^d(1/\epsilon)^d\right).$$
  where $d=\frac{\log h}{\log\frac1{1-\delta}}$.

  To show that $d$ depends only on $f$, it is enough to notice that we only need to set up the values of $h$ and $\delta$ such that $p_1<a_1$ as mentioned.
\end{proof}

\begin{proof}[Proof of Lemma~\ref{directededgelemma3}]
  Given an $(I,\Lambda,p_0,p_2,\varepsilon_1,\varepsilon_2,f)$-AND gadget which consists of many 2-set-input AND gadgets (see Figure~\ref{fig:AND}),
  we can obtain a $(\Lambda,p_1,p_2,\varepsilon_1,\varepsilon_2,f)$-probability filter gadget which is the core of an arbitrary 2-set-input AND gadget.
  We construct the $(\Upsilon,\epsilon,b,f)$-directed edge gadget by increasing the number of layers in this probability filter gadget, just as what we did earlier.
  By our analysis above, we already have $b\in(p_2-\varepsilon_2,p_2]$.
  Moreover, by Figure~\ref{fig:filtercurve}, increasing the number of layers makes $b$ closer to $p_2$.
  Therefore, we can have $b\in\left(p_2-\frac12\varepsilon_2,p_2\right]$ by just increasing the number of layers, which proves the possibility of the construction.

  By our discussion in Section~\ref{subsect:pfg}, we only need a constant number of layers to have $b\in\left(p_2-\frac12\varepsilon_2,p_2\right]$, as $\frac12\varepsilon_2$ is a constant.
  Thus, requiring $b\in\left(p_2-\frac12\varepsilon_2,p_2\right]$ does not change the number of layers asymptotically.
  Following the proof of Lemma~\ref{directededgelemma1}, we conclude the second half of the lemma.
\end{proof}

\subsection{Proof of Theorem~\ref{hardnessproof} for $a_1=0$}
\label{sect:proof3}
In the case $a_1=0$, the constructions of both the AND gadget and the directed edge gadget fail.
Modifications of the structure in Figure~\ref{fig:highlevel} as well as the structure of the AND gadget are required.
We will discuss these modifications in this section, and the remaining details are left to the readers.

\paragraph{Modification to the AND Gadget}
The AND gadget for the case $a_1=0$ is much simpler.
The input $\varepsilon_1,\varepsilon_2$ is no longer needed, and both $p_0,p_2$ in the original AND gadget are set to $\frac12a_2$.
The definition of the modified AND gadget is shown below.
\begin{definition}\label{modifiedAND}
  A \emph{$(I,\Lambda,f)$-AND gadget} takes $I$ sets of $\Lambda$ vertices each as inputs, and output a vertex such that
\begin{enumerate}
  \item if the vertices in at least one input set are infected with probability $0$, then the output vertex will be infected with probability $0$;
  \item if the vertices in all input sets are infected with independent probability at least $\frac12a_2$, then the output vertex will be infected with probability at least $\frac12a_2$,
\end{enumerate}
\end{definition}
The construction of a $(2,\Lambda,f)$-AND gadget is shown in Figure~\ref{fig:modificationAND}.
It is easy to see that the infection of the output vertex will not affect any other vertices in this circuit due to $a_1=0$.
Due to the same reason, property (1) above is trivial for the case $I=2$ here.
Let $x$ be the probability that each vertex in the two input sets is infected, and let $y$ be the probability the output is infected.
Then,
$$y=\sum_{i=2}^\Lambda\binom{\Lambda}{i}a_i(a_2x)^i(1-a_2x)^{\Lambda-i}.$$
To satisfy (2), we only need to choose $\Lambda$ large enough such that $y(\frac12a_2)\geq\frac12a_2$.
This is always possible, as we have $y(\frac12a_2)\rightarrow p^\ast>\frac12a_2$ as $\Lambda\rightarrow\infty$ (the expected number of infected neighbors of the output vertex is $\frac12a_2 \Lambda$ which goes to infinity).
\begin{figure}
\centerline{\includegraphics{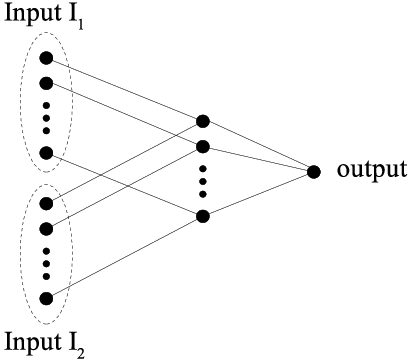}}
\caption{The modified AND gadget with parameter $(2,\Lambda,f)$}
\label{fig:modificationAND}
\end{figure}

\begin{lemma}\label{ANDlemma2inputwithzeroa1}
  For any $f$ with $a_2>a_1=0$, we can construct a $(2,\Lambda_0,f)$-AND gadget with constant size, and $\Lambda_0$ is a constant depending on $f$.
\end{lemma}
\begin{proof}
  The construction above shows the existence of the gadget, and $\Lambda_0$ is a constant that is large enough to make $y(\frac12a_2)\geq\frac12a_2$ true, which depends only on $f$.

  From Figure~\ref{fig:modificationAND}, it is clear that the gadget has $3\Lambda_0+1$ vertices and $3\Lambda_0$ edges, which are both constants.
\end{proof}

To construct a $(I,\Lambda,f)$-AND gadget with $I$ being an integer power of $2$, we use the same ``tower structure'' in Figure~\ref{fig:AND}.
Specifically, all the AND gadgets in all $\log_2I$ levels are identically the $(2,\Lambda_0,f)$-AND gadget in Figure~\ref{fig:modificationAND},
and the output vertices of $2^{\log_2I-i}$ groups of $\Lambda_0^{\log_2I-i}$ $(2,\Lambda_0,f)$-AND gadgets in Level~$i$ are connected to the input ends of $2^{\log_2I-i-1}$ groups of $\Lambda_0^{\log_2I-i-1}$ $(2,\Lambda_0,f)$-AND gadgets in Level~$(i+1)$.
It is straightforward to check that (1) and (2) in Definition~\ref{modifiedAND} hold for this construction.

\begin{lemma}
  For any $f$ with $a_2>a_1=0$ and any $I$ that is an integer power of $2$, we can construct a $(I,\Lambda,f)$-AND gadget with $O\left(I^{c+1}\right)$ vertices and $O\left(I^{c+1}\right)$ edges, and $\Lambda=I^c$, where $c$ is a constant depending on $f$.
\end{lemma}
\begin{proof}
The existence of this AND gadget is shown by the explicit construction.

The numbers of vertices and edges are both
$$\sum_{i=1}^{\log_2I}3\Lambda_0\cdot2^{\log_2I-i}\Lambda_0^{\log_2I-i}<3\Lambda_0\cdot(2\Lambda_0)^{\log_2I}=O\left(I^{c+1}\right),$$
where $c=\log_2\Lambda_0$ is a constant, and it depends only on $f$ as $\Lambda_0$ depends only on $f$ according to Lemma~\ref{ANDlemma2inputwithzeroa1}.
Notice that the number of vertices in a $(2,\Lambda_0,f)$-AND gadget is counted as $3\Lambda_0$ other than $3\Lambda_0+1$ in Lemma~\ref{ANDlemma2inputwithzeroa1}, because the output vertex of each $(2,\Lambda_0,f)$-AND gadget is counted as one of the input vertices in one of the $(2,\Lambda_0,f)$-AND gadgets in the next level.

Finally, since there are $\Lambda_0^{\log_2I-1}$ $(2,\Lambda_0,f)$-AND gadgets in each group in Level~$1$, we have
$$\Lambda=\Lambda_0\cdot\Lambda_0^{\log_2I-1}=I^c,$$
which concludes the lemma.
\end{proof}

\paragraph{Modification to the Set Cover Part}
We will use a pair of vertices to represent a subset in the \textsc{SetCover} problem and use a pair of cliques to represent an element in $U$.
The pair of vertices are connected to each vertex of the two cliques by a specially designed gadget shown in the bottom of Figure~\ref{fig:modificationsetcover}.

\begin{figure}
\centerline{\includegraphics[width=\textwidth]{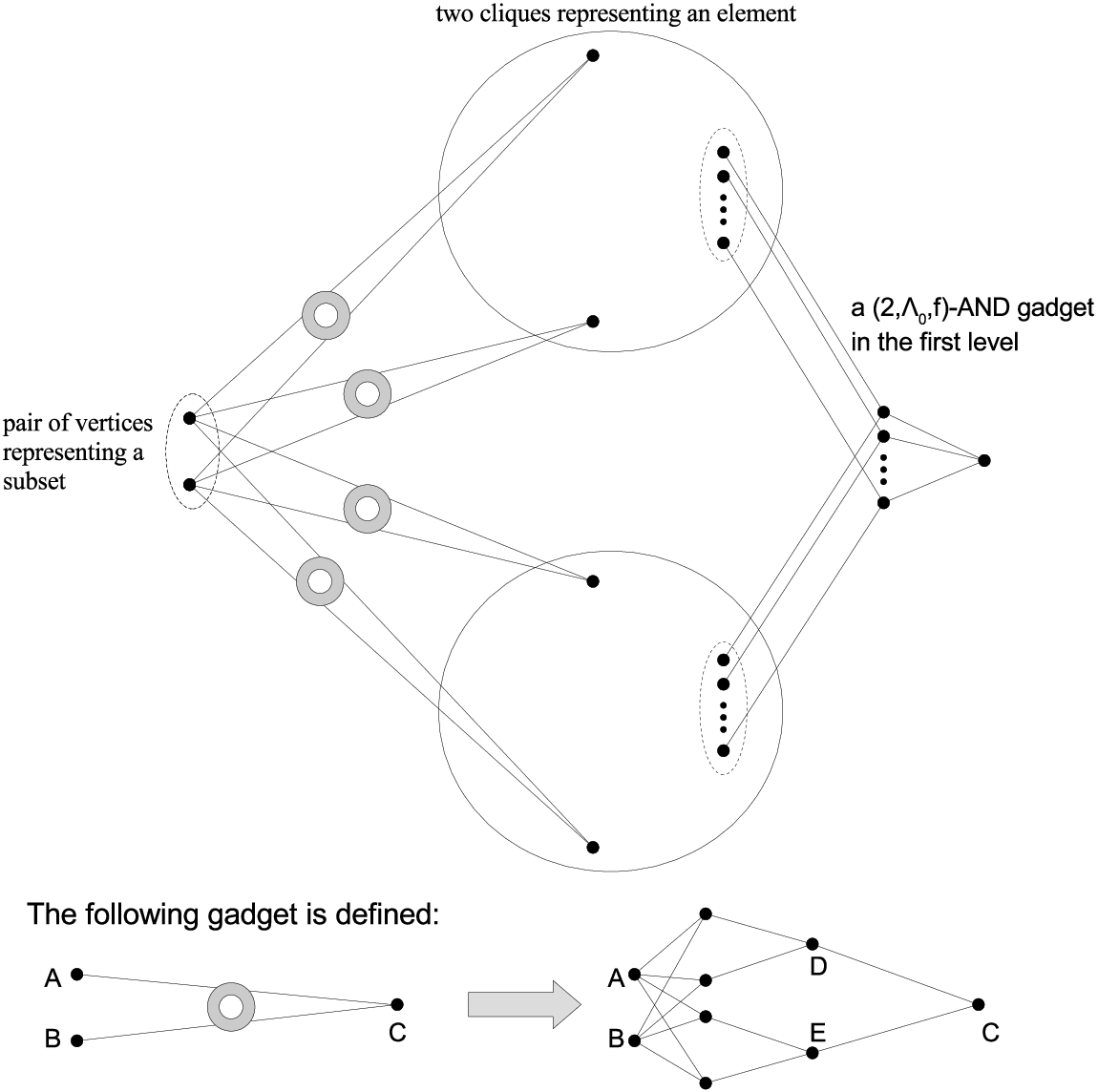}}
\caption{Connection between a pair of vertices representing a subset and vertices in the two cliques representing an element, and a $(2,\Lambda_0,f)$-AND gadget in the first level of the $\left(2n,(2n)^c,f\right)$-AND gadget.}
\label{fig:modificationsetcover}
\end{figure}

If the two vertices representing a subset are both infected, it is straightforward to check that each vertex at the output end of the gadget at the bottom of Figure~\ref{fig:modificationsetcover} will be infected with probability $a_2^7$.
Given there are $m$ vertices in a clique, the expected number of infected vertices in a clique is $a_2^7m$.
By choosing $m$ large enough (but still a constant) such that $a_{\lfloor a_2^7m\rfloor}>p^\ast-\varepsilon$, each vertex in the clique will be infected with probability at least $p^\ast-\varepsilon$.
Therefore, if a subset is picked such that the two vertices representing it are chosen as seeds, all pairs of cliques representing its elements will be activated.
Naturally, given the \textsc{SetCover} instance in which we are choosing $k$ subsets, we are asked to choose $2k$ seeds in the \textsc{InfluenceMaximization} instance.

On the other hand, since $a_1=0$, an activated clique will not be able to infect the pair of vertices representing a subset, so the connection between the pair of vertices to each vertex in the clique is like a directed edge.
Moreover, it is easy to see that we still need two seeds to pick a subset even if some cliques representing elements in this subset are activated.
Although we have the option to choose the two seeds ``on the gadget'', we still need to pick at least two seeds to ``choose a subset''.
Thus, it does not matter if any of these seeds is not exactly in the pair of vertices representing the subset.

The $M_2$ $(n,\Lambda,p^\ast-\varepsilon,p_2,1/n,\varepsilon_2,f)$-AND gadgets in Figure~\ref{fig:highlevel} is changed to $M_2$ $\left(2n,(2n)^c,f\right)$-AND gadgets here.
Moreover, each of the $n$ groups of the $(2,\Lambda_0,f)$-AND gadgets in Level~$1$ of the $\left(2n,(2n)^c,f\right)$-AND gadget corresponds to the vertices in \emph{the two cliques representing the same element} in $U$.
A single $(2,\Lambda_0,f)$-AND gadget is illustrated on the right-hand side of Figure~\ref{fig:modificationsetcover}.

\paragraph{Modification to the Connection to the $M_1$ Vertices}
In Figure~\ref{fig:highlevel}, the output vertex $v$ is connected to the $M_1$ vertices by $M_1$ edges.
Since $a_1=0$, such construction will fail to satisfy our purpose here.
To fix this, we can use $2M_2$ $\left(2n,(2n)^c,f\right)$-AND gadgets such that the outputs of every two AND gadgets are connected to each of the $M_1$ vertices.\footnote{Another way to fix this is to reduce the number of levels by $1$ in the $\left(2n,(2n)^c,f\right)$-AND gadget, such that we have two output vertices of the AND gadget instead of only one output in Definition~\ref{modifiedAND}.}

In addition, we also update the value of $M_1$ to $M_1=n^{c+10}$.

\paragraph{Modification to the Clique Size $m$}
Since there are $(2n)^c$ vertices in each of the $n$ inputs for each of the $2M_2$ $\left(2n,(2n)^c,f\right)$-AND gadgets,
to furnish enough inputs, we update the clique size to $m=2M_2\cdot(2n)^c=2^{1+c}n^{2+c}=O\left(n^{c+2}\right)$.

\paragraph{Modification to Lemma~\ref{YESNOcompare}}
To conclude this section, we have the following lemma corresponding to Lemma~\ref{YESNOcompare} in Section~\ref{sect:proof1}.
\begin{lemma}\label{YESNOcompare3}
  If the \textsc{SetCover} instance is a {\sf YES} instance, by choosing $2k$ seeds appropriately, we can infect at least $\frac14a_2^3\cdot n^{c+12}$ vertices in expectation in the graph $G$ we have constructed; if it is a {\sf NO} instance, we can infect at most $O\left(kn^{c+10}\right)$ vertices in expectation for any choice of $2k$ seeds.
\end{lemma}
\begin{proof}
  If the \textsc{SetCover} instance is a {\sf YES} instance, we choose the $2k$ seeds representing the $k$ subsets, and all the $2n$ cliques will be activated such that each vertex in all these clique will be infected with probability $p^\ast-\varepsilon$.
  Since $p^\ast\geq a_2$, we have $p^\ast-\varepsilon>\frac12a_2$ as $\varepsilon$ is sufficiently small due to large size of $m$.
  All the $2M_2$ $\left(2n,(2n)^c,f\right)$-AND gadgets fall into case (2), so that the output vertices are infected with probabilities at least $\frac12a_2$.
  The $M_1$ vertices in each of the $2M_2$ copies of the verification part are connected to two vertices with infection probabilities at least $\frac12a_2$, so the total expected number of infected vertices in $G$ is at least
  $$M_2\times \underbrace{\left(\frac12a_2\right)^2}_{\text{probability both output vertices are infected}}\times\underbrace{(a_2M_1)}_{\text{expected num of infections in }M_1\text{ vertices}}=\frac14a_2^3\cdot n^{c+12}.$$

  If the \textsc{SetCover} instance is a {\sf NO} instance, consider any choice of $2k$ seeds with $k_1$ seeds in the vertices representing subsets, $k_2$ seeds in the connection gadgets between vertices representing subsets and vertices in the cliques, $k_3$ seeds in the $2n$ cliques, $k_4$ seeds in those $(2,\Lambda_0,f)$-AND gadgets at Level~$1$ of the $\left(2n,(2n)^c,f\right)$-AND gadgets, and $k_5=2k-k_1-k_2-k_3-k_4$ seeds in the remaining part of the verification parts (the $(2,\Lambda_0,f)$-AND gadgets at the remaining levels and the $M_1$ vertices connecting to the $\left(2n,(2n)^c,f\right)$-AND gadgets).
  Again, we first prove that at least one clique will not be activated such that all its vertices are infected with probability $0$.

  First of all, those $k_5$ seeds cannot have effect in activating cliques.
  This is because their influence cannot pass through the $(2,\Lambda_0,f)$-AND gadgets in the first level, as the infection of the output vertex in each $(2,\Lambda_0,f)$-AND gadget cannot further infect the input vertices due to $a_1=0$.

  Secondly, for those $k_1$ and $k_2$ seeds, they are at the vertex-pairs representing the subsets and the gadgets connected to those pairs respectively.
  We call the vertices in those gadgets connecting to a pair \emph{the vertices around the pair}.
  It is easy to see that we need to choose at least $2$ seeds in or around a pair to pick a subset.
  To see this, even if vertex $C$ in the gadget (at the bottom of Figure~\ref{fig:modificationsetcover}) is already infected (which is possible as $C$ belongs to a clique which may have been activated already) such that $D$ and $E$ already have one infected neighbor, we still cannot make both $A$ and $B$ infected by picking only $1$ seed in or around the pair $(A,B)$.
  Thus, we assume without loss of generality that all $k_2$ seeds are on the pairs representing the subsets, as we need at least $2$ seeds in or around a pair $(A,B)$ in which case we can assume the seeds are just at $A$ and $B$.

  For those $k_3$ seeds on the cliques and $k_4$ seeds on the $(2,\Lambda_0,f)$-AND gadgets in the first level, since each AND gadget in the first level takes two sets of vertices from two cliques \emph{representing the same element in $U$}, we need at least $3$ seeds to activate two cliques representing the same element in $U$: one in the middle of the AND gadget, and one in each of the two cliques (such that the two vertices connecting to the seed in the middle of the AND gadget have two infected neighbors, and stand a chance to activate the two cliques).
  In contrast, we only need $2$ seeds to activate these two cliques, by choosing the pair of vertices representing the subset covering the element that these two cliques represent.
  Therefore, we can assume that those $k_3$ and $k_4$ seeds are also on those pairs representing subsets.

  Since $k_1+k_2+k_3+k_4\leq2k$ and the \textsc{SetCover} instance is a {\sf NO} instance, by the fact that we need $2$ seeds to pick a subset, we conclude that at least one clique will not be activated, and the vertices in this clique are infected with probability $0$.

  By the effect of the $\left(2n,(2n)^c,f\right)$-AND gadget, except for those (at most) $k_4+k_5$ AND gadgets containing seeds, the output vertices of the remaining $2M_2-k_4-k_5$ AND gadgets will be infected with probability $0$, which have no effect on those $M_1$ vertices.
  Therefore, even if all the vertices in the set cover part, the $k_4+k_5$ copies of the verification parts, and the $2M_2$ $\left(2n,(2n)^c,f\right)$-AND gadgets are infected, the total number of infected vertices cannot exceed
  $$\underbrace{2K +6K (2n)m+2nm}_{\text{size of the set cover part}}+2M_2\underbrace{(2n)^{c+1}}_{\text{size of an AND gadget}}+(k_4+k_5)\underbrace{\left((2n)^{c+1}+M_1\right)}_{\text{size of a verification part}}=O\left(kn^{c+10}\right),$$
  which concludes the lemma.
\end{proof}

Noticing that the total number of vertices in $G$ is
$$N=2K +6K (2n)m+2nm+M_2\left((2n)^{c+1}+M_1\right)=\Theta\left(n^{c+12}\right),$$
and $kn^{c+10}=O(n^{c+11})$.
we conclude Theorem~\ref{hardnessproof} in the case $a_1=0$ by setting $\tau=\frac1{c+12}$.

\section{Conclusion}
We show the hardness of approximating \infmax in several settings restricting the network structure or the cascade model.
Before our results there was some hope that the hardness of nonsubmodular influence maximization was only caused by the hardness of detecting community structure within the network.
However, our results show that even for very plain community structures, \infmax can remain hard.
Moreover, in our construction, even if the algorithm is told the community structure, the problem remains hard.
We show that it is the bidirectional nature of contagions which renders the problem hard.

We also show the inapproximability of \infmax even in the restrictive universal local influence model (Definition~\ref{defi:ulim}) with any 2-quasi-submodular local influence function $f$, even if $f$ is almost submodular.
Since it turns out assumptions on either the graph topology or the cascade model do not really make \infmax easy, a natural question is that, what if we make assumptions on both?

We conclude with the following open problem: considering the universal local influence model in Definition~\ref{defi:ulim} with 2-quasi-submodular $f$ on the stochastic hierarchical blockmodel\footnote{It does not make much sense to consider universal local influence model $I_f^G$ on the hierarchical blockmodel, as $f$ depends only on the \emph{number} of infected neighbors and ignores the weights of edges connecting to the neighbors, while the hierarchical blockmodel considers weighted graphs.},
does there exist a 2-quasi-submodular $f$, such that \infmax is NP-hard to approximate to within a constant factor?
Or is it the case that for any 2-quasi-submodular $f$, there exists a constant factor approximation to \infmax?

\section*{Acknowledgement}
We would like to thank the anonymous reviewers for their helpful and constructive comments.

\bibliographystyle{plainnat}
\bibliography{reference}

\newpage
\appendix
\section{A Variant of Theorem~\ref{hardnessproof}}
\label{sect:hardness_2submod_compare}
\citet{li2017influence} considered a model where there is only a sublinear fraction of vertices admitting nonsubmodular local influence functions that are almost submodular.
They showed that, even though this appears to make the cascade model globally closer to submodularity, \infmax is still NP-hard to approximate to within $N^\tau$ for certain constant $\tau$.
In this section, we adapt Theorem~\ref{hardnessproof} to a variant that is of a similar style of this.

\begin{theorem}\label{hardnessproof_adapted}
  Consider the \infmax problem $(G,F,\D,k)$.
  For any fixed 2-quasi-submodular $f$, any fixed submodular function $g:\mathbb{Z}_{\geq0}\to[0,1]$ with $g(1)>g(0)=0$, and any $\gamma\in(0,1)$, there exists a constant $\tau$ depending on $f$ and $\gamma$ such that, even if $\D$ is the uniform distribution on $[0,1]$, $f_v\in F$ is symmetric with either $f_v=f$ or $f_v=g$, and $|\{v\in V:f_v=f\}|\leq N^\gamma$, it is NP-hard to distinguish between the following two cases:
  \begin{itemize}
      \item \yes: there exists a seed set $S$ with $|S|=k$ such that $\sigma_{F,\D}^G(S)=\Theta(N)$;
      \item \no: for any seed set $S$ with $|S|=k$, we have $\sigma_{F,\D}^G(S)=O(N^{1-\tau})$.
  \end{itemize}
\end{theorem}
\begin{proof}
We again discuss two different cases: $a_1=f(1)>0$ and $a_1=f(1)=0$.

For the first case, the reduction in Sect.~\ref{sect:proof2} can be modified to prove this theorem (if we only need to prove this theorem for directed graphs, the much simpler reduction in Sect.~\ref{sect:proof1} can be used), with the following modifications.
\begin{itemize}
    \item Except for those $M_1$ vertices on the right-hand side of Fig.~\ref{fig:highlevel} in each of the $M_2$ copies of the verification part, the remaining vertices are equipped with $f$. Those $M_1$ vertices in each of the $M_2$ copies are equipped with $g$.
    \item Change $M_1=n^{(30c_1+30c_2+70)d}$ (as it is in Sect.~\ref{sect:proof2}) to $M_1=n^{\frac1\gamma(30c_1+30c_2+70)d}$, where $n$ is the number of elements in the \SC instance and $c_1,c_2,d$ are the constants in Lemma~\ref{ANDlemma} and Lemma~\ref{directededgelemma1}.
\end{itemize}

Recall from Sect.~\ref{sect:proof2} that the set cover part has $O(n^{(3c_1+3c_2+7)d+c_1+c_2+4})$ vertices and the AND gadget has $O(n^{c_1+c_2+1})$ vertices, the total number of vertices in $G$ is
$$N=O\left(n^{(3c_1+3c_2+7)d+c_1+c_2+4}\right)+M_2\left(O\left(n^{c_1+c_2+1}\right)+M_1\right)=\Theta\left(n^{\frac1\gamma(30c_1+30c_2+70)d+2}\right),$$
which is of polynomial size.
Moreover, the total number of vertices that are equipped with $f$ is
$$O\left(n^{(3c_1+3c_2+7)d+c_1+c_2+4}\right)+M_2O\left(n^{c_1+c_2+1}\right)=o\left(n^{(30c_1+30c_2+70)d}\right)\ll N^\gamma.$$

The remaining part of the proof is almost identical to the proof of Lemma~\ref{YESNOcompare2}.
The only difference is that, if the output of the AND gadget, the vertex $v$ in Fig.~\ref{fig:highlevel}, is infected, then each of those $M_1$ vertices is now infected with probability $g(1)$, instead of $a_1=f(1)$ before.
With this change, when the \SC instance is a \yes instance, the total number of infected vertices (for properly choosing seeds corresponding to the subsets) become
$$p_{\text{activated}}g(1)(p_2-\varepsilon_2)M_1M_2=\Theta\left(n^{\frac1\gamma(30c_1+30c_2+70)d+2}\right)=\Theta(N),$$
where $p_{\text{activated}}$ is the same as it is in the proof of Lemma~\ref{YESNOcompare2}, $p_2,\varepsilon_2$ are the parameters for the AND gadget which are the same as defined in Sect.~\ref{sect:proof2}.
When the \SC instance is a \no instance, following the same analysis, the upper bound for the total number of infected vertices can also be computed by Equation~(\ref{eqn:yesnocompare}), with $M_1$ replaced by the modified value $n^{\frac1\gamma(30c_1+30c_2+70)d}$ here.
In particular, the first three terms in (\ref{eqn:yesnocompare}) are dominated, the fourth and the fifth terms are both at most $O(n^{\frac1\gamma(30c_1+30c_2+70)d+1})$.
Therefore, we conclude the theorem for the case $a_1>0$ by setting $\tau=\frac{1}{\frac1\gamma(30c_1+30c_2+70)d+2}$.

For the second case $a_1=f(1)=0$, the reduction is almost the same as it is in Sect.~\ref{sect:proof3}, except for the following changes.
\begin{itemize}
    \item Those $M_1$ vertices in each of the $M_2$ copies are equipped with $g$, while the remaining vertices are equipped with $f$.
    \item Change the value of $M_1$ from $n^{c+10}$ (as it is in Sect.~\ref{sect:proof3}) to $n^{\frac1\gamma(c+10)}$.
\end{itemize}
Following the same analysis in Sect.~\ref{sect:proof3}, we can see that the graph has $N=n^{\frac1\gamma(c+10)+2}$ vertices, and there are only $O(n^{c+4})\ll N^\gamma$ vertices that have $f$ as their local influence functions.
Corresponding to Lemma~\ref{YESNOcompare3}, we can show that the expected number of infections is at least $\frac14a_2^2g(2)n^{\frac1\gamma(c+10)+2}=\Theta(N)$ when appropriately choosing $2k$ seeds for a \yes instance, and the expected number of infections can be at most $O(kn^{\frac1\gamma(c+10)})$ for a \no instance.
By noticing $kn^{\frac1\gamma(c+10)}=O(n^{\frac1\gamma(c+10)+1})$ and taking $\tau=\frac{1}{\frac1\gamma(c+10)+2}$, we conclude the theorem for the case $a_1=0$.
\end{proof}

We remark that Theorem~\ref{hardnessproof_adapted} can be viewed as a generalization of the inapproximability result in~\cite{li2017influence} in the following two directions.
\begin{enumerate}
    \item Our result holds for any $f$ that is fixed in advance, while $f$ is set to $f(1)=\frac{1-\varepsilon}2$ and $f(2)=1$ in~\cite{li2017influence} (where $\varepsilon$ is an arbitrary constant fixed in advance).
    \item Our result holds for undirected graphs, while it is unknown if the proof in~\cite{li2017influence} can be adapted to show the same inapproximability result for undirected graphs (notice that an undirected graph can be viewed as a special case of a directed graph with anti-parallel edges, so an inapproximability result for a more special case is stronger).
\end{enumerate}

\end{document}